\documentclass[twoside]{article}

\usepackage{graphicx}
\usepackage{amsmath,mathtools,amsthm}
\usepackage{enumitem}
\graphicspath{{figs/}}
\usepackage{wrapfig}
\usepackage{algorithm}
\usepackage[noend]{algorithmic}
\usepackage{caption}
\usepackage{subcaption}
\usepackage{comment}
\usepackage[round]{natbib}
\usepackage{xr}
\usepackage{amssymb}
\usepackage[hidelinks]{hyperref}
\usepackage{url}
\makeatletter
\newcommand*{\addFileDependency}[1]{
  \typeout{(#1)}
  \@addtofilelist{#1}
  \IfFileExists{#1}{}{\typeout{No file #1.}}
}
\makeatother
\newtheorem{prop}{Proposition}
\usepackage{booktabs}

\usepackage{booktabs}
\newcommand\ccell[1]{\multicolumn{1}{c}{#1}}

\usepackage{mathtools}
\mathtoolsset{showonlyrefs}

\newcommand{\indep}{\perp \!\!\! \perp}

\usepackage{dsfont}
\usepackage[dvipsnames, table]{xcolor}
\usepackage{longtable}
\usepackage{multirow}
\usepackage{pifont}
\usepackage{boldline}

\usepackage{makecell}
\usepackage{etoolbox}
\usepackage{collcell, xfp}
\usepackage{spreadtab}
\newcommand{\pvalf}[1]{%
\ifcase\pdfstrcmp{p-value}{#1}
#1
\or
\ifcase\pdfstrcmp{HRT}{#1}
#1
\or
\ifcase\pdfstrcmp{CRT}{#1}
#1
\or
  \ifnum\fpeval{#1 <= 0.05} = 1
   \textcolor{Green}{$#1$}
   \else 
    \ifnum\fpeval{#1 > 0.05} = 1
        \textcolor{red}{$#1$}%
    \else
        \textcolor{black}{#1}%
   \fi
  \fi
  \fi
  \fi
  \fi
}

\newcommand{\evalf}[1]{
\ifcase\pdfstrcmp{test martingale}{#1}
#1
\or
  \ifnum\fpeval{#1 > 20} = 1
   \textcolor{Green}{$#1$}%
  \else 
    \ifnum\fpeval{#1 > 10} = 1
        \textcolor{Emerald}{$#1$}%
    \else 
     \ifnum\fpeval{#1 > 3.16} = 1
        \textcolor{cyan}{$#1$}%
     \else 
      \ifnum\fpeval{#1 > 1} = 1
        \textcolor{YellowOrange}{$#1$}%
      \else 
        \ifnum\fpeval{#1 <= 1} = 1
            \textcolor{red}{$#1$}
        \else
            \textcolor{black}{#1}%
        \fi
     \fi
    \fi
   \fi
  \fi
  \fi
}

\newcommand{\neff}[1]{%
\ifcase\pdfstrcmp{stopping time}{#1}
#1
\or
  \ifnum\fpeval{#1 = 1555} = 1
   \textcolor{Gray}{#1}
  \else
    \textcolor{black}{#1}%
  \fi
  \fi
}

\newcommand{\neffstock}[1]{%
\ifcase\pdfstrcmp{stopping time}{#1}
#1
\or
  \ifnum\fpeval{#1 = 2421} = 1
   \textcolor{Gray}{$#1$}
  \else
    \textcolor{black}{#1}%
  \fi
  \fi
}

\newcommand{\sectortech}[1]{%
  \ifstrequal{#1}{Information Technology}{\textbf{#1}}{#1}
}
\newtheorem{theorem}{Theorem}
\newtheorem{lemma}{Lemma}

\newtheorem{definition}{Definition}%

\usepackage[accepted]{aistats2023}
\begin{document}

%
\runningtitle{Model-X Sequential Testing for Conditional Independence via Testing by Betting}

%
\runningauthor{Shalev Shaer, Gal Maman, Yaniv Romano}

\twocolumn[

\aistatstitle{Model-X Sequential Testing for Conditional Independence \\ via Testing by Betting}

\aistatsauthor{ Shalev Shaer$^\ast$$^{,1}$ \And Gal Maman$^\ast$$^{,1}$ \And  Yaniv Romano$^1$$^{,2}$}
\aistatsaddress{$^{1}$Department of Electrical and Computer Engineering, Technion--Israel Institute of Technology \\
$^{2}$Department of Computer Science, Technion--Israel Institute of Technology}]

\begin{abstract}
This paper develops a model-free sequential test for conditional independence. The proposed test allows researchers to analyze an incoming i.i.d. data stream with any arbitrary dependency structure, and safely conclude whether a feature is conditionally associated with the response under study. We allow the processing of data points online, as soon as they arrive, and stop data acquisition once significant results are detected, rigorously controlling the type-I error rate. Our test can work with any sophisticated machine learning algorithm to enhance data efficiency to the extent possible. The developed method is inspired by two statistical frameworks. The first is the model-X conditional randomization test, a test for conditional independence that is valid in offline settings where the sample size is fixed in advance. The second is testing by betting, a ``game-theoretic'' approach for sequential hypothesis testing. We conduct synthetic experiments to demonstrate the advantage of our test over out-of-the-box sequential tests that account for the multiplicity of tests in the time horizon, and demonstrate the practicality of our proposal by applying it to real-world tasks.
\end{abstract}

\section{INTRODUCTION}
\label{sec:intro}

A central problem in data analysis is to rigorously find conditional associations in complex data sets with nonlinear dependencies. This problem lies at the heart of causal discovery \citep{pearl2000models,peters2017elements}, variable selection \citep{barber2015controlling,candes2018panning},
machine learning interpretability \citep{burns2020interpreting,lu2018deeppink}, economics \citep{angrist2011causal,wang2018characteristic}, and genetics studies \citep{sesia2019gene,bates2020causal}, to name a few. In such applications, the data are often collected online,
and, naturally, researchers are interested in analyzing the data points immediately after they are observed so that further data acquisition can be terminated as soon as significant results are detected. This experimental setting, for example, is typical in decision-making \citep{nikolakopoulou2018continuously,bhui2019testing} and clinical trials \citep{park2018critical}, where the need for additional samples to obtain accurate statistical inference must frequently be balanced with experimental costs.

To formalize the problem, suppose we are given a stream of data points $(X_t,Y_t,Z_t)$ for $t\in\mathbb{N} =1,2,\dots$, where each triplet contains a response $Y_t \in \mathbb{R}$, a feature $X_t \in \mathbb{R}$, and a vector of covariates $Z_t \in \mathbb{R}^d$. We assume the observations are sampled i.i.d. from $P_{YXZ}=P_{Y \mid XZ}P_{XZ}$, where $P_{Y \mid XZ}$ is unknown. Given such an online data stream, our goal is to test for \emph{conditional independence} (CI), where the null hypothesis is given by
\begin{equation}
\label{eq:hypo}
    H_0: X_t \indep Y_t \mid Z_t \ \ \text{for all} \ t \in \mathbb{N}.
\end{equation}
In words, we say that $H_0$ is true if $X_t$ is independent of the response $Y_t$ after accounting for the effect of the covariates $Z_t$, \emph{simultaneously for all time steps} $t$. We refer to $X_t$ that satisfies $H_0$ as an `unimportant' feature. Analogously, the alternative  hypothesis implies that $X_t$ carries new information on the response $Y_t$ beyond what is already contained in $Z_t$, i.e., $X_t \not \indep Y_t \mid Z_t$. Therefore, we say that such a feature $X_t$ is `important'. 

The goal of sequential hypothesis testing is to formulate a concrete decision rule on whether we can confidently reject the null at each time step $t$, by monitoring and accumulating the evidence collected at each step against the null using past data $\{(X_{s},Y_{s},Z_{s})\}_{s=1}^t$ \citep{wald1945sequential}.
This allows the analyst great flexibility, as she can decide, at each step, whether new data should be collected to support the question under study. Key to this setting is the need to provide the analyst with a tool that rigorously controls the type-I error rate---defined as the probability of rejecting the null when it is in fact true---at any given desired level $\alpha$, simultaneously for all time steps $t$. This requirement should not be confused with the premise of classic offline tests for CI that attain type-I error rate control \emph{only when the sample size is fixed in advance}. We refer to these as offline tests, emphasizing that one cannot naively monitor the outcome of a classic test---a p-value---and reject the null at an optional time step $t$ without accounting for the multiplicity of the tests across the time horizon; this strategy would result in inflation of the type-I error rate.  
Beyond online type-I error rate control, ideally, we wish to have a powerful test that would reject the null when it is false, and we want it to do so as early as possible.

\subsection*{Our contribution}

In this paper, we develop a novel sequential test for CI. Our proposal takes inspiration from two powerful and attractive statistical tools that are gaining increasing attention in recent years. The first is the model-X conditional randomization test (CRT) by \cite{candes2018panning}, an \emph{offline} test for CI.
The second is testing by betting \citep{shafer2019book,grunwald2020safe}, a ``game-theoretic'' approach for sequential hypothesis testing, where our proposal is very much inspired by the line of work reported in \citep{lady,ramdas2022testing}. The method we introduce in this paper, presented in Section~\ref{sec:proposed}, generalizes the offline CRT to the challenging online setting, resulting in a new test with the following features.

\textbf{Safe testing with early stopping}: building on recent advances in sequential testing using e-values and martingales, detailed in Section~\ref{sec:testing_by_betting}, the proposed CI test is guaranteed to control the type-I error rate at any time step. In particular, the analyst is allowed to track the outcome of the test over time, and safely reject the null if it exceeds a user-defined significance level, preventing a wasteful collection of unnecessary new data points.

\textbf{Model-X setting}: similar to the offline CRT method, described in Section \ref{sec:model_x}, the online test we propose does not make any assumptions on the conditional distribution of $Y \mid X, Z$. For instance, we do not make unrealistic assumptions that the relationship between $Y$ and $(X,Z)$ is linear, or that $Y \mid X,Z$ is Gaussian. However, this advantage comes at the cost of assuming that the distribution of $X \mid Z$ is known. This assumption is common to all tests belonging to the family of model-X knockoffs, including the CRT, and it is manageable when (i) large unlabeled data are available in contrast to labeled data, or (ii) when we have good prior knowledge about the distribution of $X \mid Z$ \citep{candes2018panning,sesia2019gene,romano2020deep}. We discuss this in more detail in Section~\ref{sec:practical}.

\textbf{Online learning from past experience}: the proposed test can leverage any machine learning algorithm to powerfully discover violations of the CI null. In particular, when a new triplet $(Y_t, X_t, Z_t)$ is observed, we use online learning techniques to efficiently update the running predictive model, instead of fitting a new model from scratch. This way, the whole data stream is used for training in a computationally efficient manner. 
The proposed framework also falls under the umbrella of interactive tests \citep{lei2018adapt,lei2021general,ibet}, providing the analyst the liberty to look at past data and decide how to modify the learning algorithm at any time step---e.g., to switch to a model that is more robust to outliers---to better discriminate the null and alternative hypotheses when applied to future test points.
    
\textbf{Optimized software package}: we provide a python code that implements our testing framework, is available at \url{https://github.com/shaersh/ecrt}.
The package includes important design principles: an automatic hyper-parameter tuning that does not require fitting the machine learning model from scratch (Supplementary Section~\ref{supp:online}); an ensemble procedure for improving the power of the test by averaging multiple martingales (Section~\ref{sec:practical}); and a de-randomization procedure that also improves power by reducing inherent algorithmic randomness due to a sampling mechanism that is necessary to formulate the test (Section~\ref{sec:practical}).

\section{MODEL-X CI TESTING}
\label{sec:model_x}

The CRT, developed by \cite{candes2018panning}, is an \emph{offline} test for CI that we build upon in this work. A key advantage of the CRT is that it assumes nothing on the conditional distributions of $Y \mid X,Z$ and $Y\mid Z$. This test, however, assumes that the conditional distribution of $X \mid Z$ is known. The CRT procedure, described in Algorithm~\ref{alg:CRT} in Supplementary Section~\ref{supp:crt}, resembles classic permutation tests and has two key components: a test statistic function $T(\cdot)$ and a function that samples  dummy features $\tilde{X}$ from $P_{X \mid Z}$. Since $\tilde{X}$ is sampled without looking at $Y$, the dummy triplets $(\tilde{X},Y,Z)$ satisfy $\tilde{X} \indep Y \mid Z$ by construction. Hence, by comparing the test statistic evaluated on the original $\{(X_i,Y_i,Z_i)\}_{i=1}^n$ and dummy $\{(\tilde{X}_i, Y_i, Z_i)\}\}_{i=1}^n$ triplets, the CRT generates a valid p-value $p_n$, controlling the CI null at level $\alpha$ when the sample size $n$ is fixed in advance \citep{candes2018panning}, i.e., 
\begin{align}\label{eq:p-value}
    \mathbb{P}[p_n \leq \alpha \mid \text{the null is true}] \leq \alpha \ \text{for a fixed} \ n.
\end{align}
Put differently, when all $n$ observations are available before testing, one can use $p_n$ to rigorously control the type-I error. However, future observations cannot be utilized to generate a new p-value (e.g., in cases where the null is not rejected) without a proper correction that ensures the validity of the sequential test. To see this, suppose for simplicity that under the null $p_n \sim \text{Uniform}(0,1)$ is distributed uniformly over the $[0,1]$ interval for any fixed $n$, satisfying~\eqref{eq:p-value}. Next, let $\tau$ be a data-dependent stopping time, given by
\begin{align}
    \tau = \{\min n : p_{n} \leq \alpha, n \in \mathbb{N} \}.
\end{align}
Now, observe that with this choice of stopping time, $\mathbb{P}[p_\tau \leq \alpha \mid \text{the null is true}]$ cannot be bounded by $\alpha$ anymore: there exists $\tau$ such that a rejection rule $p_\tau \leq \alpha$ would result in an invalid $\alpha$-level test.

In many applications, however, one is interested in applying the test online to obtain reliable data-driven conclusions as soon as possible.
This motivates us to adopt a fresh statistical approach for hypothesis testing, called testing by betting, briefly described in the next section.
\section{TESTING BY BETTING}
\label{sec:testing_by_betting}

Before diving into the mathematical principles of the testing by betting approach, we follow \cite{shafer2019book} and \cite{shafer2021testing} and present an intuitive interpretation of this framework. Imagine we are playing a game, in which we start with initial toy money. At each time step, we place a bet against the null hypothesis, and then reality reveals the truth. If this bet turns out to be correct, our wealth is increased by the money we risk in the bet; otherwise, we lose and the wealth is decreased accordingly. If our wealth at time $t$ is at least $1/\alpha$ times as large as the initial toy money we started with (e.g., we have managed to multiply our initial money by a factor of $1/0.05 = 20$ for $\alpha=0.05$) we can confidently reject the null, knowing that the type-I error is guaranteed to be controlled at level~$\alpha$. A property important to the formulation of the above game is this: if the null is true, the game must be fair in the sense that it is unlikely we will be able to significantly increase our initial toy money, no matter how sophisticated our betting strategy is. 

A mathematical object that is crucial to formalize a fair game is a \textit{test martingale}, defined below.
\begin{definition}
\label{def:test_martingale}
A random process $\{S_t : {t \in \mathbb{N}_0}\}$ is a test martingale for a given null hypothesis $H_0$ if it satisfies the following conditions: (i) $S_0=1$, (ii)  $S_t \geq 0,\ \forall t \in \mathbb{N}_0$, and (iii) $\{S_t : {t \in \mathbb{N}_0}\}$ is a supermartingale under $H_0$.
\end{definition}
\noindent In the view of testing by betting, the initial value of the test martingale $S_0$ represents the initial toy money in the game, and $S_t$ corresponds to our wealth at time $t$. Now, suppose we are handed a valid test martingale $\{S_t : {t \in \mathbb{N}_0}\}$, and let $\tau \geq 1$ be a data-dependent optional stopping time. By invoking the \textit{optional stopping theorem} we get 
\begin{equation}
    \label{eq:stopping}
    \mathbb{E}_{H_0}[S_{\tau}] \leq  \mathbb{E}_{H_0}[S_0]=1,
\end{equation}
\noindent meaning that $S_\tau$ is a non-negative random variable whose expected value is bounded by one for any stopping time $\tau \geq 1$. In the literature on testing by betting, $S_\tau$ is often referred to as an \textit{e-value} \citep{evalues,ebh,grunwald2020safe}. Importantly, the consequence of~\eqref{eq:stopping} is that, under the null, the game is fair since the expected value of our wealth $S_t$ at any time step $t$ is bounded by the initial toy money $S_0$. Moreover, since $\{S_t : {t \in \mathbb{N}_0}\}$ is a non-negative supermartingale under $H_0$, we can apply  Ville's inequality~\citep{ville} and get
\begin{equation}
    \label{eq:ville}
    \mathbb{P}_{H_0}(\exists t\geq 1 : S_t \geq 1/\alpha) \leq \alpha \mathbb{E}_{H_0}[S_0] = \alpha,
\end{equation}
for any $\alpha \in (0,1)$. Therefore, the ability to form a valid test martingale allows us to rigorously test for $H_0$ and reject the null if $S_t \geq 1/\alpha$ at any time step, with the premise that the type-I error would not exceed the level $\alpha$. Crucially, when the null is false, $S_t$ can largely grow depending on how successful our betting strategy is. In Section~\ref{sec:proposed} we formulate a valid test martingale and design a powerful betting strategy.


\paragraph{Related work.} Sequential testing has a long standing history~\citep{wald1945sequential,lai1984incorporating,naghshvar2013active,lheritier2018sequential}, where the sequential probability ratio test of \cite{wald1945sequential} is perhaps
one of the first sequential hypothesis tests. More recently, the \emph{testing by betting} methodology \citep{shafer2019book,shafer2021testing} has led to the design of new powerful nonparametric approaches for constructing confidence sequences, e.g., \citep{jun2019parameter,waudby2020estimating}, for testing a single hypothesis, as well as for testing multiple hypotheses; see \citep[Section 6]{waudby2020estimating} for a detailed summary.

Related methods to our proposal are offline and online two-sample tests that are based on martingales \citep{balsubramani2015sequential,two_sample_grunwald,two_sample_aaditya,ibet}. Specifically, \cite{two_sample_aaditya} studied the problem of designing martingale-based sequential nonparametric one- and two-sample tests that are consistent, i.e., these sequential tests can attain \emph{power one} under certain conditions. In our work, we build on the foundations of~\cite{two_sample_aaditya}, and extend this framework to CI testing. 
Recently, \cite{ren2022derandomized} suggested using e-values to de-randomize the outcome of the knockoff filter---a sister method to the CRT that focuses on false discovery rate control (FDR) in an offline setting. In our work, we aggregate e-values to de-randomize our test, where the e-values we define take a different form than those proposed by \cite{ren2022derandomized}, as we focus on sequential testing of a single feature. Lastly, independent work by \cite{grunwald2022anytime}, which has been developed and posted in parallel to ours, also offers a martingale-based sequential test under the model-X setting, although suggesting a different test martingale. In Supplementary Section \ref{supp:grunwald} we provide a more detailed discussion about the relation of our proposal to that of \cite{grunwald2022anytime}, along with empirical comparisons.


\section{THE PROPOSED e-CRT}
\label{sec:proposed}
In this section, we introduce e-CRT: a sequential test for CI based on martingales and e-values. Suppose we are given a machine learning model $\hat{f}_t$, fitted on an initial batch of labeled data $\{(X_s,Y_s,Z_s) : s \leq t-1\}$ to provide an estimate of $Y$ given $(X, Z)$. At a high level, the test is initialized with toy money $S_0=1$ and proceeds as follows.
\begin{enumerate}[leftmargin=*,nolistsep]
    \item \textbf{Collect} a fresh test triplet $(X_t, Y_t, Z_t)$.
    \item \textbf{Generate a dummy feature} $\tilde{X}_t \sim P_{X\mid Z}(X_t \mid Z_t)$, and form the dummy triplet $(\tilde{X}_t, Y_t, Z_t)$.
    \item \textbf{Compute a betting score} $W_t$. Use $\hat{f}_t$ to bet against the null, where the bet is that the prediction error of $\hat{f}_t$ (or any other test statistic), evaluated on the dummy triplet $(\tilde{X}_t, Y_t, Z_t)$, would be higher than that of the original triplet $(X_t, Y_t, Z_t)$. A positive (resp. negative) score indicates that our bet is successful (resp. unsuccessful).
    \item \textbf{Update the current wealth (test martingale)} $S_{t}$: 
    if the betting score is positive, the previous $S_{t-1}$ is increased by the money we risked on placing the bet; otherwise, the previous wealth $S_{t-1}$ is decreased analogously.
    \item \label{item:preview_model} \textbf{Update the predictive model} $\hat{f}_t$ and get $\hat{f}_{t+1}$, e.g., by taking one (or more) gradient steps to minimize a loss evaluated on $\{(X_s,Y_s,Z_s) : s \leq t\}$.
    \item If $S_t \geq 1/\alpha$ reject $H_0$ and stop. Otherwise, increase $t$ and return to step (1). 
\end{enumerate}

\begin{figure}[t]
\begin{centering}
\includegraphics[width=0.9\columnwidth]{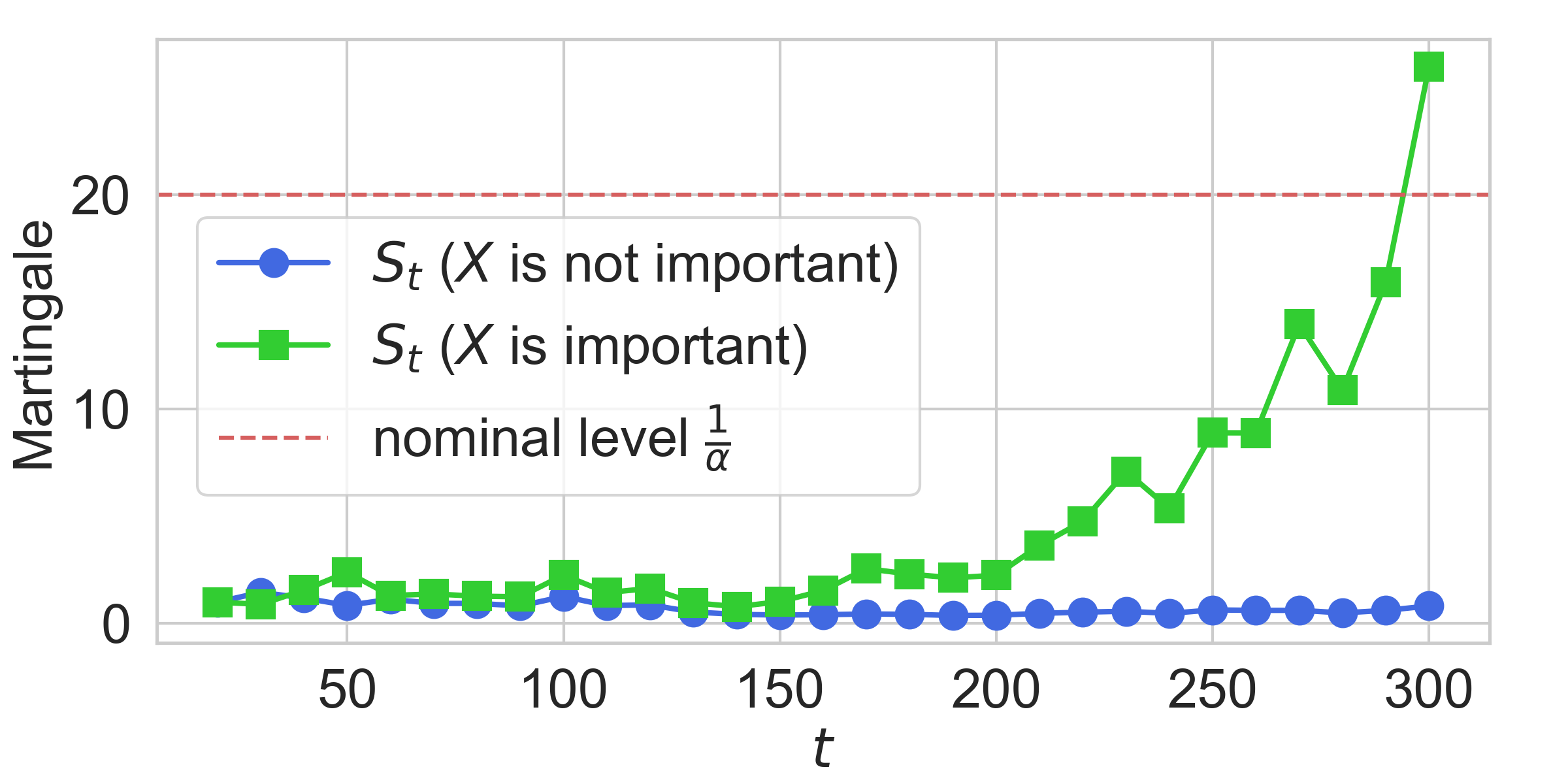}
\par\end{centering}
\vspace{-2mm}
\caption{\textbf{Illustration of the test martingale (wealth) $S_t$ as a function of $t$}. The blue (resp. green) curve represents the test martingale evaluated on simulated \texttt{null data} (resp. \texttt{non-null data}).
\label{fig:fig_size_legend}}
\vspace{-0.2cm}
\end{figure}

In what follows, we describe each of the above components in depth, define the proposed test martingale, and prove its validity. Later, in Section~\ref{sec:practical}, we provide additional design principles that improve the power of the test.

Before doing so, we pause to provide a small synthetic experiment that showcases how the wealth process $S_t$ behaves under the null and the alternative. To this end, we generate two different data sets. The first satisfies $H_0$, which we refer to as \texttt{null data} in which $X$ is \emph{un}important. The second satisfies the alternative, which we call \texttt{non-null data} in which $X$ is important. The data generation process for each case and the implementation details are described in Section~\ref{sec:exp_setup}. Next, we apply e-CRT on each data set, and present in Figure~\ref{fig:fig_size_legend} the wealth process $S_t$ as a function of $t$. When the test is applied to the \texttt{null data}, the value of $S_t$ remains close to the initial wealth $S_0=1$ for all presented time steps $t$. In particular, $S_t$ does not exceed the value $1/\alpha=20$, and thus $H_0$ cannot be rejected. By contrast, when the test is applied to the \texttt{non-null data}, the wealth process grows as the testing procedure proceeds, until reaching a target value of $1/\alpha = 20$. In this case, we reject the null and report that $X$ is indeed important. This experiment illustrates the advantage of monitoring the value of $S_t$ over time: we can safely terminate the test after collecting 300 samples and avoid a wasteful collection of new data.

\subsection{Formulating the Test Martingale}
 \label{sec:formulating}

Our procedure exploits the dummy feature $\tilde{X}_t$, sampled from the conditional distribution of $X \mid Z$ to form a fair game. In the sequel, we state key properties of the dummies, which we will use to define our test. The proofs of all statements given in this section are provided in Supplementary Section~\ref{supp:proofs}. We start by emphasizing that we sample $\tilde{X}_t \sim P_{X\mid Z}(X_t \mid Z_t)$ without looking at $Y_t$, and so $\tilde{X}_t \indep X_t \mid Z_t$ for all $t\in\mathbb{N}$ by construction. Therefore, $X_t$ and its dummy $\tilde{X}_t$ are exchangeable conditional on $Z_t$; that is, $(X_t,\tilde{X}_t) \mid Z_t \overset{d}{=} (\tilde{X}_t,X_t) \mid Z_t$, where $\overset{d}{=}$ reads as `equal in distribution'. This implies that it is impossible to distinguish between $X_t$ and its dummy $\tilde{X}_t$ when viewing $Z_t$, for any time step $t$. Furthermore, under the null, this exchangeability property holds not only conditionally on $Z_t$ but also on $Y_t$.
\begin{lemma}
Take $(X_t,Y_t,Z_t) \sim P_{XYZ}$, and let $\tilde{X}_t$ be drawn independently from $P_{X \mid Z}$ without looking at $Y_t$. If $Y_t \indep X_t \mid Z_t$, then
$(X_t,\tilde{X}_t,Y_t,Z_t) \overset{d}{=} (\tilde{X}_t,X_t,Y_t,Z_t)$.
\label{lemma:swap_ZY}
\end{lemma}
\noindent
The above result lies at the heart of the knockoff and CRT frameworks, and its proof follows \citep[Lemma 3.2]{candes2018panning}, \citep[Lemma 1]{barber2020robust}. Lemma~\ref{lemma:swap_ZY} implies that, if the null is true, it is impossible to tell which is the original feature and which is the dummy when viewing the full observation, at any time step $t$. This result is essential for proving the validity of the CRT p-value, as well as for formulating our test martingale, as we do next.

Denote by $\mathcal{F}_{t} = \sigma(\{X_s,Y_s,Z_s\}_{s=1}^{t})$ the sigma-algebra generated by observations collected up to time $t$, where $\mathcal{F}_0$ is the trivial sigma-algebra. Let $q_t = T(X_t,Y_t,Z_t;\hat{f}_t) \in \mathbb{R}$ and $\tilde{q}_t = T(\tilde{X}_t,Y_t,Z_t;\hat{f}_t) \in \mathbb{R}$
be the test statistics evaluated on the original and dummy triplets, respectively. Importantly, $T(\cdot;\hat{f}_t)$ can be any function, and its choice may affect the power of the test. For instance, one can define $T(\cdot;\hat{f}_t)$ as the squared prediction error evaluated on the current sample $T(x,y,z;\hat{f}_t) = (\hat{f}_t(x,z) - y)^2$ using a model $\hat{f}_t$ trained on past data $\{X_s,Y_s,Z_s\}_{s=1}^{t-1}$. Observe that $\hat{f}_t$ is not fitted on the new triplet $(X_t,Y_t,Z_t)$, thus it is considered as a fixed function once conditioning on $\mathcal{F}_{t-1}$.

We then proceed by evaluating a \textbf{betting score} 
\begin{equation}
\label{eq:lambda_t}
 W_t = g(q_t,\tilde{q}_t),
\end{equation}
where the function $g(a,b) \in [-m,m]$ is antisymmetric $g(a,b) = - g(b,a)$, satisfying $g(a,b) > 0$ if $b>a$ and $g(a,b_1) \geq  g(a,b_2)$ for $b_1 \geq b_2$. For example, $g(a,b)=m \cdot \textrm{sign}(b - a)$.
The hyper-parameter $0 < m \leq 1$ controls the magnitude of the score. As in the knockoff filter, our design of $g$ ensures it follows the \emph{flip sign property}, requiring that a swap of the original feature $X_t$ and its dummy $\tilde{X}_t$ will flip the sign of $W_t$~\citep{candes2018panning}.  

Under the alternative, one should interpret a strictly positive betting score $W_t>0$ as a successful bet, which will increase our wealth. This means that we gain some evidence that ${X}_t$ carries extra predictive power about $Y_t$ beyond what is already known in $Z_t$. Analogously, a strictly negative $W_t<0$ indicates an erroneous bet, which will reduce our wealth even though the null is false. 
Crucially, under the null, $W_t$ will be zero on average, no matter how accurate $\hat{f}_t$ is. In other words, it is impossible to have a systematically positive $W_t$ when $H_0$ is true.
\begin{lemma}
\label{lemma:R_bet}
Under the same conditions as in Lemma~\ref{lemma:swap_ZY}, if $H_0$ is true then $\mathbb{E}_{H_0}[W_t\mid \mathcal{F}_{t-1}] = 0$ for all $t \in \mathbb{N}$.
\end{lemma}
\noindent The core idea behind the proof of the above lemma is that, under the null, $W_t$
has a symmetric distribution about zero conditional on $\mathcal{F}_{t-1}$, and thereby its expected value is zero; see \citep{ramdas2020admissible} for a related property of symmetric distributions. In particular, $W_t$ is equally likely to have positive and negative values, which is a well-known result in the knockoff literature with the important difference that in our case we show it holds conditionally on $\mathcal{F}_{t-1}$.

Armed with the betting score $W_t$ at time $t$, we turn to define a test martingale $\{S_t : {t \in \mathbb{N}_0}\}$ for $H_0$. The martingale can be thought of as the wealth process, initialized by toy money $S_0=1$, and our ultimate goal is to maximize it. We begin with defining the \emph{base martingale} as follows:
\begin{equation}
    \label{eq:test_martingale_v}
    S_t^v :=\prod_{j=1}^t (1 + v \cdot W_j),
\end{equation}
where $v\in[0,1]$ is a fixed amount of toy money that we are willing to risk at step $t$.\footnote{We can set a different $v_t$ for each time step, yet $v_t$ must be chosen without looking at the current $(X_t,Y_t,Z_t)$ as otherwise the test will cease to be valid. Intuitively, in such a case one can always set $v_t=0$ when $W_t$ is negative and $v_t=1$ otherwise, and increase the wealth regardless on whether the null is true or false.}
Proposition~\ref{prop:valid_stv} in Supplementary Section~\ref{supp:stv} shows that $\{S_t^v : {t \in \mathbb{N}_0}\}$ in~\eqref{eq:test_martingale_v} is a valid test martingale.
As a result, following Ville's inequality in~\eqref{eq:ville}, one can monitor $S_{t}^v$ and control the type-I error for any choice of $v\in[0,1]$. Importantly, the amount of toy money $v$ that we risk when placing the bet affects the power. 

The above immediately raises the question of how should we choose $v$? Ideally, we want to set the best constant $v^*$ so that $S_t^{v^*}$ is maximized under the alternative. The problem is that we are not allowed to look at the current betting score $W_t$, so it is impossible to find such an ideal data-dependent $v^*$ in foresight. As a thought experiment, consider the simplest choice for $g$ as the sign function for which $W_t\in\{+1,-1\}$, and suppose we adopt an aggressive betting strategy with $v~=~1$. With this choice, when we win a bet we will increase $S_t^v$ by the maximal amount possible at step $t$. However, if we lose a bet even once, we will have $S_t^v~=~0$, resulting in a powerless test; to see this, assign $W_t = -1$ in~\eqref{eq:test_martingale_v}. We give a concrete example that visualizes this discussion in Supplementary Section~\ref{supp:imp_fig2}.


As a way out, we formulate a powerful betting strategy using the mixture-method of~\cite{two_sample_aaditya}, which is intimately connected to universal portfolio optimization~\citep{cover2011universal}. The mixture-method is defined as the average over $S_t^v$ for all $v \in [0,1]$:
\begin{equation}
    \label{eq:test_martingale}
    S_t = \int_0^1 S_t^v \cdot h(v) dv,
\end{equation}
where $h(v)$ is a probability density function (pdf) whose support is on the $[0,1]$ interval, e.g., a uniform distribution. We adopt the mixture method betting strategy to formulate our test martingale since it has appealing power properties, which we discuss soon. Before doing so, however, we shall first prove that the test martingale in~\eqref{eq:test_martingale} is valid. The theorem presented below states that by monitoring $S_t$ one can safely reject the null the first time $S_t$ exceeds $1/\alpha$, while rigorously controlling the type-I error simultaneously for all optional stopping times. This result holds in finite samples, without making any modeling assumptions on the conditional distribution of $Y \mid X$, and for any machine learning model $\hat{f}_t$, which we use to bet against the null. 
The proof follows \cite[Section~2.2]{two_sample_aaditya}.
\begin{theorem}
\label{th:validity}
    Under the same conditions as in Lemma~\ref{lemma:swap_ZY}, if the null hypothesis $H_0$ is true then for any $\alpha \in (0,1)$,
    $$ \mathbb{P}_{H_0}(\exists t : S_t \geq 1/\alpha) \leq \alpha.$$
\end{theorem}
Having established the validity of the test, we turn to discuss the key advantage of the mixture method betting strategy. The idea behind this approach is that one of the base martingales $S_t^v$ in~\eqref{eq:test_martingale} must hit the best constant $v^*$, which, in turn, drives the average martingale $S_t$ upwards. We demonstrate this visually in Supplementary Section~\ref{supp:imp_fig2}.
In fact, \cite{two_sample_aaditya} proved that $S_t$
is not only dominated by $S_t^{v^*}$, but can also provably form a consistent test that achieves \emph{power one} in the limit of infinite data. 

\begin{prop} [\cite{two_sample_aaditya}]
\label{prop:power_1}
If $\liminf_{t \to \infty}\frac{1}{t}\sum_{s=1}^t W_s > 0$ under the alternative $H_1$. Then, $ \mathbb{P}_{H_1}(\exists t : S_t \geq 1/\alpha)=1 $ for any $\alpha \in (0,1)$. 
\end{prop}

The condition of \mbox{$\liminf_{t \to \infty}\frac{1}{t}\sum_{s=1}^t W_s > 0$} implies that it suffices that only on average the predictive model will be able to tell apart the original and dummy triplets, so at the limit of infinite data we will attain a consistent test.

\subsection{Practical Considerations}
\label{sec:practical}

In this section, we provide design principles that improve the power of the test while maintaining its validity. For ease of reference, we provide a pseudo code that implements the following ideas in Algorithm~\ref{alg:practical_eCRT}.

\textbf{De-randomization.} To reduce the algorithmic randomness induced by the generated dummy feature $\tilde{X}_t$, we (i) sample $K>1$ independent dummy copies of $X_t$; (ii) compute the corresponding betting scores $W_t^{(k)}$, $k=1,\dots,K$; and (iii) evaluate the average betting score $W_t = \frac{1}{K}\sum_{k=1}^K W_t^{(k)}$. We refer to $K$ as the de-randomization hyper-parameter. Importantly, this strategy preserves the validity of the test martingale $S_t$ in~\eqref{eq:test_martingale}, since the expected value of the average betting score is also equal to zero under the null, i.e., $\frac{1}{K}\sum_{k=1}^K \mathbb{E}_{H_0}[W_t^{(k)}\mid \mathcal{F}_{t-1}] = 0$.
The ablation study presented in Section~\ref{sec:ablation} demonstrates that the above de-randomization procedure improves the power of the test, even for a moderate choice of $K$.

\begin{algorithm}[t]
\textbf{Input}: Data batch $\{(X_s,Y_s,Z_s)\}_{s=1}^{b}$; conditional distribution ${P}_{X \mid Z}$; test statistic $T(\cdot)$; fixed predictive model $\hat{f}$; de-randomization parameter $K$; betting score function $g(\cdot)$.

\begin{algorithmic}[1]
\STATE Compute $q \leftarrow T(\{({X}_s,Y_s,Z_s)\}_{s=1}^{b};\hat{f})$
\FOR{$k=1,\dots,K$}
\STATE  Sample $\tilde{X}_s \sim {P}_{X \mid Z}(X_s \mid Z_s)$ for $s=1,\dots,b$
\STATE Compute $\tilde{q} \leftarrow T(\{(\tilde{X}_s,Y_s,Z_s)\}_{s=1}^{b};\hat{f})$
\STATE Compute a betting score $W^{(k)} \leftarrow g(q,\tilde{q})$
\ENDFOR
\end{algorithmic}
\textbf{Output} An average betting score $W \leftarrow \frac{1}{K}\sum_{k=1}^K W^{(k)}$
\vspace{0.1cm}
\caption{Betting score evaluation}
\label{alg:payoff}
\end{algorithm}

\begin{algorithm}[t]
\textbf{Input}: Data stream $(X_t,Y_t,Z_t), t=1,\dots$; test level $\alpha \in (0,1)$; set of batch sizes $\mathcal{B}$.

\begin{algorithmic}[1]
\STATE Train a predictive model $\hat{f}_{1}$ on an initial batch of data $\{(X_{t'},Y_{t'},Z_{t'})\}_{t'=1}^{n_\text{init}}$
\STATE Set $S_{0,b}\leftarrow 1$ for all $b \in \mathcal{B}$
\FOR{$t=1,2,\dots$}
\FOR{$b$ in $\mathcal{B}$}
\STATE Set $t' \leftarrow \lfloor{t/b}\rfloor$
\STATE Set $\mathcal{D}_{t,b} = \{(X_j,Y_j,Z_j)\}_{j=(t'-1) \cdot b+1}^{t' \cdot b}$
\STATE {Compute the average betting score $W_{t',b}$ by applying Algorithm~\ref{alg:payoff} to $\mathcal{D}_{t,b}$ with $\hat{f}_{(t'-1)\cdot b+1}$}
\STATE Update $S_{t,b} \leftarrow \int_0^1 \prod_{s=1}^{t'} (1 + v \cdot W_{s,b}) \cdot h(v) dv$
\ENDFOR
\STATE Compute the ensemble-over-batches martingale $S_t \leftarrow \frac{1}{|\mathcal{B}|} \sum_{b \in \mathcal{B}} S_{t,b}$
\IF{$S_{t} \geq 1/\alpha$}
\STATE Reject the null hypothesis $H_0$ and stop
\ELSE
\STATE Obtain $\hat{f}_{t+1}$ by online updating $\hat{f}_{t}$ after adding the most recent point $(X_t,Y_t,Z_t)$ to the train set
\ENDIF
\ENDFOR
\end{algorithmic}

\vspace{0.05cm}
\caption{e-CRT: sequential test for CI}
\label{alg:practical_eCRT}
\end{algorithm}

\begin{figure*}
     \centering
     \begin{subfigure}[b]{0.33\textwidth}
         \centering
         \includegraphics[width=0.99\textwidth]{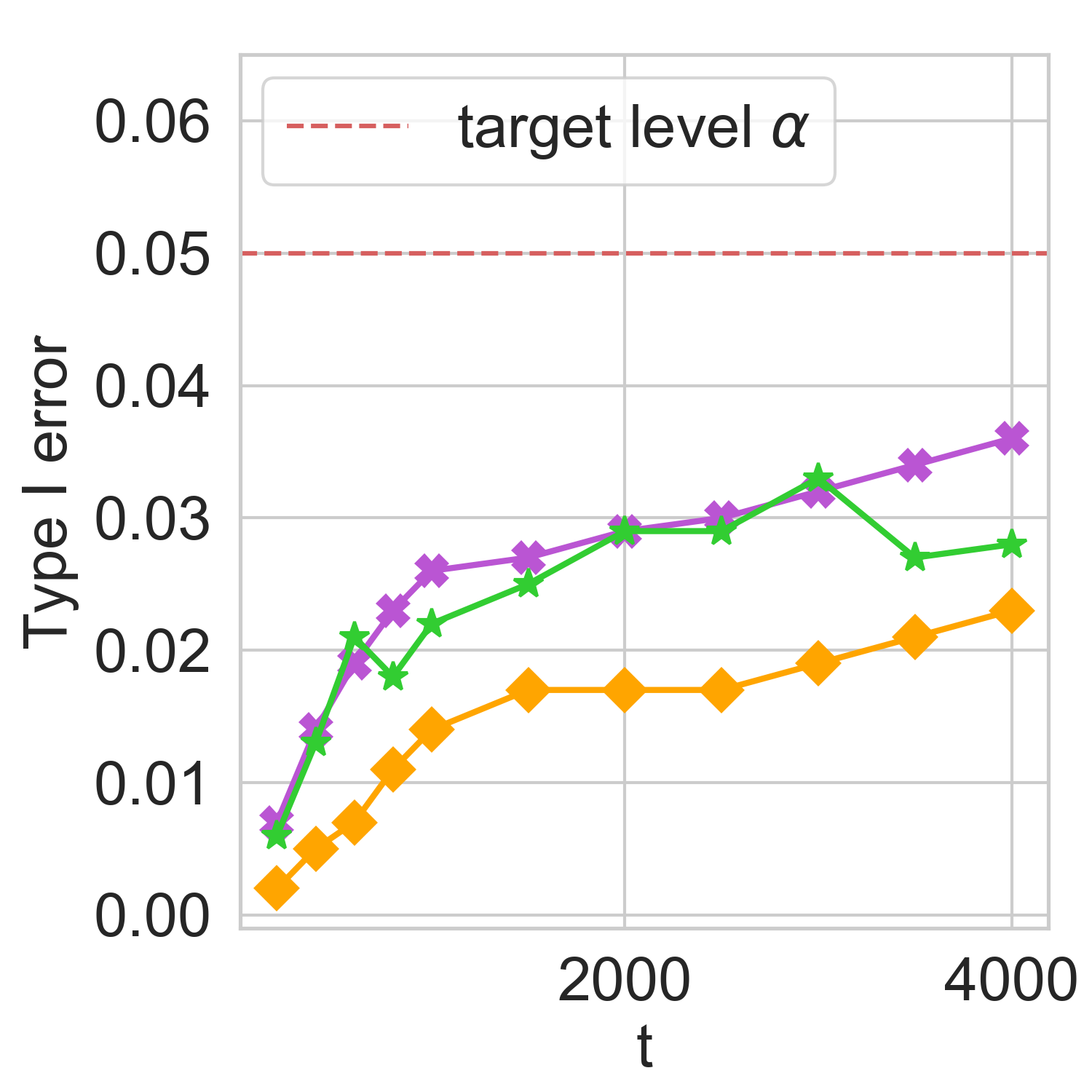}
         \label{fig:error_vs_n}
     \end{subfigure}
     \begin{subfigure}[b]{0.33\textwidth}
         \centering
         \includegraphics[width=0.99\textwidth]{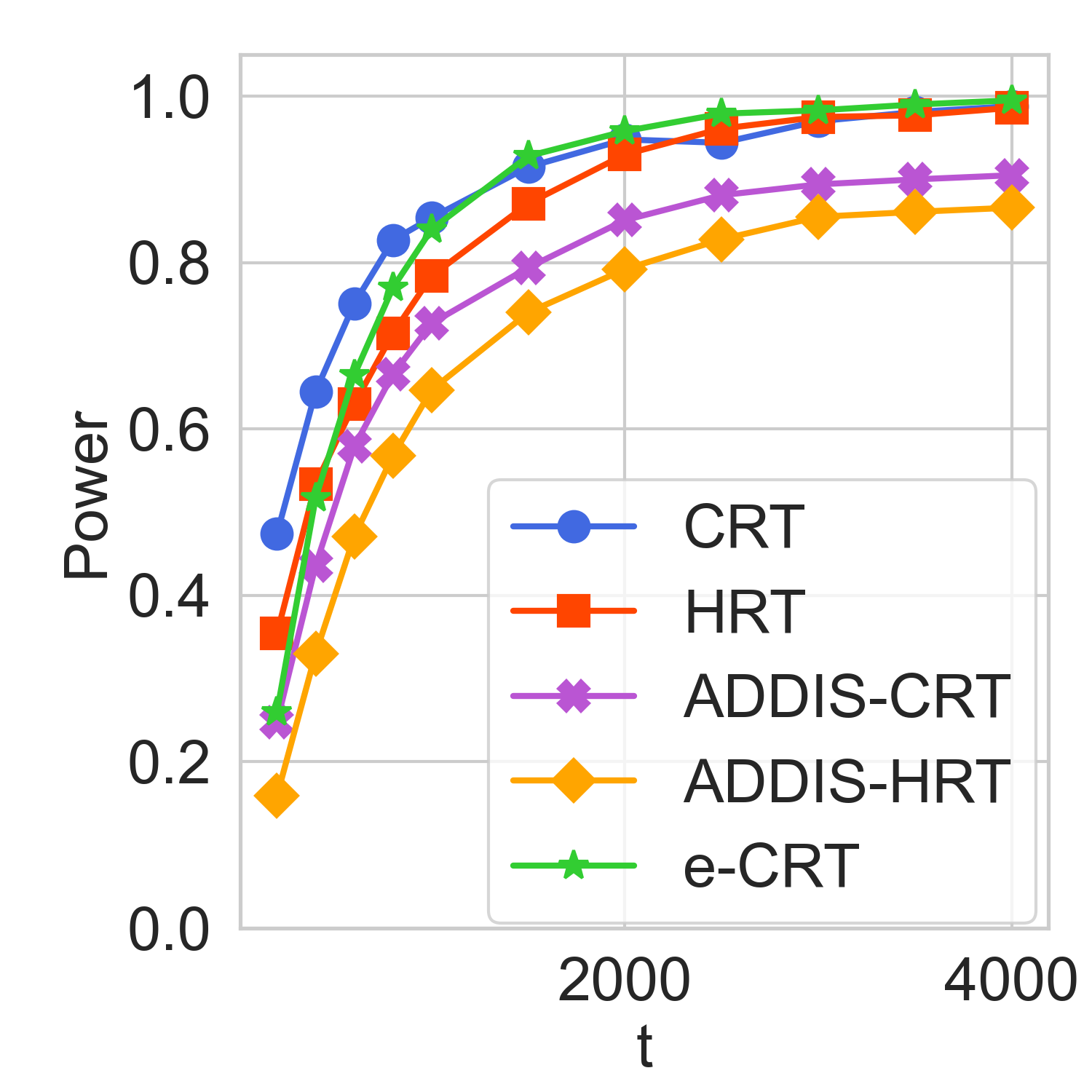}
         \label{fig:power_vs_n}
     \end{subfigure}
     \begin{subfigure}[b]{0.33\textwidth}
         \centering
        \includegraphics[width=0.99\textwidth]{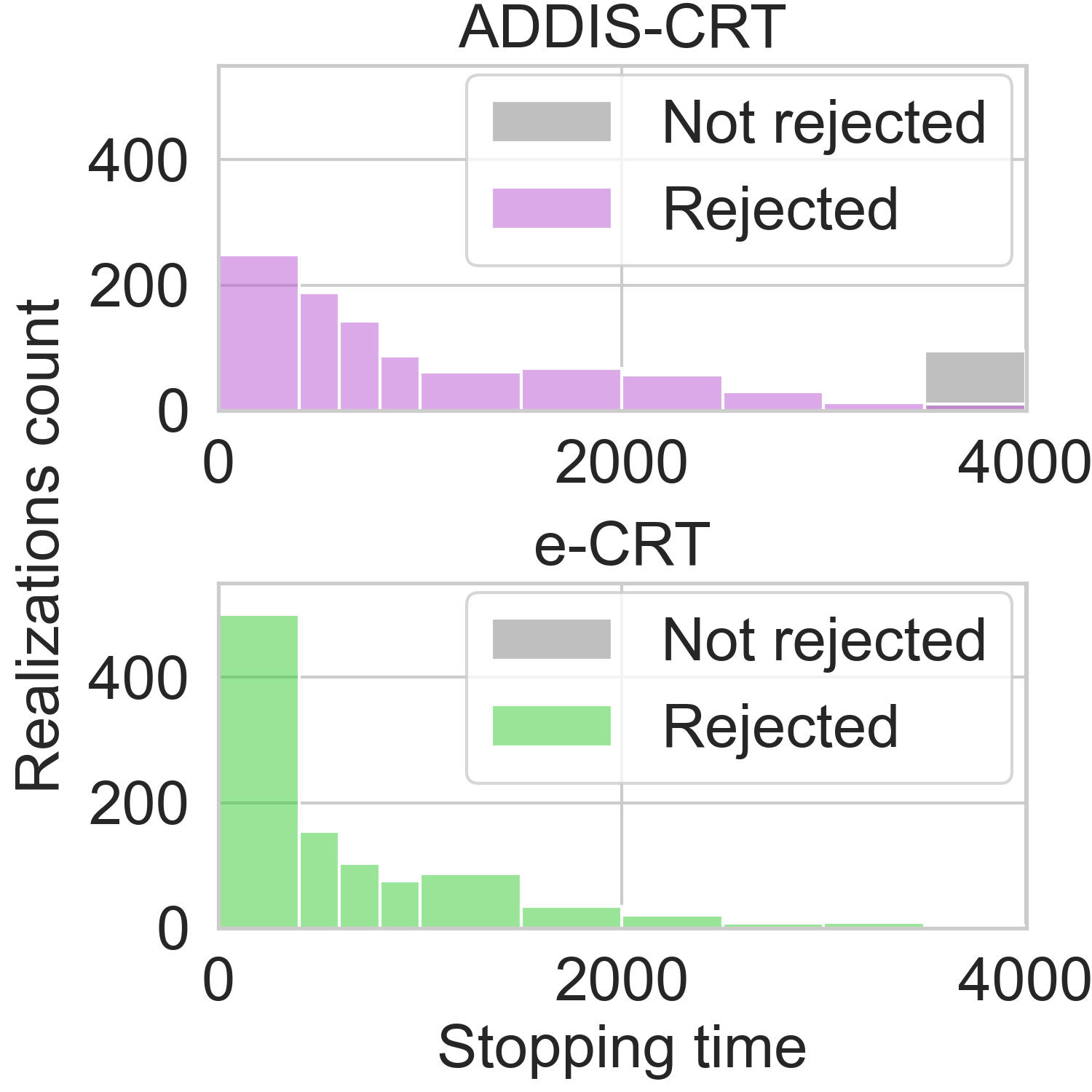}
         \label{fig:neff_vs_n}
     \end{subfigure}
        \caption{\textbf{Performance evaluation with simulated data}, evaluated over 1000 independent trials. Left: type-I error. Middle: empirical power. Right: histogram of the stopping times evaluated on the \texttt{non-null data}. We set $\alpha=0.05$.}
        \label{fig:comparison_graphs}
    \vspace{-0.25cm}
\end{figure*}

\textbf{Ensemble over batches.} The betting score $W_t$ presented in~\eqref{eq:lambda_t} is evaluated on a \emph{single} data point. This choice, however, might be inferior to evaluating $W_t$ on a \emph{batch} of several data points, as working with a batch is often less sensitive to the randomness in the data. On the other hand, a larger batch size reduces the total number of updates of the wealth process $S_t$ that we can make, and this may result in slower growth of the total wealth for a given number of data points. To mitigate the above trade-off, we suggest an ensemble approach by formulating the test martingale
\begin{equation}
    \label{eq:batch_ensemble}
    S_{t} = \frac{1}{|\mathcal{B}|}\sum_{b \in \mathcal{B}}{S_{t,b}}, \ \ \ t \in \mathbb{N}
\end{equation}
as an average of $|\mathcal{B}|$ test martingales $S_{t,b}$, where $\mathcal{B}$ is a set containing the batch-sizes and $|\cdot|$ returns the set size. Each of the batch martingales $S_{t,b}$ is evaluated analogously to~\eqref{eq:test_martingale}, however on a batch of size $b$ instead of a single data point. We refer the reader to Supplementary Section~\ref{supp:ensemble_batches} for more details on the ensemble approach, where we rigorously explain why~\eqref{eq:batch_ensemble} is a valid test martingale. This procedure is summarized in Algorithm~\ref{alg:practical_eCRT}.
The ablation study in Section~\ref{sec:ablation} demonstrates the above trade-off and the advantage of our ensemble procedure. 

\textbf{Unknown conditional.}
As in the CRT, to generate the dummy $\tilde{X}$ we assume that we have access to $P_{X|Z}$, however, in many real applications this distribution may not be known precisely. Here, we briefly discuss several use-cases where it is sensible to assume we have reasonable knowledge about $X \mid Z$, and we refer the reader to~\cite{candes2018panning} for a more detailed discussion. One such use-case is controlled experiments, e.g., genetic crossing experiments~\citep{haldane1931inbreeding}, sensitivity analysis of numerical models~\citep{saltelli2008global}, and gene knockout experiments~\citep{cong2013multiplex}. Another important use-case is when we have a large number of unlabeled observations of $(X, Z)$, so we can estimate $P_{X \mid Z}$ before applying the test. This is a reasonable assumption in various economic or genetic applications as we can collect covariates from different populations, or leverage previous studies that have acquired the same $(X,Z)$ however with different response variables. In such situations, we can utilize powerful machine learning techniques to estimate $P_{X\mid Z}$ using the available data, as suggested in previous work on model-X tests \citep{HRT,gimenez2019knockoffs,GANCRT,romano2020deep,sesia2019gene}. We will rely on such ideas in our experiments with real data. Importantly, the above line of research demonstrates the robustness of the CRT and knockoffs to errors in the estimation of $P_{X\mid Z}$. We believe it will be striking to provide a rigorous robustness theory for our e-CRT, possibly by following the approach presented in \citep{barber2020robust,berrett2020conditional}.

\section{SYNTHETIC EXPERIMENTS}
\label{sec:syn_experiments}
\subsection{Experimental Setup}
\label{sec:exp_setup}

In this section, we evaluate the performance of e-CRT in a controlled synthetic setting, where $P_{X \mid Z}$ and $P_{Y \mid XZ}$ are known. We generate a sequence of i.i.d. data points $\{(X_t,Y_t,Z_t)\}_{t=1}^n$, where $n$ is the maximal number of samples that can be collected. The features ${X_t,Z_t}$ are sampled as follows: $ Z_t \sim \mathcal{N}(0,I_{d})$ and $ X_t \mid Z_t \sim \mathcal{N}(u^\top Z_t,1)$ where $u \sim \mathcal{N}(0,I_d)$, 
so as $X_t$ and $Z_t$ are dependent by design~\citep{PCR}. Unless specified otherwise, we fix the number of covariates to $d=19$, so in total we have $20$ features.
Next, we consider two different conditional models for $P_{Y \mid X Z}$. The first model is used to examine the validity of our test by evaluating the type-I error rate, which we refer to as \texttt{null data} model. Specifically, we sample $Y_t$ such that $Y_t \indep X_t \mid Z_t$ as follows: \mbox{$Y_t \mid X_t,Z_t \sim \mathcal{N}((w^\top Z_t)^2,1)$}, where $w \sim \mathcal{N}(0,1)$.
To evaluate the empirical power---i.e., the rate we reject $H_0$ when applying a test with significance level $\alpha$---we define the following \texttt{non-null data} model in which \mbox{$Y_t \not\indep X_t \mid Z_t$} by construction: $Y_t \mid X_t,Z_t \sim \mathcal{N}((w^\top Z_t)^2+3X_t,1)$, where $w \sim \mathcal{N}(0,1)$.
We compare the performance of our e-CRT to offline~CRT~\citep{candes2018panning}, detailed in Section~\ref{sec:model_x}, as well as to the holdout randomization test (HRT)~\citep{HRT}, which is a computationally efficient variant of the offline CRT that often comes at the price of reduced power due to data splitting. All the methods use lasso regression model to compute the test statistics, whereas in e-CRT we fit the model online as described in Supplementary Section~\ref{supp:online}.
Additional implementation details on all methods are provided in Supplementary Section~\ref{supp:imp_details}.

We also compare e-CRT to out-of-the-box sequential versions of the offline CRT/HRT that allows monitoring the p-value $p_t$ over time. Towards that end, we apply the state-of-the-art ADDIS-spending approach \citep{tian2021online}, which rigorously adjusts the p-value at time $t$ by accounting for multiple comparisons in the time horizon, controlling $H_0$. We refer to the ADDIS-spending version of the CRT and HRT as ADDIS-CRT and ADDIS-HRT; see Supplementary Section~\ref{supp:imp_details} for implementation details. Unfortunately, we find it infeasible to generate a p-value $p_t$ for each time step $t$ due to the high computational complexity of the CRT: it requires fitting $M$ predictive models from scratch for each $t$, where we set $M=1000$ in our experiments to have a reasonable resolution for the p-value corrected by ADDIS-spending. Therefore, we apply ADDIS on p-values evaluated using CRT/HRT over a grid of 11 times steps in total.


\subsection{Type-I Error, Power, and Early Stopping}
\label{sec:exp_power_err}



\textbf{Type-I error.} Recall Figure~\ref{fig:fig_size_legend} from Section~\ref{sec:proposed}, illustrating that the test martingale $S_t$ does not grow significantly over time for a single realization of the \texttt{null data}. Here, we expand this experiment by reporting the type-I error rates of all the sequential tests as a function of $t$, evaluated on 1000 independent realizations of the \texttt{null data} model. Following Figure~\ref{fig:comparison_graphs} (left), we can see how the type-I error of our e-CRT is controlled and falls below the level $\alpha=0.05$ for all time steps $t$, as expected. The same conclusion holds for the sequential tests ADDIS-CRT and ADDIS-HRT. We also present in Supplementary Figure~\ref{fig:type_1_crt_hrt} the type-I error of the offline CRT and offline HRT, showing these also control the type-I error rate, however for a fixed sample size $t$.

\textbf{Robustness.} In practice, we do not have access to the sampling distribution of $X\mid Z$, and thus it is important to study the robustness of the test to approximation error in the sampling of $\tilde{X}$. In Supplementary Section~\ref{supp:parametric_est} we conduct such an experiment, showing that inflation in the type-I error occurs only when the estimation of $P_{X \mid Z}$ is far from the true distribution. Further, in Supplementary Section~\ref{supp:non_parametric_est} we consider a more challenging $X \mid Z$ that follows a student-$t$ distribution and show how a recent non-parametric method \citep{rosenberg2022fast} can be used to effectively estimate the conditional distribution. There, we demonstrate how the type-I error is controlled, and also that the power grows with $t$. A related experiment is given in Section~\ref{sec:real_data}, where we apply our method to real data for which $P_{X \mid Z}$ is unknown and thus must be estimated from data.

\textbf{Power.} The middle panel in Figure~\ref{fig:comparison_graphs} presents the empirical power as a function of $t$, evaluated on 1000 independent realizations of the \texttt{non-null data} model. Observe how the offline tests outperform the sequential ones when the sample size $t$ is relatively small.
Yet, for a larger number of samples with $t>500$, the e-CRT tends to outperform the offline HRT; this may be due to the sample inefficiency of the HRT as it involves data splitting. Interestingly, the e-CRT has comparable performance to CRT for $t>1000$, and the three tests nearly achieve power one when the sample size is large. Importantly, the e-CRT outperforms both ADDIS-CRT and ADDIS-HRT for all $t$, highlighting the advantages of our specialized test for CI compared to more out-of-the-box sequential solutions.

\textbf{Early stopping.} The strength of any sequential test---including our e-CRT---is the ability to monitor the outcome of the test and reject the null as soon as it exceeds a pre-defined threshold. Of course, the earlier the rejection happens the better the sample efficiency of the test. The right panel in Figure~\ref{fig:comparison_graphs} presents the histogram of the stopping times of ADDIS-CRT and e-CRT over the 1000 realizations of the \texttt{non-null data} used in the power experiments. As can be seen, e-CRT tends to reject the null earlier than ADDIS-CRT, indicating superior sample efficiency.


\textbf{Additional experiments.} 
Supplementary Section~\ref{supp:varying_d} studies the effect of the number of covariates $d$ on the performance of e-CRT, where the overall trend is that the power is decreased as $d$ is increased, and the type-I error is controlled. We also examine the impact of the dependency strength between $X$ and $Z$ on the performance of our e-CRT in Supplementary Section~\ref{supp:corr}. There, we observe a decrease in power as the correlation increases, while controlling the type-I error. This trend is in line with previous studies that focus on the offline setting, see, e.g., \cite{candes2018panning,dCRT,MRD}.

\begin{figure}[ht!]
     \centering
     \begin{subfigure}[b]{0.495\textwidth}
         \centering
         \includegraphics[width=0.6\textwidth]{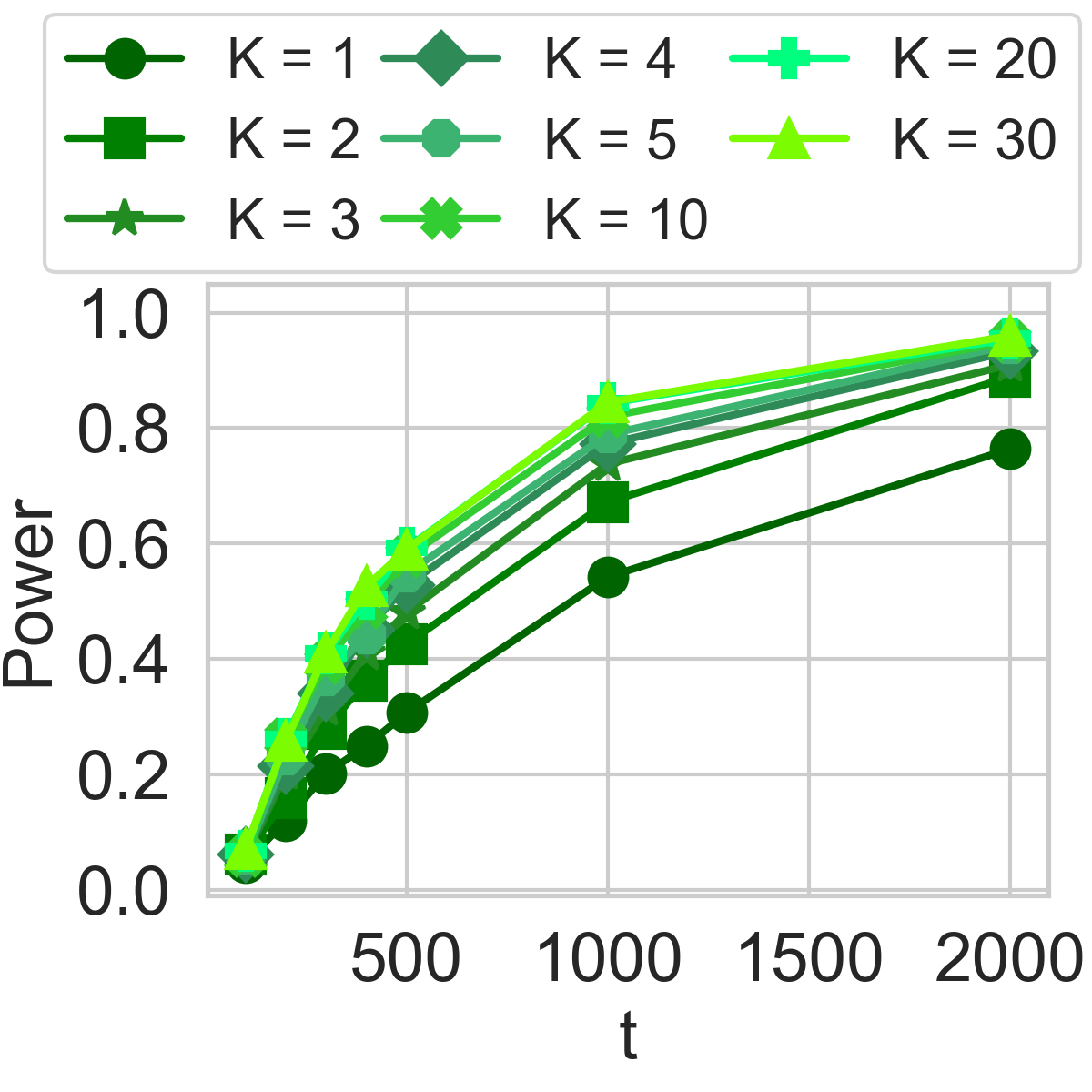}
         \label{fig:power_vs_n_for_k}
     \end{subfigure}
     \begin{subfigure}[b]{0.495\textwidth}
         \centering
         \includegraphics[width=0.6\textwidth]{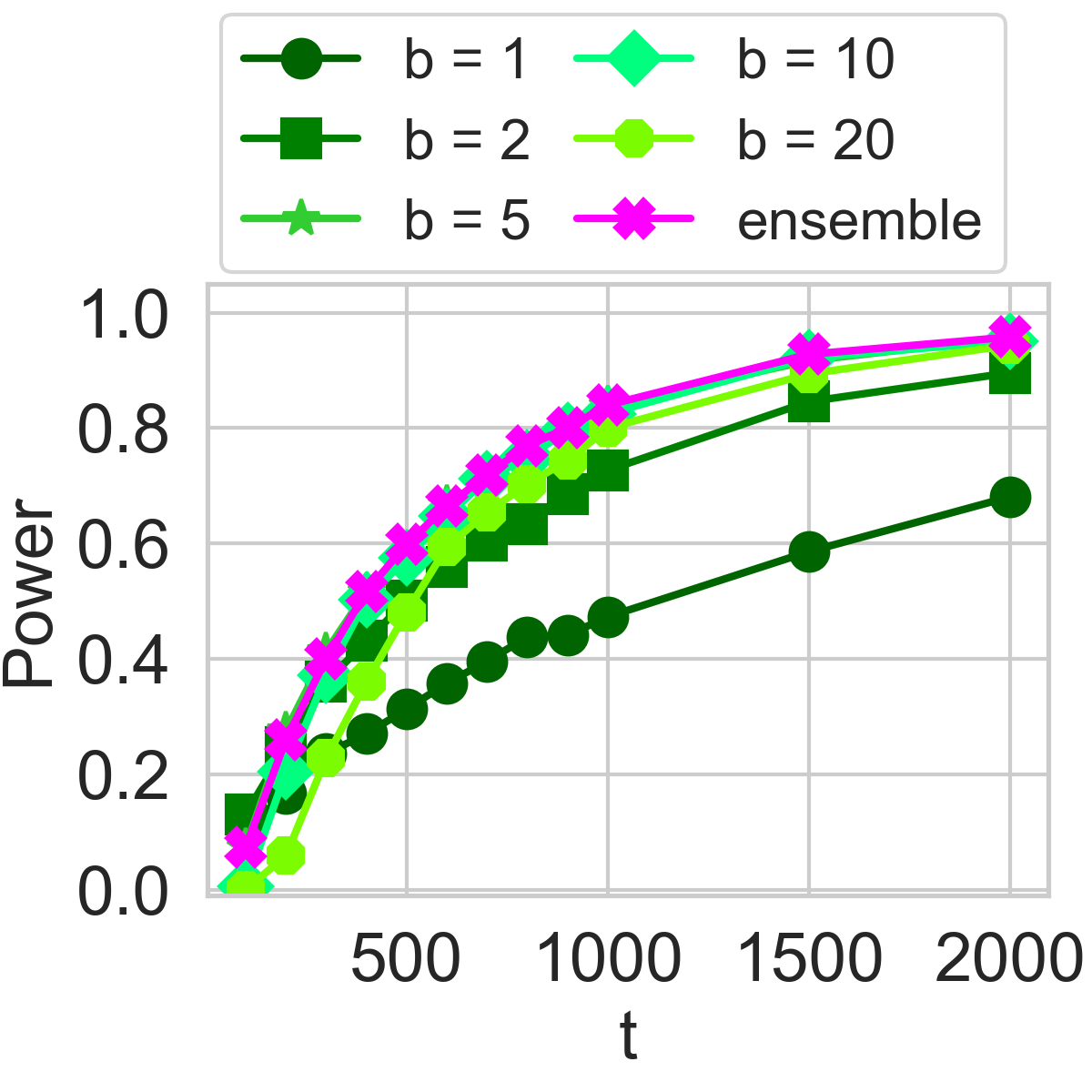}
         \label{fig:power_vs_n_for_b}
     \end{subfigure}
        \caption{\textbf{Ablation study.} Empirical power of e-CRT ($\alpha=0.05$), evaluated on $1000$ realizations of the \texttt{non-null data}. Top: the effect of the de-randomization parameter $K$ on power.
        Bottom: the effect of the batch size on power, in comparison to the batch-ensemble approach in~\eqref{eq:batch_ensemble}.
        }
        \label{fig:k_and_b_plots}
\end{figure}

\subsection{Ablation Study}
\label{sec:ablation}

\textbf{The effect of the de-randomization parameter $K$.}
In Section~\ref{sec:practical} we suggest using an average betting score over $K$ realizations of $W_t$ to reduce the randomization induced by the generation of the dummy features.  To study the effect of the de-randomization parameter $K$ on the power of the e-CRT, we generate \texttt{non-null data} stream of length $n=2000$ and present in Figure~\ref{fig:k_and_b_plots} (top) the power of e-CRT as a function of $t$ for several values of $K$. It is evident that the power of the test is improved as $K$ increases, with a maximal absolute improvement of about $20\%$.

\textbf{The effect of the batch size.} In Section~\ref{sec:practical} we discuss how different batch sizes might affect the performance of the e-CRT, and describe the trade-off between using small and large batch sizes.  To visualize this trade-off, we apply the e-CRT on \texttt{non-null data} with and without the ensemble-over-batches approach, for several choices of batch sizes. Following Figure~\ref{fig:k_and_b_plots}~(bottom), we can see that, for smaller sample sizes, e-CRT with small batches performs better than with large batches in terms of power. However, the opposite is true for larger sample sizes, in which larger batches are favored. Our proposed ensemble approach achieves a remarkable performance: it tends to follow the leading choice for all sample sizes tested.


\section{REAL DATA EXPERIMENTS}
\label{sec:real_data}

\subsection{Fund Replication}
\label{sec:fund_rep}
We begin with the task of identifying which stocks contribute to the performance of known index funds.
Our study follows \cite{challet2021financial,spectorasymptotically} who deployed the knockoff filter in this context. Although the samples are not i.i.d. and thus not satisfying the model-X assumptions, we present this application to illustrate how the e-CRT can be applied to real data. In our experiments, we focus on a technology sector index fund named XLK, and follow the data collection procedure of \cite{spectorasymptotically}. The data consists of the daily log returns of each stock in the S\&P 500 since 2013 and the corresponding daily log return of XLK. We exclude samples and features with missing data, resulting in $n=2421$ and $d=457$. More details are in Supplementary Section~\ref{supp:fund_imp}.

Table~\ref{tab:fund_rep} in Supplementary Section~\ref{supp:fund_table} summarizes the results obtained by applying e-CRT, CRT and HRT on each of the $d = 457$ stocks. We classify a stock as important if it currently belongs to the technology sector XLK. We report the p-values obtained by CRT and HRT for each stock, and the test martingale  $S_{t_\text{stop}}$ for e-CRT, where $t_\text{stop}$ is the stopping time for a test level of $\alpha=0.05$.
Following Supplementary Table~\ref{tab:fund_rep}, we can see that the e-CRT tends to reject the null for stocks that are currently in XLK, and avoids rejecting stocks that are currently not in XLK. Observe the advantage of early stopping: the e-CRT rejects some of the important stocks with a relatively small sample size $t_\text{stop}$.

\subsection{HIV Drug Resistance}
\label{sec:hiv}
Herein, we consider the task of detecting genetic mutations in human immunodeficiency virus (HIV) of type-I that are associated with drug resistance~\citep{HIV}. We follow \cite{romano2020deep} and study the resistance to the Lopinavir protease inhibitor drug, applying the same pre-processing steps to the raw data. Consequently, this data set consists of $n=1555$ samples of $d=150$ features. We consider the data points as if they arrive sequentially and apply the test on each feature. See Supplementary Section~\ref{supp:hiv_imp} for details on the data and the tests we apply.

Table~\ref{tab:real_full_results} in Supplementary Section~\ref{supp:hiv_res} summarizes the results obtained by running e-CRT, CRT, and HRT on each of the $d=150$ mutations, analogously to the fund replication experiment. We classify each mutation by its effect on drug resistance, as reported in previous studies.
Following Supplementary Table~\ref{tab:real_full_results}, we can see that our e-CRT tends to reject the null for mutations that have been previously reported to have a `major' or `accessory' effect and avoid the `unknown' ones. The e-CRT rejects some of the `major'/`accessory' mutations with a relatively small sample size $t_\text{stop}$, demonstrating the advantage of early stopping. Figure~\ref{fig:real_data_martingales} in Supplementary Section~\ref{supp:hiv_res} portrays three test martingales for representative mutations. Figure~\ref{fig:real_data_ecrt_12s} corresponds to a mutation that has not been reported in previous studies to have an effect on drug resistance. Indeed, the test martingale does not grow significantly. By contrast, Figure~\ref{fig:real_data_ecrt_47v} corresponds to a mutation that has been reported to have a major effect on drug resistance, for which $S_t$ grows fast and reaches $1/\alpha=20$ using only $240$ samples.
Lastly, Figure~\ref{fig:real_data_ecrt_48m} corresponds to a mutation that has been also reported to have a major effect. Here, $S_t$ grows at a slower rate and does not reach the nominal level of $1/\alpha=20$. Yet, it achieves a final value of 6.5, which provides substantial evidence against the null~\citep{evalues}.

\section{CONCLUSION}
\label{sec:conclusions}
In this paper, we develop e-CRT---a sequential CI test that allows processing each data point as soon as it arrives, supporting early stopping while controlling the type-I error for all $t$ simultaneously. Our proposed test is inspired by the model-X randomization test \cite{candes2018panning}, and testing by betting~\cite{shafer2019book,shafer2021testing}. We prove the validity of e-CRT and propose several design choices to improve its power, including de-randomization and batch ensemble. Numerical experiments demonstrate the validity of our e-CRT, its superiority over existing out-of-the-box sequential tests for CI, and the impact of the design choices we make on the power of the test.

One important future direction is to theoretically analyze the robustness of the e-CRT to errors in the estimation of $P_{X\mid Z}$, rigorously characterizing the potential inflation in type-I error rate \cite{barber2020robust,grunwald2022anytime}. From a practical perspective, it would be illuminating to develop more robust betting scores, for example, by taking into account the error in the estimation of $P_{X\mid Z}$ or by borrowing ideas from the doubly robust literature \cite{shah2020hardness,shi2021double,niu2022reconciling}. Another direction is to explore new ways to fit predictive models and form more powerful betting scores. In our experiments, we train the predictive model on the original data to minimize the MSE while ignoring the dummy features. Recently, in the context of HRT, \cite{MRD} developed a new loss function that takes into account the dummy features during training to improve the power of the test. We believe such an approach can be used in our context as well.

\section*{Acknowledgements}

This research was supported by the ISRAEL SCIENCE FOUNDATION (grant No. 729/21). Y.R. also thanks the Career Advancement Fellowship, Technion, for providing additional research support.

\bibliography{biblio}

\begin{thebibliography}{60}
\providecommand{\natexlab}[1]{#1}
\providecommand{\url}[1]{\texttt{#1}}
\expandafter\ifx\csname urlstyle\endcsname\relax
  \providecommand{\doi}[1]{doi: #1}\else
  \providecommand{\doi}{doi: \begingroup \urlstyle{rm}\Url}\fi

\bibitem[Angrist and Kuersteiner(2011)]{angrist2011causal}
J.~D. Angrist and G.~M. Kuersteiner.
\newblock Causal effects of monetary shocks: Semiparametric conditional
  independence tests with a multinomial propensity score.
\newblock \emph{Review of Economics and Statistics}, 93\penalty0 (3):\penalty0
  725--747, 2011.

\bibitem[Balsubramani and Ramdas(2016)]{balsubramani2015sequential}
A.~Balsubramani and A.~Ramdas.
\newblock Sequential nonparametric testing with the law of the iterated
  logarithm.
\newblock In \emph{Proceedings of the Thirty-Second Conference on Uncertainty
  in Artificial Intelligence (UAI)}, 2016.

\bibitem[Barber and Cand{\`e}s(2015)]{barber2015controlling}
R.~F. Barber and E.~J. Cand{\`e}s.
\newblock Controlling the false discovery rate via knockoffs.
\newblock \emph{The Annals of Statistics}, 43\penalty0 (5):\penalty0
  2055--2085, 2015.

\bibitem[Barber et~al.(2020)Barber, Cand{\`e}s, and Samworth]{barber2020robust}
R.~F. Barber, E.~J. Cand{\`e}s, and R.~J. Samworth.
\newblock Robust inference with knockoffs.
\newblock \emph{The Annals of Statistics}, 48\penalty0 (3):\penalty0
  1409--1431, 2020.

\bibitem[Bates et~al.(2020)Bates, Sesia, Sabatti, and
  Cand{\`e}s]{bates2020causal}
S.~Bates, M.~Sesia, C.~Sabatti, and E.~Cand{\`e}s.
\newblock Causal inference in genetic trio studies.
\newblock \emph{Proceedings of the National Academy of Sciences}, 117\penalty0
  (39):\penalty0 24117--24126, 2020.

\bibitem[Bellot and van~der Schaar(2019)]{GANCRT}
A.~Bellot and M.~van~der Schaar.
\newblock Conditional independence testing using generative adversarial
  networks.
\newblock \emph{Advances in Neural Information Processing Systems}, 32, 2019.

\bibitem[Berrett et~al.(2020)Berrett, Wang, Barber, and
  Samworth]{berrett2020conditional}
T.~B. Berrett, Y.~Wang, R.~F. Barber, and R.~J. Samworth.
\newblock The conditional permutation test for independence while controlling
  for confounders.
\newblock \emph{Journal of the Royal Statistical Society: Series B (Statistical
  Methodology)}, 82\penalty0 (1):\penalty0 175--197, 2020.

\bibitem[Bhui(2019)]{bhui2019testing}
R.~Bhui.
\newblock Testing optimal timing in value-linked decision making.
\newblock \emph{Computational Brain \& Behavior}, 2\penalty0 (2):\penalty0
  85--94, 2019.

\bibitem[Boyd et~al.(2011)Boyd, Parikh, Chu, Peleato, Eckstein, et~al.]{ADMM}
S.~Boyd, N.~Parikh, E.~Chu, B.~Peleato, J.~Eckstein, et~al.
\newblock Distributed optimization and statistical learning via the alternating
  direction method of multipliers.
\newblock \emph{Foundations and Trends{\textregistered} in Machine learning},
  3\penalty0 (1):\penalty0 1--122, 2011.

\bibitem[Burns et~al.(2020)Burns, Thomason, and Tansey]{burns2020interpreting}
C.~Burns, J.~Thomason, and W.~Tansey.
\newblock Interpreting black box models via hypothesis testing.
\newblock In \emph{Proceedings of the 2020 ACM-IMS on Foundations of Data
  Science Conference}, pages 47--57, 2020.

\bibitem[Candes et~al.(2018)Candes, Fan, Janson, and Lv]{candes2018panning}
E.~Candes, Y.~Fan, L.~Janson, and J.~Lv.
\newblock Panning for gold: {Model-X} knockoffs for high dimensional controlled
  variable selection.
\newblock \emph{Journal of the Royal Statistical Society: Series B (Statistical
  Methodology)}, 80\penalty0 (3):\penalty0 551--577, 2018.

\bibitem[Challet et~al.(2021)Challet, Bongiorno, and
  Pelletier]{challet2021financial}
D.~Challet, C.~Bongiorno, and G.~Pelletier.
\newblock Financial factors selection with knockoffs: fund replication,
  explanatory and prediction networks.
\newblock \emph{Physica A: Statistical Mechanics and its Applications},
  580:\penalty0 126105, 2021.

\bibitem[Cong et~al.(2013)Cong, Ran, Cox, Lin, Barretto, Habib, Hsu, Wu, Jiang,
  Marraffini, et~al.]{cong2013multiplex}
L.~Cong, F.~A. Ran, D.~Cox, S.~Lin, R.~Barretto, N.~Habib, P.~D. Hsu, X.~Wu,
  W.~Jiang, L.~A. Marraffini, et~al.
\newblock Multiplex genome engineering using crispr/cas systems.
\newblock \emph{Science}, 339\penalty0 (6121):\penalty0 819--823, 2013.

\bibitem[Cover(2011)]{cover2011universal}
T.~M. Cover.
\newblock Universal portfolios.
\newblock In \emph{The Kelly Capital Growth Investment Criterion: Theory and
  Practice}, pages 181--209. World Scientific, 2011.

\bibitem[Duan et~al.(2022)Duan, Ramdas, and Wasserman]{ibet}
B.~Duan, A.~Ramdas, and L.~Wasserman.
\newblock Interactive rank testing by betting.
\newblock In \emph{Conference on Causal Learning and Reasoning}, pages
  201--235. PMLR, 2022.

\bibitem[Gimenez et~al.(2019)Gimenez, Ghorbani, and Zou]{gimenez2019knockoffs}
J.~R. Gimenez, A.~Ghorbani, and J.~Zou.
\newblock Knockoffs for the mass: new feature importance statistics with false
  discovery guarantees.
\newblock In \emph{The 22nd International Conference on Artificial Intelligence
  and Statistics}, pages 2125--2133. PMLR, 2019.

\bibitem[Gr{\"u}nwald et~al.(2020)Gr{\"u}nwald, de~Heide, and
  Koolen]{grunwald2020safe}
P.~Gr{\"u}nwald, R.~de~Heide, and W.~M. Koolen.
\newblock Safe testing.
\newblock In \emph{2020 Information Theory and Applications Workshop (ITA)},
  pages 1--54. IEEE, 2020.

\bibitem[Gr{\"u}nwald et~al.(2022)Gr{\"u}nwald, Henzi, and
  Lardy]{grunwald2022anytime}
P.~Gr{\"u}nwald, A.~Henzi, and T.~Lardy.
\newblock Anytime valid tests of conditional independence under model-x.
\newblock \emph{arXiv preprint arXiv:2209.12637}, 2022.

\bibitem[Haldane and Waddington(1931)]{haldane1931inbreeding}
J.~Haldane and C.~Waddington.
\newblock Inbreeding and linkage.
\newblock \emph{Genetics}, 16\penalty0 (4):\penalty0 357, 1931.

\bibitem[Javanmard and Mehrabi(2021)]{PCR}
A.~Javanmard and M.~Mehrabi.
\newblock Pearson chi-squared conditional randomization test.
\newblock \emph{arXiv preprint arXiv:2111.00027}, 2021.

\bibitem[Jun and Orabona(2019)]{jun2019parameter}
K.-S. Jun and F.~Orabona.
\newblock Parameter-free online convex optimization with sub-exponential noise.
\newblock In \emph{Conference on Learning Theory}, pages 1802--1823. PMLR,
  2019.

\bibitem[Lai(1984)]{lai1984incorporating}
T.~Lai.
\newblock Incorporating scientific, ethical and economic considerations into
  the design of clinical trials in the pharmaceutical industry: A sequential
  approach.
\newblock \emph{Communications in Statistics-Theory and Methods}, 13\penalty0
  (19):\penalty0 2355--2368, 1984.

\bibitem[Lei and Fithian(2018)]{lei2018adapt}
L.~Lei and W.~Fithian.
\newblock Adapt: an interactive procedure for multiple testing with side
  information.
\newblock \emph{Journal of the Royal Statistical Society: Series B (Statistical
  Methodology)}, 80\penalty0 (4):\penalty0 649--679, 2018.

\bibitem[Lei et~al.(2021)Lei, Ramdas, and Fithian]{lei2021general}
L.~Lei, A.~Ramdas, and W.~Fithian.
\newblock A general interactive framework for false discovery rate control
  under structural constraints.
\newblock \emph{Biometrika}, 108\penalty0 (2):\penalty0 253--267, 2021.

\bibitem[Lh{\'e}ritier and Cazals(2018)]{lheritier2018sequential}
A.~Lh{\'e}ritier and F.~Cazals.
\newblock A sequential non-parametric multivariate two-sample test.
\newblock \emph{IEEE Transactions on Information Theory}, 64\penalty0
  (5):\penalty0 3361--3370, 2018.

\bibitem[Liu et~al.(2022)Liu, Katsevich, Janson, and Ramdas]{dCRT}
M.~Liu, E.~Katsevich, L.~Janson, and A.~Ramdas.
\newblock Fast and powerful conditional randomization testing via distillation.
\newblock \emph{Biometrika}, 109\penalty0 (2):\penalty0 277--293, 2022.

\bibitem[Lu et~al.(2018)Lu, Fan, Lv, and Stafford~Noble]{lu2018deeppink}
Y.~Lu, Y.~Fan, J.~Lv, and W.~Stafford~Noble.
\newblock {DeepPINK}: reproducible feature selection in deep neural networks.
\newblock \emph{Advances in neural information processing systems}, 31, 2018.

\bibitem[Naghshvar and Javidi(2013)]{naghshvar2013active}
M.~Naghshvar and T.~Javidi.
\newblock Active sequential hypothesis testing.
\newblock \emph{The Annals of Statistics}, 41\penalty0 (6):\penalty0
  2703--2738, 2013.

\bibitem[Nikolakopoulou et~al.(2018)Nikolakopoulou, Mavridis, Egger, and
  Salanti]{nikolakopoulou2018continuously}
A.~Nikolakopoulou, D.~Mavridis, M.~Egger, and G.~Salanti.
\newblock Continuously updated network meta-analysis and statistical monitoring
  for timely decision-making.
\newblock \emph{Statistical methods in medical research}, 27\penalty0
  (5):\penalty0 1312--1330, 2018.

\bibitem[Niu et~al.(2022)Niu, Chakraborty, Dukes, and
  Katsevich]{niu2022reconciling}
Z.~Niu, A.~Chakraborty, O.~Dukes, and E.~Katsevich.
\newblock Reconciling model-x and doubly robust approaches to conditional
  independence testing.
\newblock \emph{arXiv preprint arXiv:2211.14698}, 2022.

\bibitem[Park et~al.(2018)Park, Thorlund, and Mills]{park2018critical}
J.~J. Park, K.~Thorlund, and E.~J. Mills.
\newblock Critical concepts in adaptive clinical trials.
\newblock \emph{Clinical epidemiology}, 10:\penalty0 343, 2018.

\bibitem[Pearl et~al.(2000)]{pearl2000models}
J.~Pearl et~al.
\newblock Models, reasoning and inference.
\newblock \emph{Cambridge, UK: CambridgeUniversityPress}, 19\penalty0 (2),
  2000.

\bibitem[Peters et~al.(2017)Peters, Janzing, and
  Sch{\"o}lkopf]{peters2017elements}
J.~Peters, D.~Janzing, and B.~Sch{\"o}lkopf.
\newblock \emph{Elements of causal inference: foundations and learning
  algorithms}.
\newblock The MIT Press, 2017.

\bibitem[Ramdas and Wenbe(2020)]{lady}
A.~Ramdas and L.~Wenbe.
\newblock The lady keeps tasting coffee: randomization inference by betting.
\newblock \emph{Unpublished manuscript}, 2020.

\bibitem[Ramdas et~al.(2020)Ramdas, Ruf, Larsson, and
  Koolen]{ramdas2020admissible}
A.~Ramdas, J.~Ruf, M.~Larsson, and W.~Koolen.
\newblock Admissible anytime-valid sequential inference must rely on
  nonnegative martingales.
\newblock \emph{arXiv preprint arXiv:2009.03167}, 2020.

\bibitem[Ramdas et~al.(2022)Ramdas, Ruf, Larsson, and
  Koolen]{ramdas2022testing}
A.~Ramdas, J.~Ruf, M.~Larsson, and W.~M. Koolen.
\newblock Testing exchangeability: Fork-convexity, supermartingales and
  e-processes.
\newblock \emph{International Journal of Approximate Reasoning}, 141:\penalty0
  83--109, 2022.

\bibitem[Ren and Barber(2022)]{ren2022derandomized}
Z.~Ren and R.~F. Barber.
\newblock Derandomized knockoffs: leveraging e-values for false discovery rate
  control.
\newblock \emph{arXiv preprint arXiv:2205.15461}, 2022.

\bibitem[Rhee et~al.(2006)Rhee, Taylor, Wadhera, Ben-Hur, Brutlag, and
  Shafer]{HIV}
S.-Y. Rhee, J.~Taylor, G.~Wadhera, A.~Ben-Hur, D.~L. Brutlag, and R.~W. Shafer.
\newblock Genotypic predictors of human immunodeficiency virus type 1 drug
  resistance.
\newblock \emph{Proceedings of the National Academy of Sciences}, 103\penalty0
  (46):\penalty0 17355--17360, 2006.

\bibitem[Romano et~al.(2020)Romano, Sesia, and Cand{\`e}s]{romano2020deep}
Y.~Romano, M.~Sesia, and E.~Cand{\`e}s.
\newblock Deep knockoffs.
\newblock \emph{Journal of the American Statistical Association}, 115\penalty0
  (532):\penalty0 1861--1872, 2020.

\bibitem[Rosenberg et~al.(2022)Rosenberg, Vedula, Romano, and
  Bronstein]{rosenberg2022fast}
A.~A. Rosenberg, S.~Vedula, Y.~Romano, and A.~M. Bronstein.
\newblock Fast nonlinear vector quantile regression.
\newblock \emph{arXiv preprint arXiv:2205.14977}, 2022.

\bibitem[Saltelli et~al.(2008)Saltelli, Ratto, Andres, Campolongo, Cariboni,
  Gatelli, Saisana, and Tarantola]{saltelli2008global}
A.~Saltelli, M.~Ratto, T.~Andres, F.~Campolongo, J.~Cariboni, D.~Gatelli,
  M.~Saisana, and S.~Tarantola.
\newblock \emph{Global sensitivity analysis: the primer}.
\newblock John Wiley \& Sons, 2008.

\bibitem[Sesia et~al.(2019)Sesia, Sabatti, and Cand{\`e}s]{sesia2019gene}
M.~Sesia, C.~Sabatti, and E.~J. Cand{\`e}s.
\newblock Gene hunting with hidden markov model knockoffs.
\newblock \emph{Biometrika}, 106\penalty0 (1):\penalty0 1--18, 2019.

\bibitem[Shaer and Romano(2023)]{MRD}
S.~Shaer and Y.~Romano.
\newblock Learning to increase the power of conditional randomization tests.
\newblock \emph{Machine Learning}, pages 1--41, 2023.

\bibitem[Shafer(2021)]{shafer2021testing}
G.~Shafer.
\newblock Testing by betting: A strategy for statistical and scientific
  communication.
\newblock \emph{Journal of the Royal Statistical Society: Series A (Statistics
  in Society)}, 184\penalty0 (2):\penalty0 407--431, 2021.

\bibitem[Shafer and Vovk(2019)]{shafer2019book}
G.~Shafer and V.~Vovk.
\newblock \emph{Game-Theoretic Foundations for Probability and Finance}, volume
  455.
\newblock John Wiley \& Sons, 2019.

\bibitem[Shah and Peters(2020)]{shah2020hardness}
R.~D. Shah and J.~Peters.
\newblock The hardness of conditional independence testing and the generalised
  covariance measure.
\newblock 2020.

\bibitem[Shekhar and Ramdas(2021)]{two_sample_aaditya}
S.~Shekhar and A.~Ramdas.
\newblock Game-theoretic formulations of sequential nonparametric one-and
  two-sample tests.
\newblock \emph{arXiv preprint arXiv:2112.09162}, 2021.

\bibitem[Shi et~al.(2021)Shi, Xu, Bergsma, and Li]{shi2021double}
C.~Shi, T.~Xu, W.~Bergsma, and L.~Li.
\newblock Double generative adversarial networks for conditional independence
  testing.
\newblock \emph{The Journal of Machine Learning Research}, 22\penalty0
  (1):\penalty0 13029--13060, 2021.

\bibitem[Spector and Fithian(2022)]{spectorasymptotically}
A.~Spector and W.~Fithian.
\newblock Asymptotically optimal knockoff statistics via the masked likelihood
  ratio.
\newblock 2022.

\bibitem[Tansey et~al.(2022)Tansey, Veitch, Zhang, Rabadan, and Blei]{HRT}
W.~Tansey, V.~Veitch, H.~Zhang, R.~Rabadan, and D.~M. Blei.
\newblock The holdout randomization test for feature selection in black box
  models.
\newblock \emph{Journal of Computational and Graphical Statistics}, 31\penalty0
  (1):\penalty0 151--162, 2022.

\bibitem[Tian and Ramdas(2021)]{tian2021online}
J.~Tian and A.~Ramdas.
\newblock Online control of the familywise error rate.
\newblock \emph{Statistical Methods in Medical Research}, 30\penalty0
  (4):\penalty0 976--993, 2021.

\bibitem[Tonelli(1909)]{tonelli}
L.~Tonelli.
\newblock Sull’integrazione per parti.
\newblock \emph{Rend. Acc. Naz. Lincei}, 5\penalty0 (18):\penalty0 246--253,
  1909.

\bibitem[Turner et~al.(2021)Turner, Ly, and Gr{\"u}nwald]{two_sample_grunwald}
R.~Turner, A.~Ly, and P.~Gr{\"u}nwald.
\newblock Two-sample tests that are safe under optional stopping, with an
  application to contingency tables.
\newblock \emph{arXiv preprint arXiv:2106.02693}, 2021.

\bibitem[Ville(1939)]{ville}
J.~Ville.
\newblock \emph{1ere these: Etude critique de la notion de collectif; 2eme
  these: La transformation de Laplace}.
\newblock PhD thesis, Gauthier-Villars \& Cie, 1939.

\bibitem[Vovk and Wang(2021)]{evalues}
V.~Vovk and R.~Wang.
\newblock E-values: Calibration, combination and applications.
\newblock \emph{The Annals of Statistics}, 49\penalty0 (3):\penalty0
  1736--1754, 2021.

\bibitem[Wald(1945)]{wald1945sequential}
A.~Wald.
\newblock Sequential tests of statistical hypotheses.
\newblock \emph{Annals of Mathematical Statistics}, 1945.

\bibitem[Wang and Ramdas(2022)]{ebh}
R.~Wang and A.~Ramdas.
\newblock False discovery rate control with e-values.
\newblock \emph{Journal of the Royal statistical society: series B
  (Methodological)}, 2022.

\bibitem[Wang and Hong(2018)]{wang2018characteristic}
X.~Wang and Y.~Hong.
\newblock Characteristic function based testing for conditional independence: A
  nonparametric regression approach.
\newblock \emph{Econometric Theory}, 34\penalty0 (4):\penalty0 815--849, 2018.

\bibitem[Waudby-Smith and Ramdas(2023)]{waudby2020estimating}
I.~Waudby-Smith and A.~Ramdas.
\newblock Estimating means of bounded random variables by betting.
\newblock \emph{Journal of the Royal Statistical Society: Series B
  (Methodology), with discussion}, 2023.

\bibitem[Wright(2015)]{coordinate_descent}
S.~J. Wright.
\newblock Coordinate descent algorithms.
\newblock \emph{Mathematical Programming}, 151\penalty0 (1):\penalty0 3--34,
  2015.

\end{thebibliography}
\medskip

\clearpage
\begin{appendix}
\onecolumn
\aistatstitle{Supplementary Material: Model-X Sequential Testing for Conditional Independence via Testing by Betting
}
\section{MATHEMATICAL PROOFS}
\label{supp:proofs}

\begin{proof}[Proof of Lemma~\ref{lemma:swap_ZY}]
Observe that it is equivalent to showing that
\begin{align}
\label{eq:XXtY_Z}
(X_t,\tilde{X}_t,Y_t) \mid Z_t \overset{d}{=} (\tilde{X}_t,X_t,Y_t) \mid Z_t,
\end{align}
since the marginal distribution 
$P_Z$ is identical on both sides of \eqref{eq:XXtY_Z}. Below, we use discrete random variables for simplicity, as the continuous case can be proved analogously. From the law of total probability, we can write relation \eqref{eq:XXtY_Z} as follows:
\begin{align}
\label{eq:law_of_tot}
    \mathbb{P}_{Y\mid X \tilde{X} Z}(y \mid a, b, z) \cdot \mathbb{P}_{X \tilde{X} \mid Z}(a, b \mid z) &= \mathbb{P}_{Y\mid \tilde{X} X Z}(y \mid a, b,z) \cdot \mathbb{P}_{\tilde{X} X \mid Z}(a, b \mid z).
\end{align}
Now, recall that $(X_t,\tilde{X}_t) \mid Z_t \overset{d}{=} (\tilde{X}_t,X_t) \mid Z_t$ by construction, therefore the conditional distributions $\mathbb{P}_{\tilde{X} X \mid Z}$ and $\mathbb{P}_{ X \tilde{X}\mid Z}$ on both sides of \eqref{eq:law_of_tot} are the same. As a result, it suffices to show that 
\begin{equation}
    \label{eq:Y_XX}
    Y_t \mid (X_t,\tilde{X}_t,Z_t) \overset{d}{=} Y_t \mid (\tilde{X}_t,X_t,Z_t).
\end{equation}
The above relation holds once observing that
\begin{align}
    \mathbb{P}_{Y \mid X \tilde{X} Z }(y \mid a,b,z) &= \mathbb{P}_{Y \mid Z }(y \mid z) \\
    &= \mathbb{P}_{Y \mid \tilde{X} {X} Z }(y \mid a,b,z), 
\end{align}
where the first and second equality hold since $Y_t \indep X_t \mid Z_t$, $Y_t \indep \tilde{X}_t \mid Z_t$, and $X_t \indep \tilde{X}_t \mid Z_t$, implying that $Y_t \indep (X_t, \tilde{X}_t) \mid Z_t$. This completes the proof.
\end{proof}

\begin{proof}[Proof of Lemma~\ref{lemma:R_bet}]
Observe that the predictive model $\hat{f}_t$ is a fixed function given $\mathcal{F}_{t-1}$, as it is fitted to $\{(X_s,Y_s,Z_s)\}_{s=1}^{t-1}$.
Thus, we can invoke Lemma~\ref{lemma:swap_ZY}, implying that under the null
\begin{equation}
    \label{eq:swap_g}
    g(q_t, \tilde{q}_t) \mid \mathcal{F}_{t-1} \overset{d}{=} g( \tilde{q}_t, q_t) \mid \mathcal{F}_{t-1},
\end{equation}
as $\hat{f}_t$ is a fixed function.
Now, recall that $g(\cdot)$ is an antisymmetric function, i.e.,
\begin{equation}
    \label{eq:anti_sym}
    g(q_t,\tilde{q}_t) = -g(\tilde{q}_t,q_t),
\end{equation} 
and observe that by combining~\eqref{eq:swap_g} and~\eqref{eq:anti_sym} we get the following 
\begin{equation}
    \label{eq:sym_pdf}
        g(q_t, \tilde{q}_t) \mid \mathcal{F}_{t-1} \overset{d}{=} -g(q_t,\tilde{q}_t) \mid \mathcal{F}_{t-1}.
\end{equation}
This implies that, under the null, the density function of $g(q_t, \tilde{q}_t) \mid \mathcal{F}_{t-1}$ is symmetric about $0$, and therefore
\begin{equation}
    \label{eq:expect_0}
    \mathbb{E}_{H_0}[g(q_t,\tilde{q}_t)\mid \mathcal{F}_{t-1}] = \mathbb{E}_{H_0}[W_t\mid \mathcal{F}_{t-1}] = 0.
\end{equation}

\end{proof}

\begin{proof}[Proof of Theorem~\ref{th:validity}]
Note that $S_0=1$ and $S_t$ in \eqref{eq:test_martingale} is non-negative for all $t=\mathbb{N}$ by construction. According to Ville's inequality \eqref{eq:ville}, it is enough to show that $\{S_t : {t \in \mathbb{N}_0}\}$ is a supermartingale under $H_0$ with respect to the filtration $\{\mathcal{F}_{t-1} : {t \in \mathbb{N}}\}$. This statement holds true since
\begin{align}
    \mathbb{E}_{H_0}[S_t \mid \mathcal{F}_{t-1}] &= \mathbb{E}_{H_0}\bigg[\int_0^1 \prod_{s=1}^t (1 + v \cdot W_j) \cdot h(v) dv \mid \mathcal{F}_{t-1}\bigg] \\
    &= \int_0^1 \prod_{j=1}^{t-1} (1 + v \cdot W_j)\cdot  \mathbb{E}_{H_0}[1 + v \cdot W_t\mid \mathcal{F}_{t-1}] \cdot h(v) dv \\
    &= \int_0^1 \prod_{j=1}^{t-1} (1 + v \cdot W_j)\cdot  (1+v\cdot\mathbb{E}_{H_0}[W_t \mid \mathcal{F}_{t-1}]) \cdot h(v) dv\\
&= \int_0^1 \prod_{j=1}^{t-1} (1 + v \cdot W_j) \cdot h(v) dv =S_{t-1}.
\end{align}
The second equality is due Tonelli's theorem~\cite{tonelli}, as $\prod_{j=1}^t (1 + v \cdot W_j) \cdot h(v)$ is non-negative for all $t \in \mathbb{N}$, and $h(v)$ is a probability density function. The last equality holds by invoking Lemma~\ref{lemma:R_bet}.
\end{proof}

\section{THE OFFLINE CONDITIONAL RANDOMIZATION TEST}
\label{supp:crt}
\begin{algorithm}[H]
\textbf{Input}: Data $\{(X_i,Y_i,Z_i)\}_{i=1}^n$; test statistic $T(\cdot)$; number of iterations $M$.

\begin{algorithmic}[1]
\STATE Set $t \leftarrow T(\{(X_i,Y_i,Z_i)\}_{i=1}^n)$
\FOR{$m = 1,\dots,M$}
\STATE  Sample dummy variables $\tilde{X}_i \sim P_{X \mid Z}(X_i \mid Z_i)$ for $i=1,\dots,n$
\STATE  Set $\tilde{t}^{(m)} \leftarrow T(\{\tilde{X}_i,Y_i,Z_i)\}_{i=1}^n)$ 
\ENDFOR
\end{algorithmic}
\textbf{Output}: A p-value $p_n = \frac{1}{1+M}\big({1 + \sum_{m=1}^M{\mathds{1}\{\tilde{t}^{(m)} \leq t\}}}\big)$.
\vspace{0.1cm}
\caption{Offline Conditional Randomization Test}
\label{alg:CRT}
\end{algorithm}

\section{SUPPLEMENTARY DETAILS ON THE PROPOSED METHOD}
\label{supp:proposed}

\subsection{Validity of the Base Martingale}
\label{supp:stv}
In Section~\ref{sec:formulating} we formulate the base martingale $S_t^v$ in~\eqref{eq:test_martingale_v}. Here, we prove in Proposition~\ref{prop:valid_stv} that $\{S_t^v : t \in \mathbb{N}_0\}$ is a valid test martingale, according to Definition~\ref{def:test_martingale}, for any $v \in [0,1]$.

\begin{prop}
\label{prop:valid_stv}
The base martingale $\{S_t^v : t \in \mathbb{N}_0\}$~\eqref{eq:test_martingale_v} is a valid test martingale w.r.t the filtration $\{\mathcal{F}_{t-1} : {t \in \mathbb{N}}\}$, i.e., satisfying Definition~\ref{def:test_martingale}, for any constant $v \in [0,1]$.
\end{prop}

\begin{proof}
Let $v \in [0,1]$. Note that $S_0^v=1$ and $v \cdot W_t \geq 0$ for any $t \geq 1$, hence $S_t^v \geq 0, \forall t \in \mathbb{N}_0$ by construction. To conclude the prof, we show that $\{S_t^v : t \in \mathbb{N}_0\}$ is a supermartingale under the null with respect to $\{\mathcal{F}_{t-1} : {t \in \mathbb{N}}\}$:
$$
\mathbb{E}_{H_0}[S_t^v \mid \mathcal{F}_{t-1}] = \prod_{j=1}^{t-1}(1 + v \cdot W_j) \cdot \mathbb{E}_{H_0}[1 + v \cdot W_t \mid \mathcal{F}_{t-1}]=\prod_{j=1}^{t-1}(1 + v \cdot W_j) \cdot (1 + v \cdot \mathbb{E}_{H_0}[W_t \mid \mathcal{F}_{t-1}]) = S_{t-1}^v.
$$
\end{proof}

\subsection{Further Details on the Synthetic Experiment from Section~\ref{sec:formulating}}
\label{supp:imp_fig2}



Here, we conduct an experiment that visualizes the advantage of the mixture-method $S_t$~\eqref{eq:test_martingale} over a constant $v$ in $S_t^v$~\eqref{eq:test_martingale_v}. To this end, we present in Figure~\ref{fig:density_plots} the wealth process $S_t^v$ obtained for several values of $v$ and the mixture-method test martingale $S_t$~\eqref{eq:test_martingale}, for two different realizations of the \texttt{non-null data} generating model from Figure~\ref{fig:fig_size_legend}. We construct $S_t$ by applying Algorithm~\ref{alg:practical_eCRT} to the generated data, with the choice of $h(v)$ as the pdf of the uniform distribution on $[0,1]$. The base test martingales $S_t^v$ are constructed in the same fashion but with~\eqref{eq:test_martingale_v} in line 8 of Algorithm~\ref{alg:practical_eCRT} instead of the mixture approach. Below, we provide the implementation details of Algorithm~\ref{alg:practical_eCRT} for this experiment.
\begin{itemize}
    \item We set the betting score function in~\eqref{eq:lambda_t} to be $g(a,b) = \textrm{sign}(b-a)$.
    \item The online learning model for $\hat{f}_t$ takes the form of lasso regression using the hyper-parameter tuning approach described in Section~\ref{sec:practical}; we trained $L=20$ models, each corresponds to a different $\eta$, where the number of samples for initial training is set to be $n_{\text{init}}=20$.
    \item  The test statistic function is the mean squared error of a given batch $T(\{(X_s,Y_s,Z_s)\}_{s=1}^b;\hat{f}) = \frac{1}{b} \sum_{s=1}^b (\hat{f}(X_s)-Y_s)^2$, where we use a batch size of $b=5$.
    \item  We set the de-randomization parameter $K$, described in Section~\ref{sec:practical}, to be equal to $20$.
\end{itemize}

Although the data distribution is identical in both cases, the wealth processes presented in Figure~\ref{fig:density_plots} behave very differently: the left panel portrays that the best-performing constant is $v^*=0.3$, whereas the right panel indicates that $v^*=0.7$. Importantly, observe how the martingale $S_t^{v^*}$ grows exponentially with $t$, and thus has a strong traction force on the average martingale $S_t$. Observe also how the effect of the growing base martingales on $S_t$ is stronger than that of the ones that do not grow with $t$. This demonstrates the advantage of the mixture-method $S_t$~\eqref{eq:test_martingale} over the base martingale with a constant $v$ in $S_t^v$~\eqref{eq:test_martingale_v}.

\begin{figure}
     \centering
     \begin{subfigure}[b]{0.49\columnwidth}
         \centering
         \includegraphics[width=\textwidth]{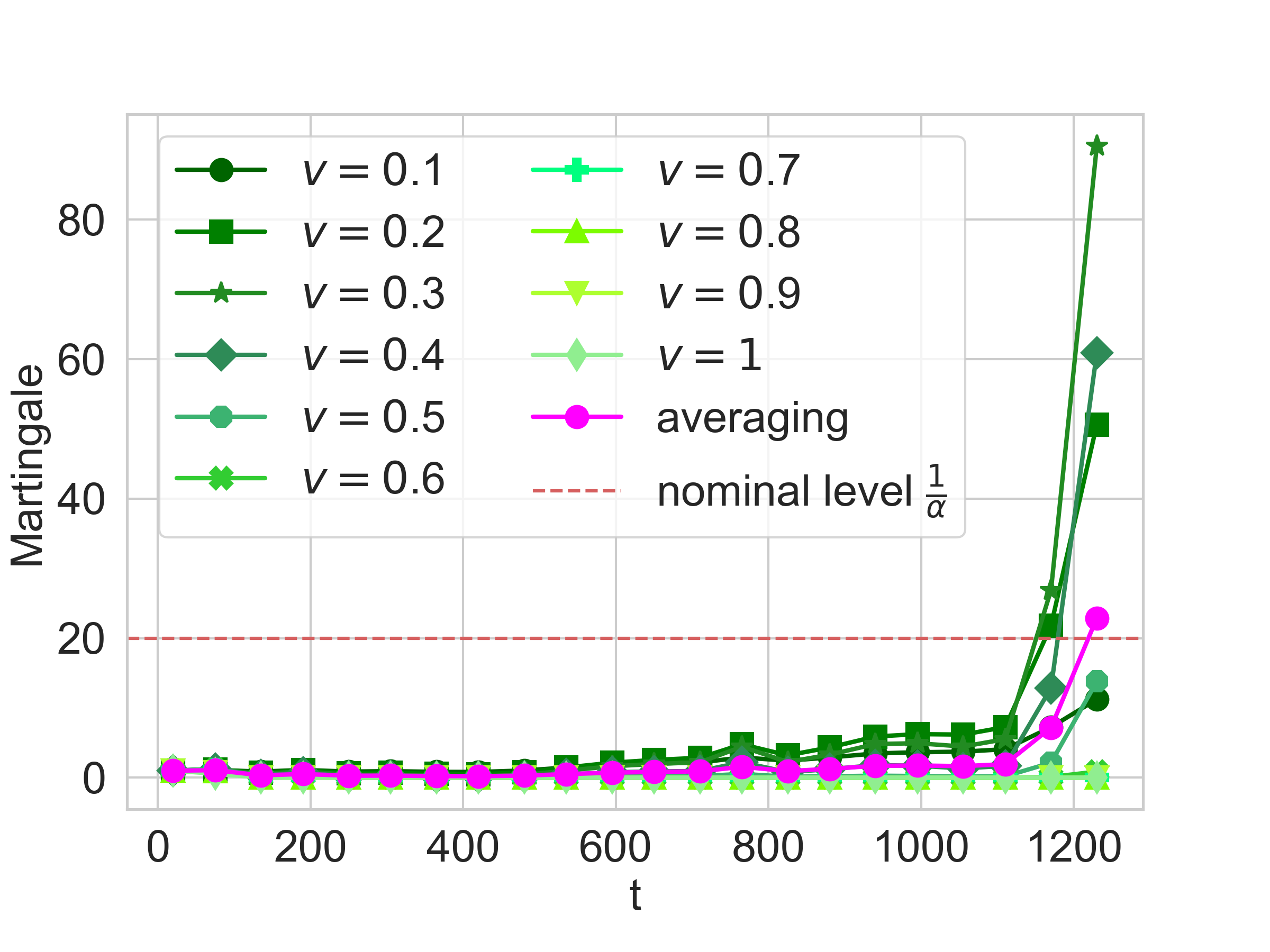}
         \label{fig:density_seed_4}
     \end{subfigure}
     \hfill
     \begin{subfigure}[b]{0.49\columnwidth}
         \centering
         \includegraphics[width=\textwidth]{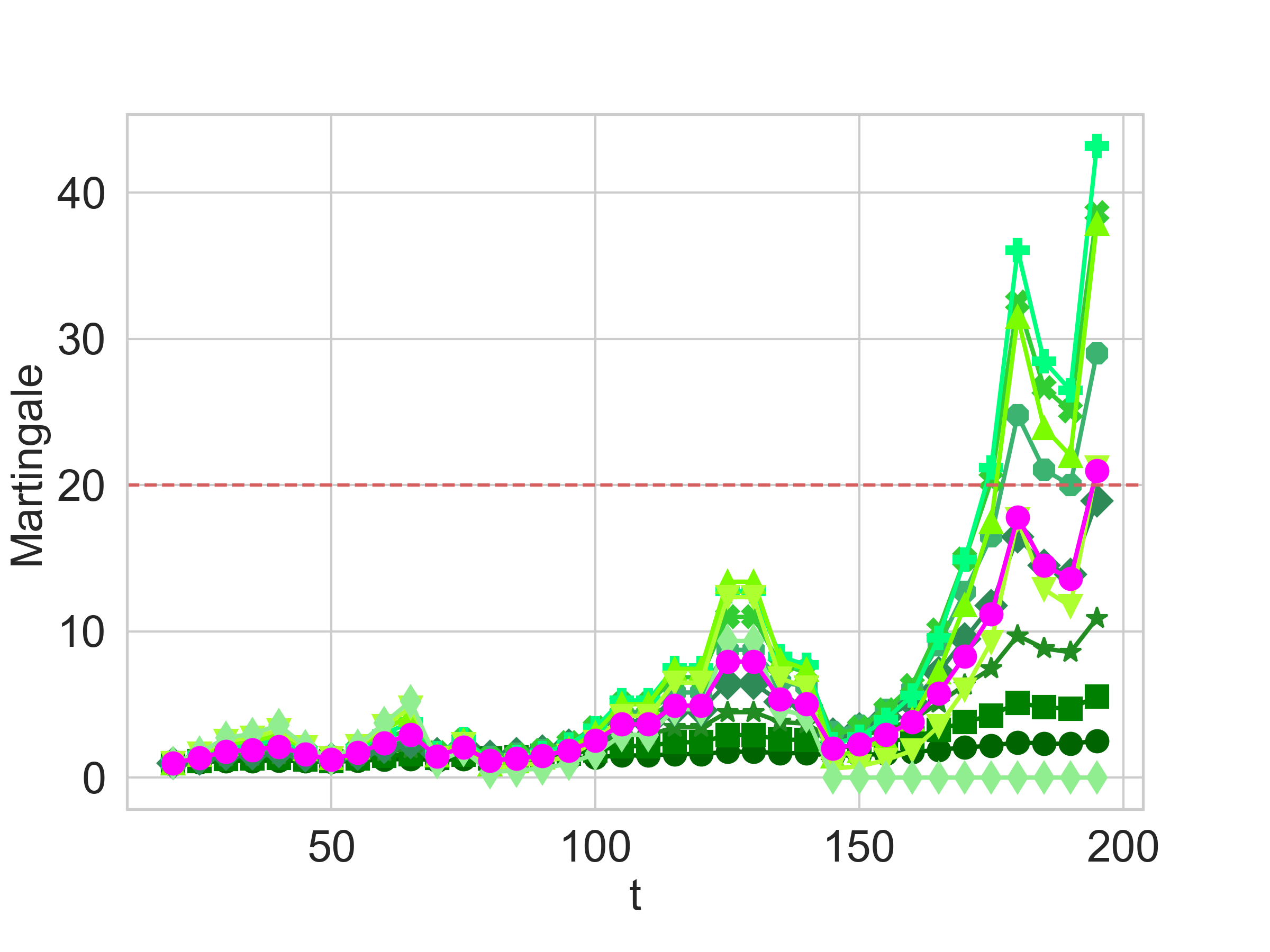}
         \label{fig:density_seed_5}
     \end{subfigure}
        \caption{\textbf{The effect of the amount of toy money $v$ we risk on the wealth, in comparison with the mixture-method martingale.} Each of the two panels corresponds to a different realization of the \texttt{non-null data}. The green curves represent $S_t^v$ in~\eqref{eq:test_martingale_v} with different choices of $v$. The magenta curve represents the average test martingale $S_t$ in~\eqref{eq:test_martingale}.}
        \label{fig:density_plots}
\end{figure}



\section{SUPPLEMENTARY DISCUSSION ON RELATED WORK}
\label{supp:grunwald}

The contemporary work by \cite{grunwald2022anytime} also offers an approach to test for CI sequentially based on test martingales.  Although \cite{grunwald2022anytime} test martingale shares similarities with the method proposed in this work, there are several key differences. \cite{grunwald2022anytime} martingale can be conceptualized as a likelihood ratio process and it involves integration over $P_{X \mid Z}$. By contrast, ours resembles the knockoffs approach which measures differences of a test statistic evaluated on the original and dummy triplets, and is valid due to the anti-symmetry of the betting score~\eqref{eq:lambda_t}. In terms of power, under the assumption of a fast converging estimator for $P_{Y \mid X, Z}$, \cite{grunwald2022anytime} prove their test martingale has a growth-rate optimality \citep{grunwald2020safe}. On the other hand, following Proposition~\ref{prop:power_1}, we merely assume a weaker assumption to achieve power one asymptotically (should not be confused with growth-rate optimality): the model should distinguish the original and dummy triplets \emph{on average}. Given the aforementioned variations, one may opt for the technique put forth by \cite{grunwald2022anytime} if there is a reasonable estimate for the conditional distribution of $Y \mid X, Z$. On the other hand, our method may be a more suitable choice if greater flexibility is desired for designing the machine learning model.

To support the above discussion, we compare the two methods based on simulated data. We focus on binary classification since \cite{grunwald2022anytime} present a concrete test martingale in this context. To this end, we
generate $X_t,Z_t$ akin to Section~\ref{sec:exp_setup}. The binary response $Y_t$ is generated from a Bernoulli distribution with a probability obtained by applying the sigmoid function to $c \cdot \varphi(X_t,Z_t)$, where we conduct two experiments. In the first, we choose $\varphi$ to be a linear function and $c=1$, whereas in the second we choose $\varphi$ to be a non-linear function and set $c=0.8$. We deploy the method of \citet{grunwald2022anytime} as described in \citep{grunwald2022anytime}[Section~3.3], and our e-CRT with $K=20$, $\mathcal{B}=\{2,5,10\}$, $g(a,b)=\textrm{tanh}(20 \cdot (b-a)/\textrm{max}\{a,b\})$, and $T(\cdot)$ be the binary cross entropy loss. We use a neural network classifier for both testing methods. 

We begin with a linear case, demonstrating a scenario where a reasonable estimation of $P_{Y \mid X,Z}$ can be achieved. In this experiment, we set $\varphi(X_t,Z_t) = \beta^\top Z_t + 3 \cdot X_t$ with $\beta \sim \mathcal{N}(0,I_d)$ and $d=5$. Figure~\ref{fig:grun_pow_lin} presents the empirical power of both methods evaluated on 100 realizations of the described data. As can be seen, the method of \cite{grunwald2022anytime} tends to be more powerful than ours, indicating the advantage of \citet{grunwald2022anytime} in cases where a good estimation of $P_{Y \mid X,Z}$ is attainable. 

Next, we consider a more complex non-linear interaction model, where $\varphi(X_t,Z_t) = \beta^\top Z_t + 6 \cdot X_t \cdot |Z_t^{(1)}|  \cdot 
|Z_t^{(2)}|$. Here, $\beta \sim \mathcal{N}(0,I_d)$ and $d=10$, where $Z_t^{(j)}$ represents the $j$th covariate of $Z_t$. The empirical power of the methods, evaluated on 100 realizations of the data, is depicted in Figure~\ref{fig:grun_pow}. As portrayed, our e-CRT performs better in this case, indicating the superiority of our e-CRT when it is harder to attain a fast converging estimator for $P_{Y \mid X, Z}$. 

Lastly, Figure~\ref{fig:grun_err} presents the type-I error of both e-CRT and the method proposed by \cite{grunwald2022anytime}, evaluated on 100 realizations of the data with $\varphi(X_t,Z_t) = \beta^\top Z_t$ and $c=0.8$ where $\beta \sim \mathcal{N}(0,I_d)$ and $d=10$. There, the type-I error is controlled for both methods.

\begin{figure}[t]
\centering
     \begin{subfigure}[b]{0.33\columnwidth}
         \centering
         \includegraphics[width=\textwidth]{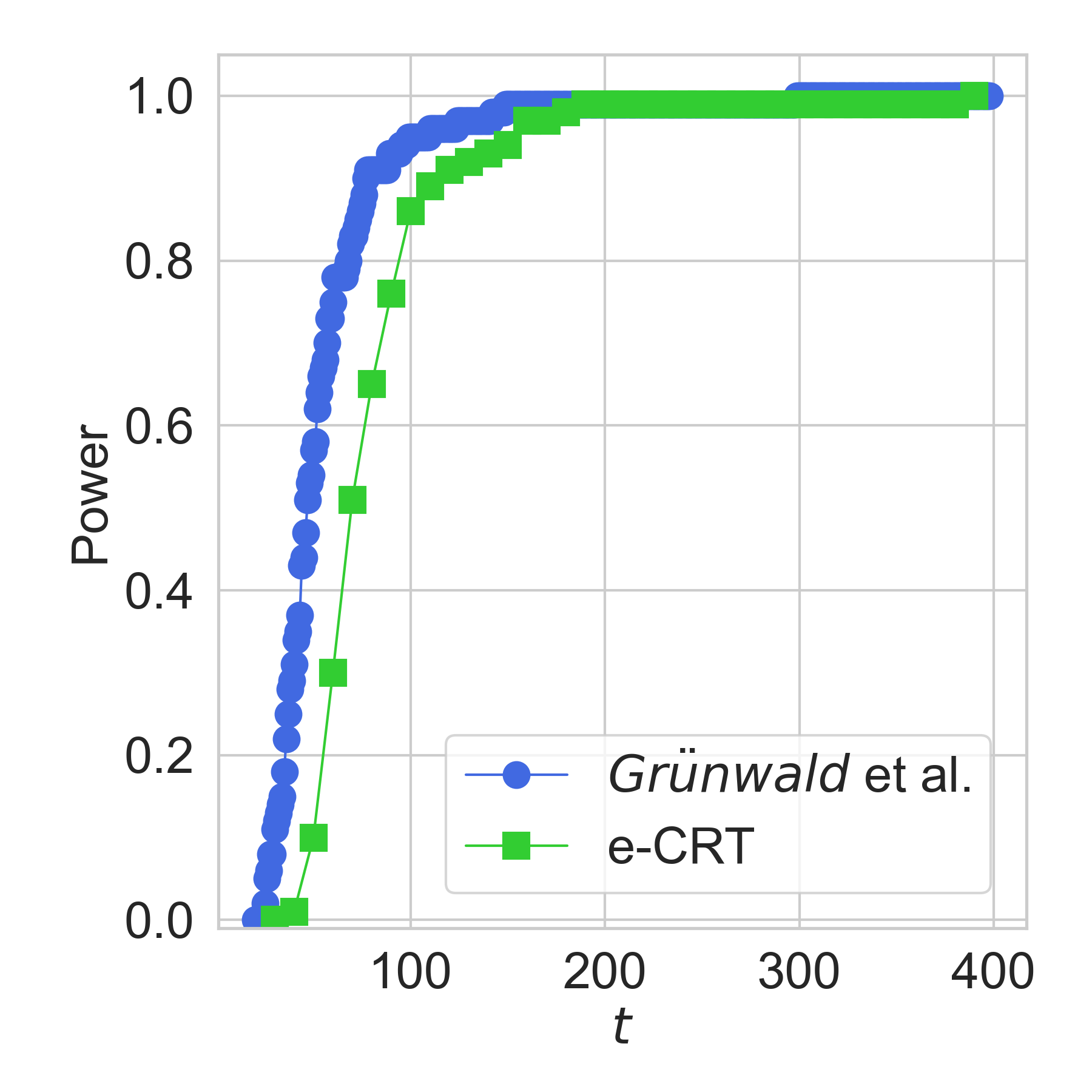}
         \caption{}
         \label{fig:grun_pow_lin}
     \end{subfigure}
     \hfill
     \begin{subfigure}[b]{0.33\columnwidth}
         \centering
         \includegraphics[width=\textwidth]{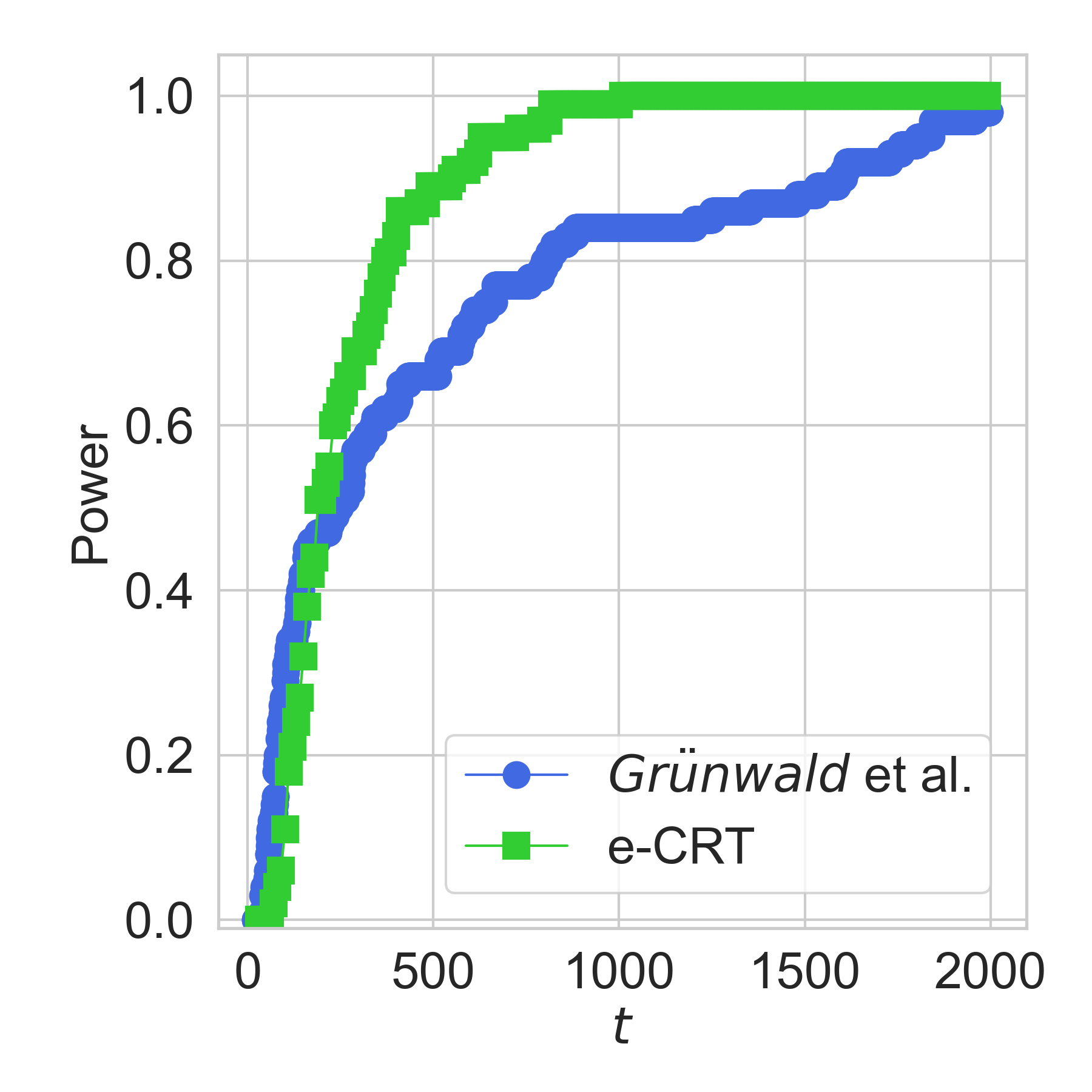}
         \caption{}
         \label{fig:grun_pow}
     \end{subfigure}
     \hfill
     \begin{subfigure}[b]{0.33\columnwidth}
         \centering
         \includegraphics[width=\textwidth]{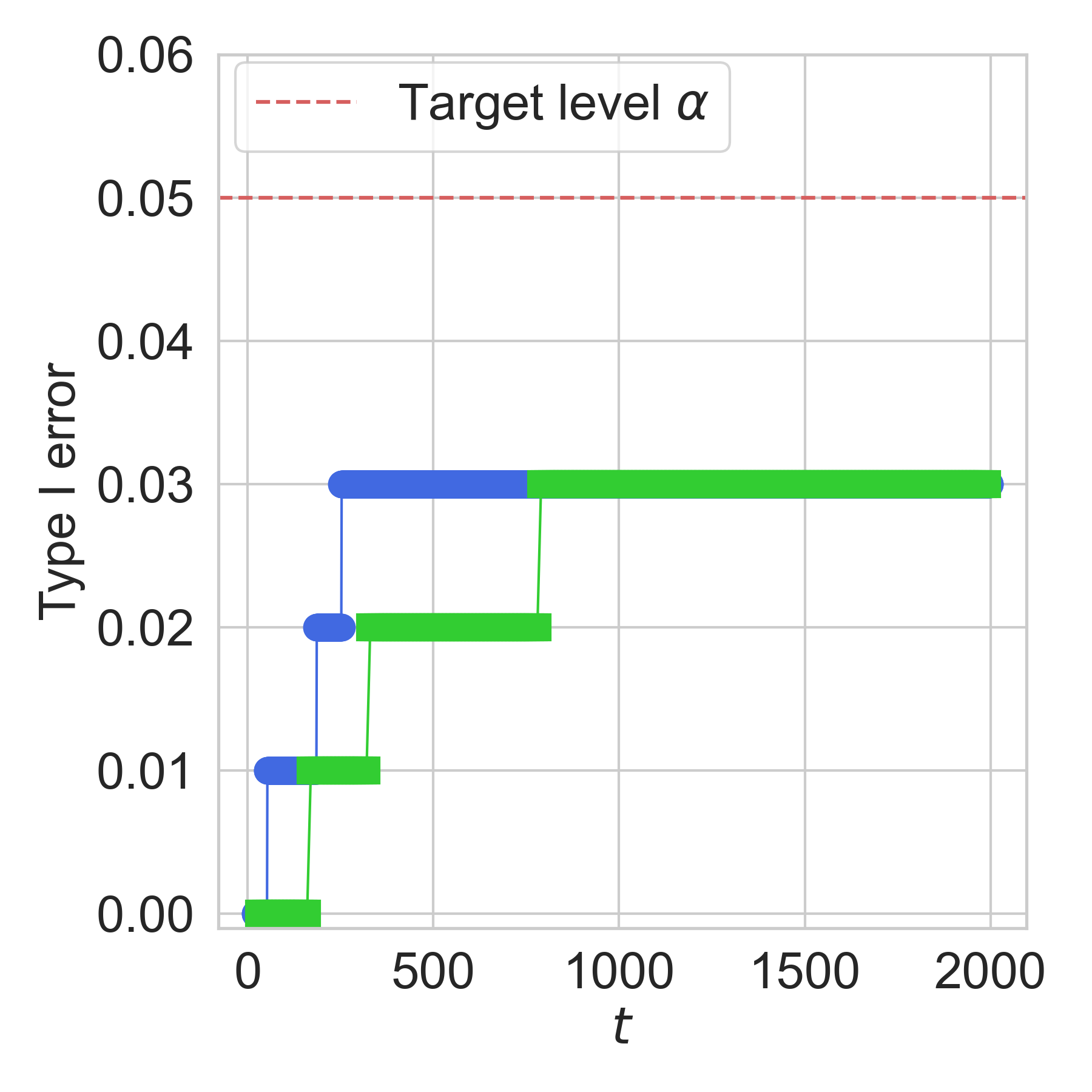}
         \caption{}
         \label{fig:grun_err}
     \end{subfigure}
\caption{\textbf{Empirical power and type-I error rate of e-CRT and the method proposed by \cite{grunwald2022anytime}}, evaluated on $100$ realizations of the data. (a) empirical power in a linear case, where $\varphi(X_t,Z_t) = \beta^\top Z_t + 3 \cdot X_t$. (b) empirical power in a non-linear case, where $\varphi(X_t,Z_t) = \beta^\top Z_t + 6 \cdot X_t \cdot |Z_t^{(1)}| \cdot |Z_t^{(2)}|$. (c) type-I error rate with $\varphi(X_t,Z_t) = \beta^\top Z_t$.}
\label{fig:grun_comp}
\end{figure}

\section{ONLINE LEARNING WITH AUTOMATIC HYPER-PARAMETER TUNING}
\label{supp:online}
Recall that after we place a bet for a new test point $(X_t,Y_t,Z_t)$, we update the predictive model $\hat{f}_t$ using $(X_t,Y_t,Z_t)$ by leveraging online learning techniques and get $\hat{f}_{t+1}$. The online update is done sequentially and thus computationally efficient. For example, in our experiments we use lasso regression model, minimizing
\begin{equation}
\label{eq:lasso}
    \hat{\beta}_t:=\arg\min_\beta \ \frac{1}{t}\sum_{s=1}^t (X_s^{\top}\beta - Y_s)^2 + \eta\|\beta\|_1,
\end{equation}
where $\eta$ is a hyper-parameter that controls the regularization strength. The above optimization problem is convex and it is minimized using an iterative solver \cite{coordinate_descent,ADMM}. To form a computationally efficient learning algorithm, at each step $t$ we initialize the iterative solver with the previous $\hat{\beta}_{t-1}$ and update the regression coefficients by applying a few additional steps with the squared error term in \eqref{eq:lasso} that includes the new observed point. To obtain a powerful predictive model we should tune the hyper-parameter $\eta$, but tuning this parameter via standard cross-validation may break the sequential update of $\hat{\beta}_t$. As a way out, we apply the following procedure for tuning $\eta$.
We train \emph{online} a series of $L$ models $\hat{f}^{l}_{t_{\text{tr}}}$ on $\{(X_s, Y_s, Z_s)\}_{s<t_{\text{tr}}}$ over a grid of possible values of $\eta_l$, $l=1\dots,L$, where $t_{\text{tr}} < t$, and evaluate the models on the $t - t_{\text{tr}}$ recent holdout points. Next, we update the running model $\hat{f}_t$ by minimizing \eqref{eq:lasso} with $\eta_{l^*}$, where $l^*$ is the index of the model $\hat{f}^{l}_{t_{\text{tr}}}$ that achieves the smallest prediction error; this is done by applying a few steps of any iterative solver, initialized with the previous $\hat{\beta}_{t-1}$.

\section{SUPPLEMENTARY DETAILS ON ENSEMBLE OVER BATCHES}
\label{supp:ensemble_batches}

In Section~\ref{sec:practical} we present our approach of ensemble of batch martingales $S_{t,b}$. Each of $S_{t,b}$ is evaluated on a batch of size $b$ instead of a single data point. In more detail,
\begin{equation}
    \label{eq:batch_martingale}
    S_{t,b}=\int_0^1 S_{t,b}^v \cdot h(v) dv, \ \text{with} \ \ S_{t,b}^{v} :=\prod_{s=1}^{\lfloor{t/b}\rfloor} (1 + v \cdot W_{s,b}),
\end{equation}
and $\lfloor{\cdot}\rfloor$ is the floor function. Above, we use the convention that $\prod_{s=1}^{0} (1 + v \cdot W_{s,b}) = 1$. The betting score $W_{s,b}$ in~\eqref{eq:batch_martingale} is evaluated similarly to \eqref{eq:lambda_t} but on a batch via the following test statistic function $$q_{s,b} = T(\{(X_j,Y_j,Z_j)\}_{j=(s-1)\cdot b+1}^{s \cdot b}; \hat{f}_{(s-1)\cdot b + 1}) \in \mathbb{R},$$
that operated on the original batch of triplets; analogously, $\tilde{q}_{s,b}$ is obtained by invoking the same $T$ function, however on the dummy triplets, resulting in $W_{s,b} = g({q}_{s,b},\tilde{q}_{s,b})$.
For example, $T$ can be a function returning the mean squared error of $\hat{f}_{(s-1)\cdot b + 1}$ evaluated on the observed data. Importantly, the test martingale in~\eqref{eq:batch_ensemble} is valid since
$$
\mathbb{E}_{H_0}\left[ S_t \mid \mathcal{F}_{t-1}\right]=\mathbb{E}_{H_0}\left[ \frac{1}{|\mathcal{B}|}\sum_{b\in \mathcal{B}}{S_{t,b}} \mid \mathcal{F}_{t-1}\right] =\frac{1}{|\mathcal{B}|}\sum_{b\in\mathcal{B}}{\mathbb{E}_{H_0}[S_{t,b} \mid \mathcal{F}_{t-1}]} = \frac{1}{|\mathcal{B}|}\sum_{b\in\mathcal{B}}{S_{t-1,b}} = S_{t-1},
$$
where the third equation is because each of the batch martingales $S_{t,b}$ is valid.

\section{SUPPLEMENTARY DETAILS ON SYNTHETIC EXPERIMENTS}
\label{supp:synth_exp}

\subsection{Implementation Details}
\label{supp:imp_details}
In Section~\ref{sec:exp_setup} we describe the experimental setup of our synthetic experiments. Here we provide the implementation details of e-CRT and the baseline methods---CRT and HRT.

We implement the e-CRT procedure as described in Algorithm~\ref{alg:practical_eCRT}, with the following choices. We set the betting score function~\eqref{eq:lambda_t} to be $g(a,b) = \textrm{sign}(b-a)$. The martingale in~\eqref{eq:test_martingale} is evaluated by choosing $h(v)$ as the pdf of the uniform distribution on $[0,1]$. The online learning model for $\hat{f}_t$ takes the form of lasso regression using the hyper-parameter tuning approach described in Section~\ref{sec:practical}; we trained $L=20$ models, each corresponds to a different $\eta$, where the number of samples for initial training is set to be $n_{\text{init}}=20$. The test statistic function is the mean squared error of a given batch, defined as $T(\{(X_s,Y_s,Z_s)\}_{s=1}^b;\hat{f}) = \frac{1}{b} \sum_{s=1}^b (\hat{f}(X_s)-Y_s)^2$, where we apply the batch-ensemble approach with $\mathcal{B} = \{2,5,10\}$ in~\eqref{eq:batch_ensemble}. Lastly, we set the de-randomization parameter $K$, described in Section~\ref{sec:practical}, to be equal to $20$.

The machine learning method we use in both CRT and HRT is a 5-fold cross-validated lasso regression algorithm. As for the ADDIS-Spending approach~\cite{tian2021online} for adjusting CRT and HRT, we use the software package available at \url{https://github.com/jinjint/onlineFWER}, with the default parameters.

\subsection{Type-I Error of the Offline CRT and HRT}
\label{supp:type_1}
In Section~\ref{sec:exp_power_err} we present the empirical power of e-CRT compared to CRT, HRT and out-of-the-box sequential versions of them ADDIS-CRT and ADDIS-HRT, evaluated on simulated data. There, we present the type-I error only for the sequential tests: e-CRT, ADDIS-CRT and ADDIS-HRT. Here, we present in Figure~\ref{fig:type_1_crt_hrt} the type-I error of CRT and HRT evaluated on the same data as in Section~\ref{sec:exp_power_err}. Importantly, the presented type-I error is evaluated by treating the data at each presented time step as a fixed size dataset.

\begin{figure}[H]
\begin{centering}
\includegraphics[width=0.35\textwidth]{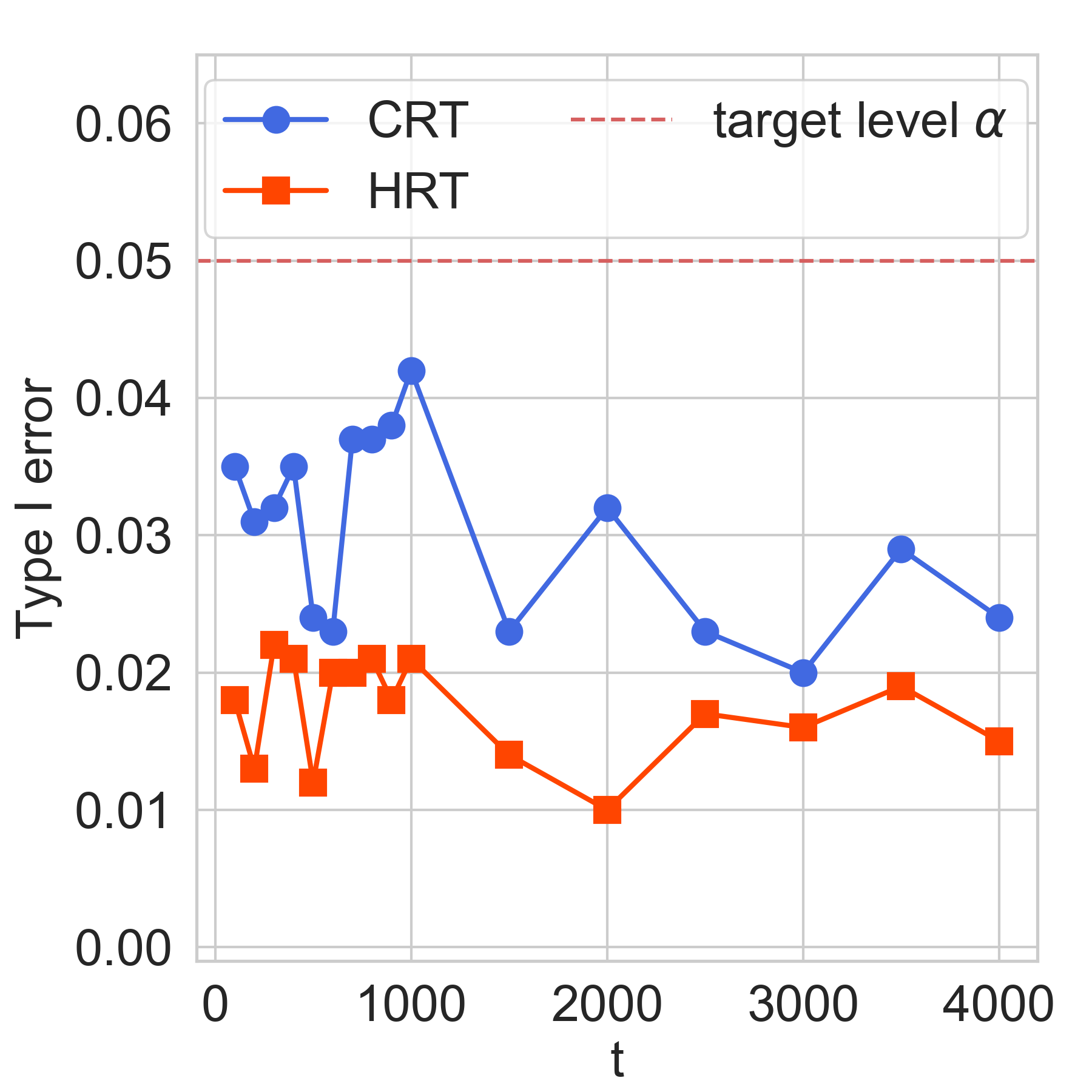}
\par\end{centering}
\vspace{-2mm}
\caption{Type-I error of CRT and HRT evaluated over 1000 realizations of the \texttt{null data} model. Other details are as in Figure~\ref{fig:comparison_graphs}.}
\label{fig:type_1_crt_hrt}
\vspace{-2mm}
\end{figure}

\subsection{Additional Synthetic Experiment with Varying Number of Covariates}
\label{supp:varying_d}
In this section we evaluate the performance of e-CRT as a function of the number of covariates $d$. To do so, we follow the data generation process described in Section~\ref{sec:exp_setup} and sample $n=1000$ data points of different dimensions $d$. Then, we apply the e-CRT to each data set and we also apply CRT and HRT on the whole generated data (i.e., only once) to serve as baseline for reference. Figure~\ref{fig:d_plots} presents the empirical power and the type-I error as a function of the number of covariates $d$. It can be seen that the type-I error is controlled for all $d$, and the empirical power is decreased as we increase the dimension $d$. 

\begin{figure}[H]
     \centering
     \begin{subfigure}[b]{0.49\textwidth}
         \centering
         \includegraphics[width=\textwidth]{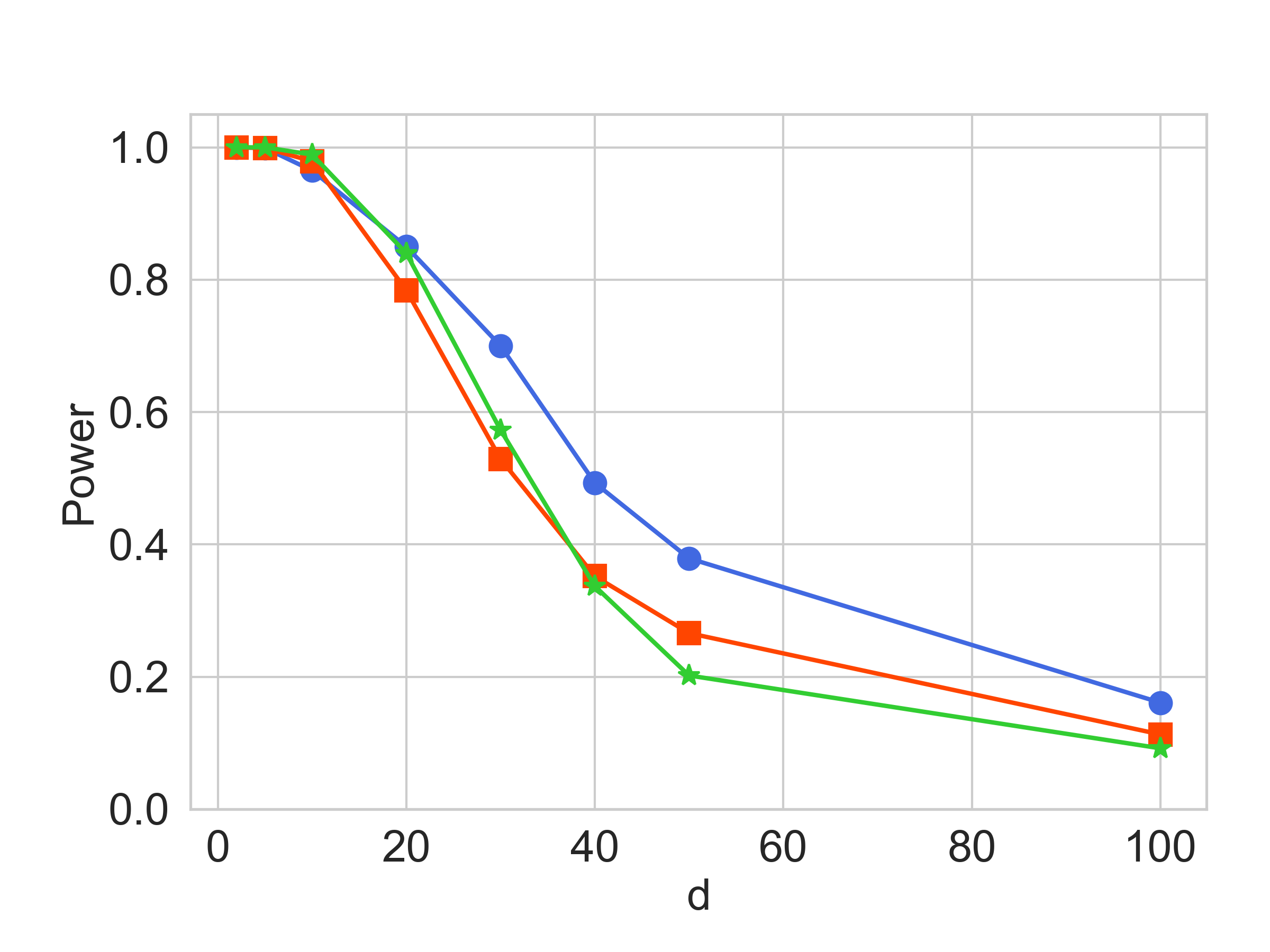}
         \caption{Power}
         \label{fig:power_vs_d}
     \end{subfigure}
     \hfill
     \begin{subfigure}[b]{0.49\textwidth}
         \centering
         \includegraphics[width=\textwidth]{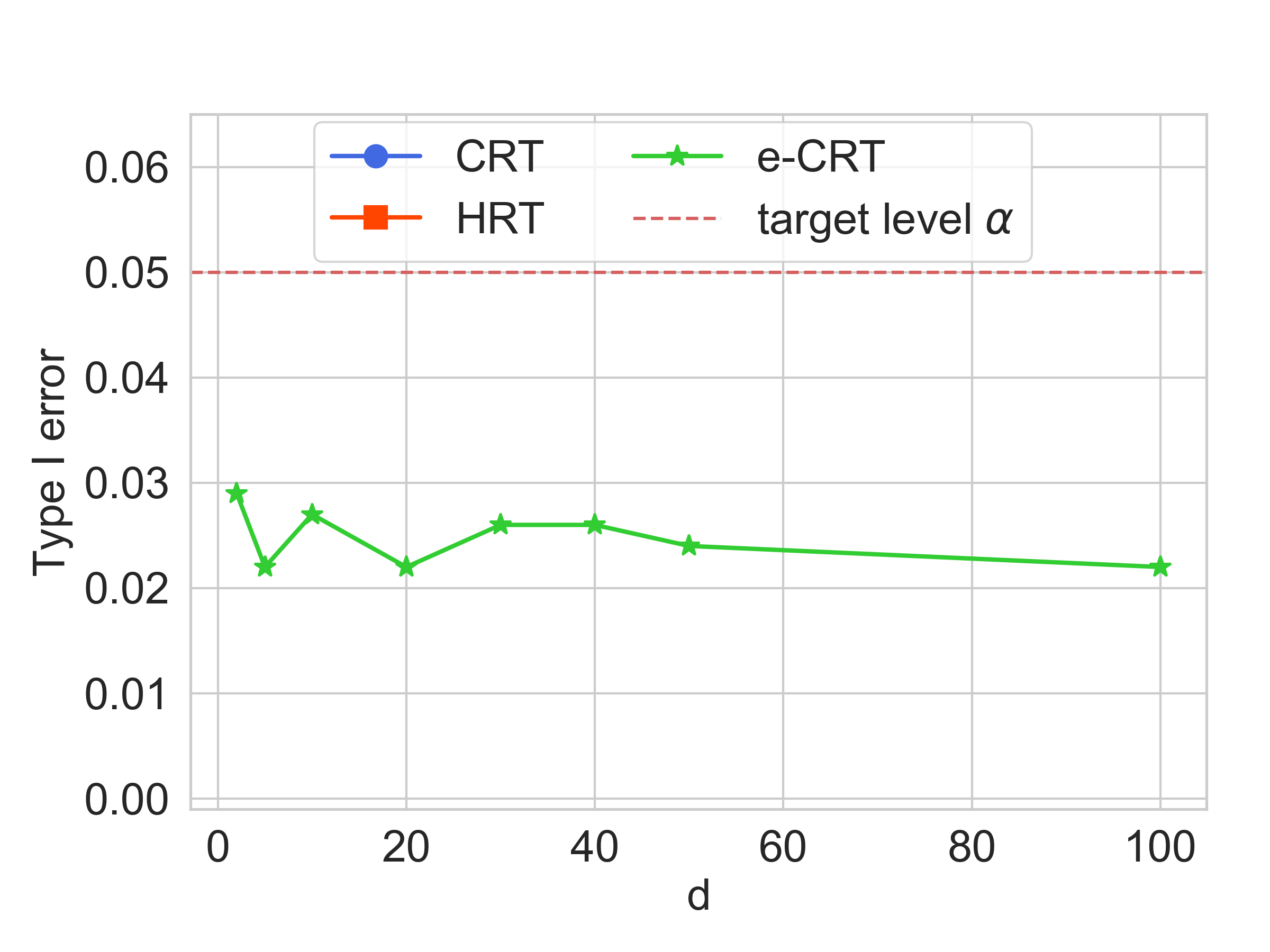}
         \caption{Type-I error}
         \label{fig:error_vs_d}
     \end{subfigure}
        \caption{\textbf{Empirical power and type-I error rate of e-CRT of level $\alpha=0.05$ as a function of number of covariates $d$.} Left: empirical power evaluated on $1000$ realizations of the \texttt{non-null data} model. Right: type-I error rate evaluated on $1000$ realizations of the \texttt{null data} model.}
        \label{fig:d_plots}
\end{figure}

\subsection{Additional Synthetic Experiment with Increasing Correlation Between the Features}
\label{supp:corr}

In this section, we study the effect of the dependency structure between $X$ and $Z$ on the performance of the proposed method. To this end, we sample $(X_t,Z_t) \in \mathbb{R}^{d}$ jointly from $\mathcal{N}(0,\Sigma)$, where $X_t$ is the first covariate of the generated $d$-dimension vector, and $Z_t$ are the rest $d-1$ covariates. We set the $(i,j)$ entry in the covariance matrix to be  $\Sigma_{i,j}=\rho^{|i-j|}$, where $\rho \in [0,1]$ is the auto-correlation parameter. The response $Y_t$ is generated the same way as in Section~\ref{sec:exp_setup}. To examine the impact of the dependency strength, we vary the auto-correlation parameter $\rho$ and apply the e-CRT to the generated data as described in Supplementary Section~\ref{supp:imp_details}. Figure~\ref{fig:corr} presents the empirical power and type-I error evaluated over 100 realizations of the data. Following that figure, one can see that the type-I error is controlled, as expected. One can also see that the power decreases as the correlation increases. This result aligns with previous analysis in the field of CI testing; see, for example, \cite{candes2018panning,dCRT,MRD}.

\begin{figure}
     \centering
     \begin{subfigure}[b]{0.49\textwidth}
         \centering
         \includegraphics[width=\textwidth]{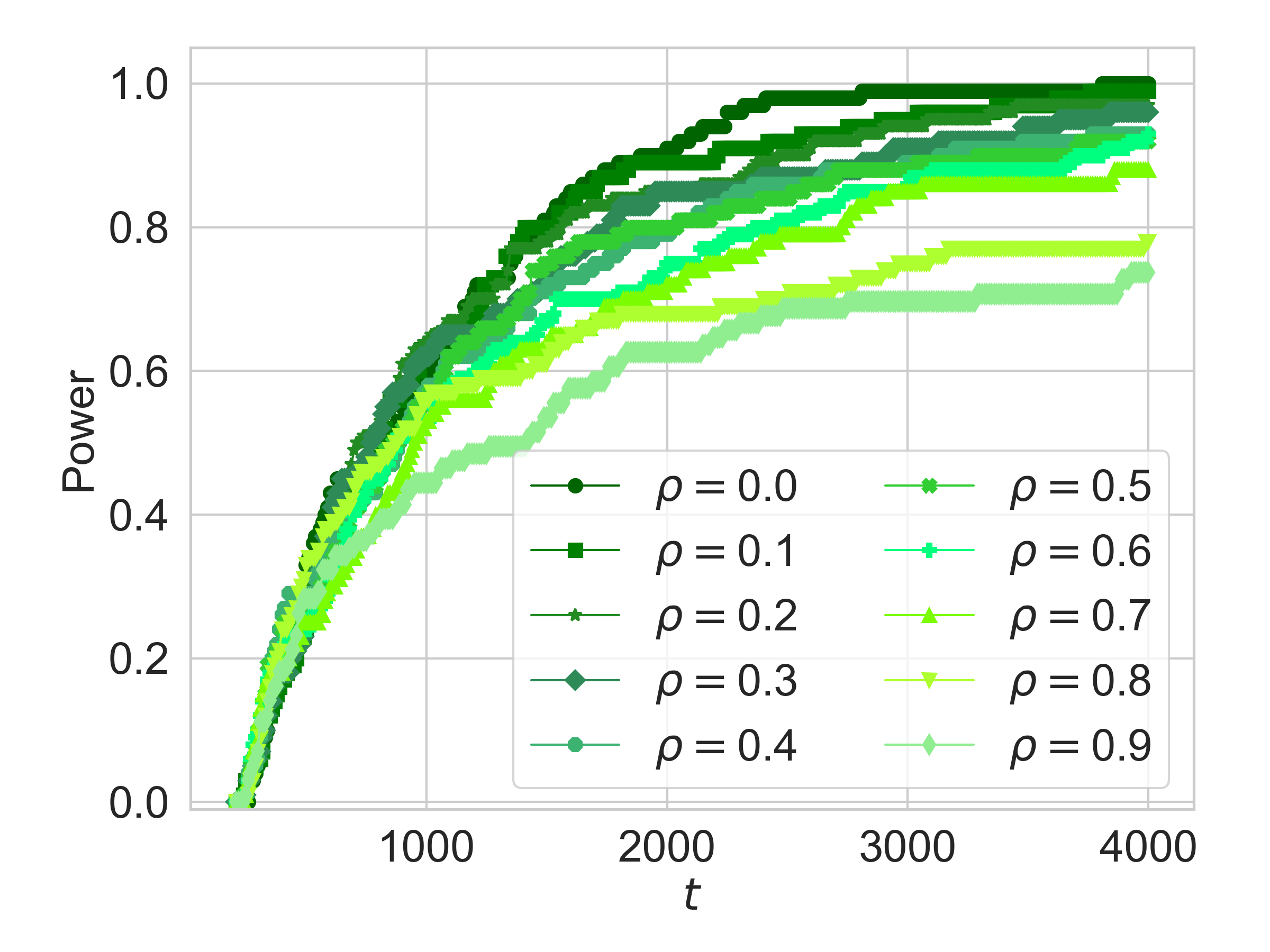}
         \caption{Empirical power}
         \label{fig:corr_pow}
     \end{subfigure}
     \hfill
     \begin{subfigure}[b]{0.49\textwidth}
         \centering
         \includegraphics[width=\textwidth]{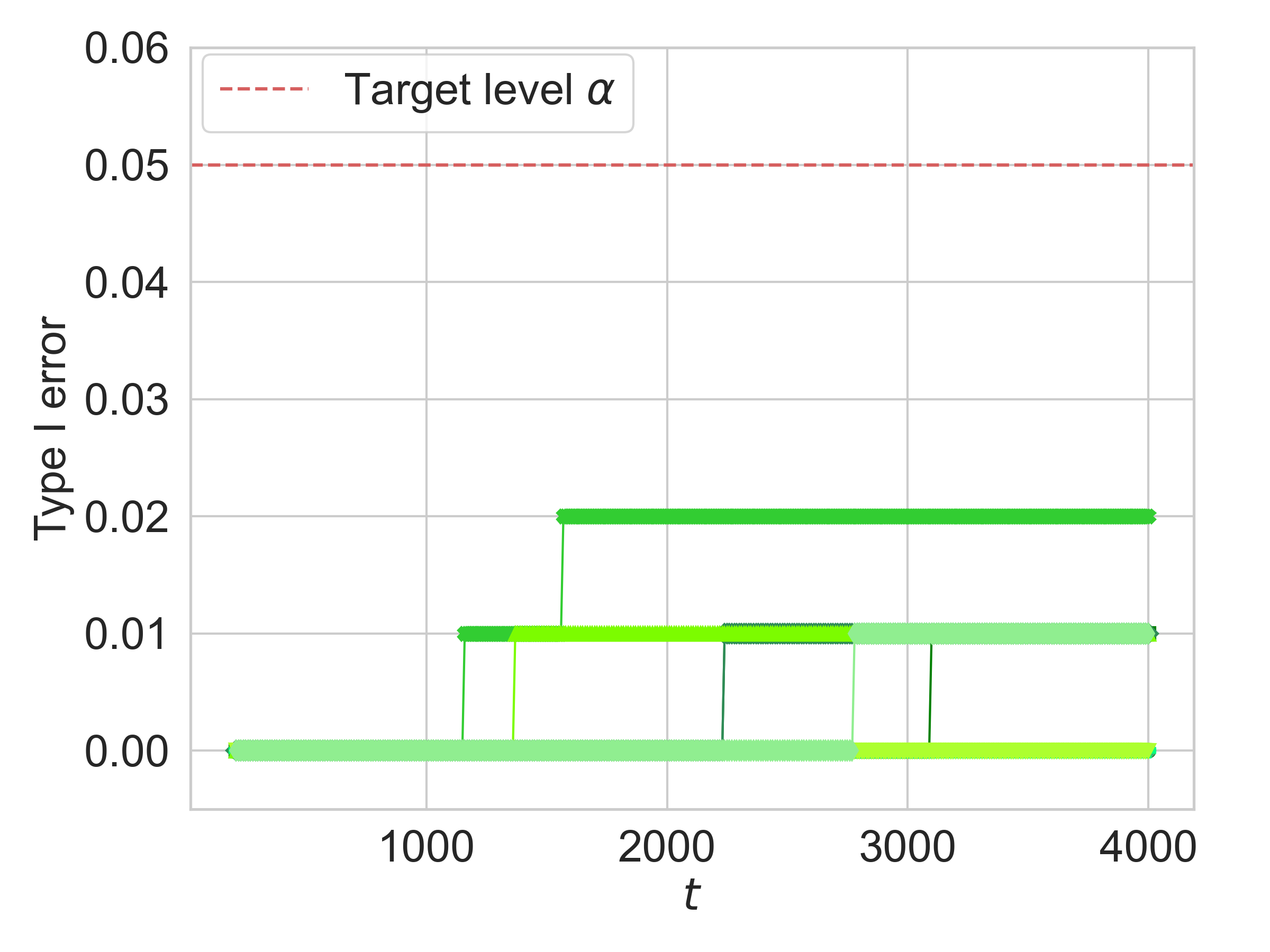}
         \caption{Type-I error}
         \label{fig:corr_err}
     \end{subfigure}
   \caption{\textbf{Empirical power and type-I error rate of e-CRT as a function of the auto-correlation parameter $\rho$.} The empirical power and type-I error evaluated over 100 realizations of the data.}
\label{fig:corr}
\end{figure}

\subsection{Additional Synthetic Experiment on the Robustness of e-CRT}
\label{supp:robustness}

To implement e-CRT, we must generate dummy features $\tilde{X}_t$ from $P_{X\mid Z}$. In practice, when this conditional distribution is unknown it should be estimated from data. In general, there is no formal type-I error control in this case. Therefore, it is important to study the robustness of e-CRT to errors in the estimation of $P_{X \mid Z}$.

\subsubsection{Parametric Estimation of $P_{X \mid Z}$}
\label{supp:parametric_est}
Here, we apply e-CRT to a sequence of data points generated from the \texttt{null data} model from Section~\ref{sec:exp_setup}, but instead of sampling $\tilde{X}_t$ from the true ${P}_{X\mid Z}$, we consider the following data distribution:
\begin{equation}
    \label{eq:X_Z_misspecify}
    X_t \mid Z_t \sim \mathcal{N}(u^\top Z_t,\tilde{\sigma}), \quad \text{where} \quad u \sim \mathcal{N}(0,I_d).
\end{equation}
When setting $\tilde{\sigma}=\sigma=1$ we recover the true ${P}_{X\mid Z}$, and by increasing (resp. decreasing) $\tilde{\sigma}$ we move further away from the true conditional distribution. Figure~\ref{fig:error_sigma_ecrt} presents the type-I error rate as a function of $\tilde{\sigma}$, where each curve corresponds to a different number of samples used for initial training $n_\text{init}$. The test is applied to $n=1000$ fresh samples, in addition to the $n_\text{init}$ ones. Interestingly, observe how the type-I error is conservatively controlled for small values of $\tilde{\sigma}<\sigma = 1$, whereas inflation in the type-I error is reported for larger values of $\tilde{\sigma}$. Observe also how this type-I error inflation is mitigated when using more samples for initial training.
To illustrate this behavior from a different angle, we present in Figure~\ref{fig:mse_sigma_ecrt} the difference $\tilde{q}_t-q_t$ as a function of $t$ for different values of $\tilde{\sigma}$, with the choice of $n_\text{init}=20$. As displayed, the difference $\tilde{q}_t-q_t$ that corresponds to $\tilde{\sigma}=0.1$ tends to be smaller than the one corresponding to the true $\tilde{\sigma} = \sigma = 1$, which is in line with the tendency of our method to construct conservative martingale when $\tilde{\sigma}$ is small. On the other hand, when $\tilde{\sigma}=3$ the difference $\tilde{q}_t-q_t$ tends to be larger than that of $\tilde{\sigma} = 1$, and this gap decreases as $t$ increases, where the difference is closer to zero for $t=150$. This shows that the predictive model we use (lasso) tends to ignore the null feature with the increase of the sample size, and thus the type-I error is moderated for larger $n_\text{init}$.

\begin{figure}
     \centering
     \begin{subfigure}[b]{0.49\textwidth}
         \centering
         \includegraphics[width=\textwidth]{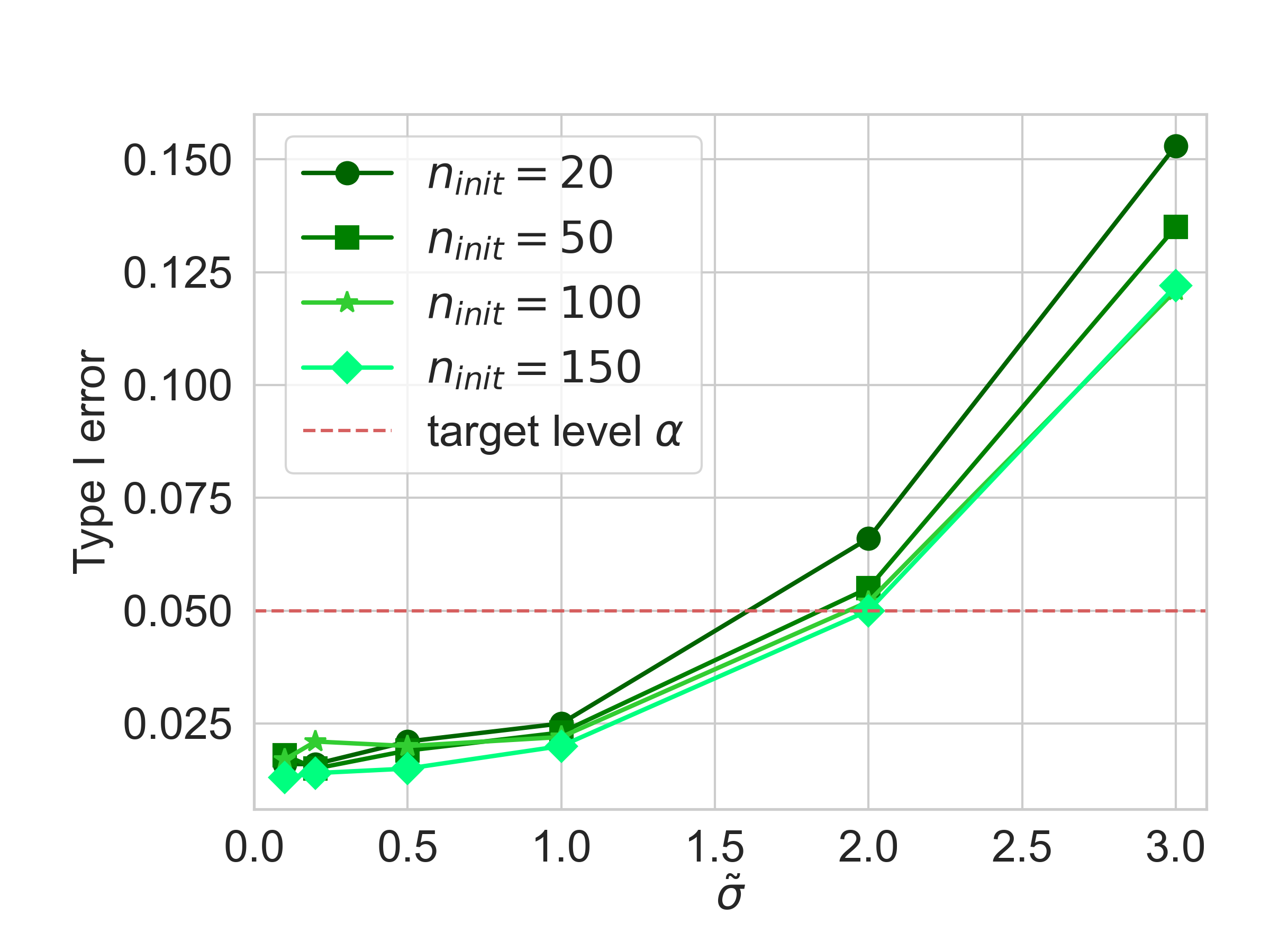}
         \caption{}
         \label{fig:error_sigma_ecrt}
     \end{subfigure}
     \hfill
     \begin{subfigure}[b]{0.49\textwidth}
         \centering
         \includegraphics[width=\textwidth]{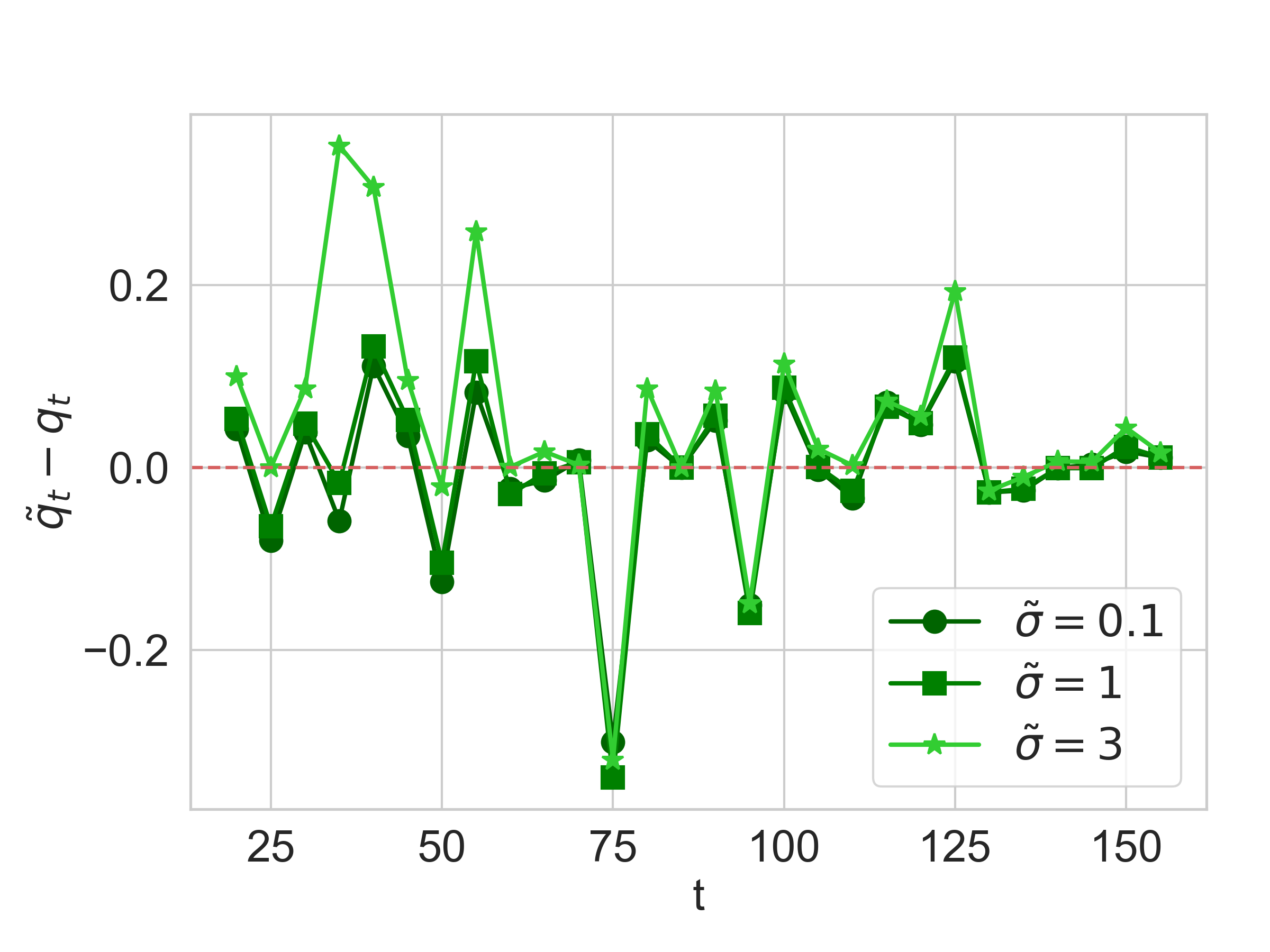}
         \caption{}
         \label{fig:mse_sigma_ecrt}
     \end{subfigure}
        \caption{\textbf{Robustness experiments with simulated data.} The dummies $\tilde{X}_t$ are generated from a misspecify model of $P_{X \mid Z}$, where the farther $\hat{\sigma}$ from $\sigma = 1$ the larger the estimation error of $P_{X \mid Z}$. (a) Type-I error of e-CRT evaluated over 1000 realizations of the \texttt{null data} model, where each curve represents a different number of samples used for initial training of the predictive model. (b) $\tilde{q}_t-q_t$ of a single realization of the \texttt{null data} model as a function of $t$, with $n_\text{init}=20$.}
        \label{fig:sigma_ecrt_plots}
\end{figure}

\subsubsection{Non-Parametric Estimation of $P_{X \mid Z}$}
\label{supp:non_parametric_est}
In this section, we consider a more challenging scenario in which $Z_t \sim \mathcal{N}(0,I_{19})$ and $X_t \mid Z_t$ follows a Student-$t$ distribution with 5 degrees of freedom and mean equals to $Z_t^{(1)}Z_t^{(2)}$. Here, $Z_t^{(i)}$ refers to the $i$'th covariate of $Z_t$. We generate the response $Y_t$ as described in Section~\ref{sec:exp_setup}. To estimate $P_{X \mid Z}$, we first generate 3000 unlabeled samples $(X,Z)$ and use the non-linear density estimation method proposed by \cite{rosenberg2022fast}, called NL-VQR. For hyperparameter tuning, we train NL-VQR on 2000 of the unlabeled samples and subsequently evaluated its goodness-of-fit on the remaining 1000 unlabeled samples using the KDE-L1 metric as described by \cite{rosenberg2022fast}. We then train the NL-VQR with all the 3000 unlabeled samples using the chosen hyperparameters. We deploy the e-CRT on 3000 new, labeled data points, sampled from the same distribution, with a kernel ridge regression model $\hat{f}_t$ with a polynomial kernel of degree 2, whose parameters are tuned by 5-fold cross-validation. The model is fitted on $\{(X_s,Y_s,Z_s)\}_{s=1}^{t-1}$, at each time step $t$. Figure~\ref{fig:stud_t} presents the empirical power and type-I error obtained by the e-CRT. Observe how the type-I error rate grows slowly, but it is controlled even for a relatively large number of samples. Observe also how the power reaches 1 when the sample size is relatively large. 

\begin{figure}
     \centering
     \begin{subfigure}[b]{0.49\textwidth}
         \centering
         \includegraphics[width=\textwidth]{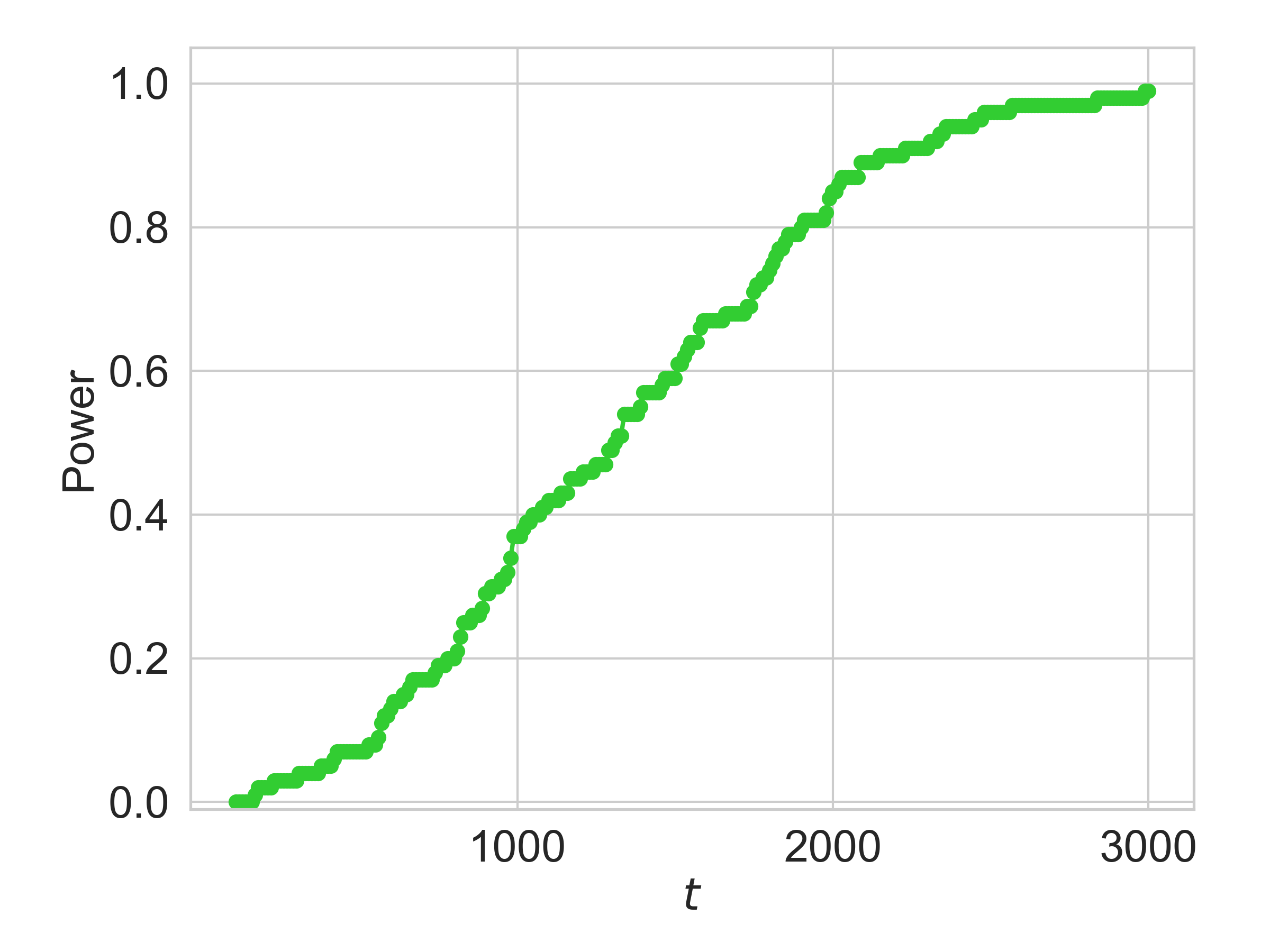}
         \caption{Empirical power}
         \label{fig:stud_t_pow}
     \end{subfigure}
     \hfill
     \begin{subfigure}[b]{0.49\textwidth}
         \centering
         \includegraphics[width=\textwidth]{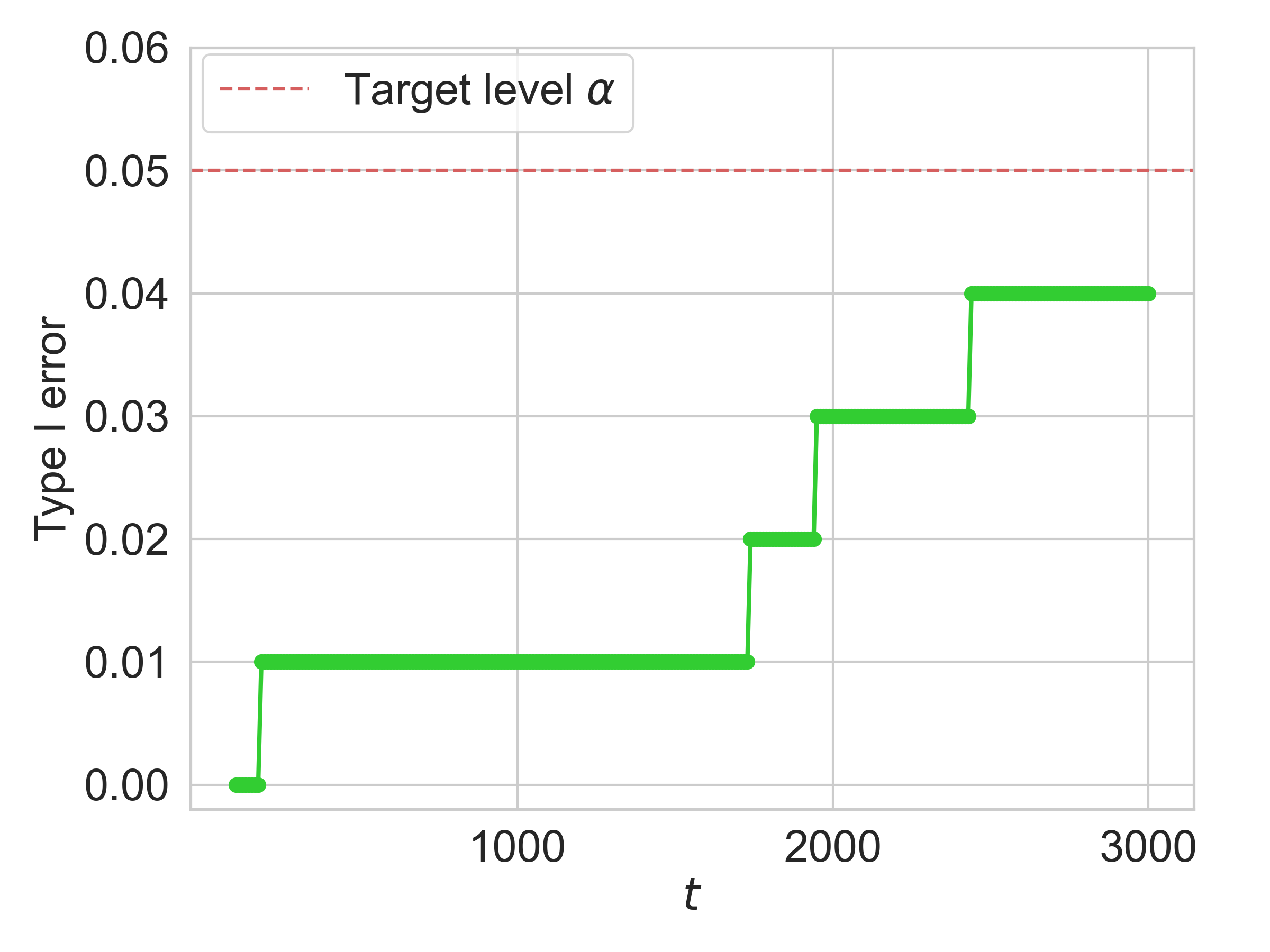}
         \caption{Type-I error}
         \label{fig:stud_t_err}
     \end{subfigure}
    \caption{\textbf{Robustness experiments with simulated data where $X \mid Z$ follows a Student-$t$ distribution.} The dummies $\tilde{X}_t$ are generated from an estimated $P_{X \mid Z}$, using the density estimation method proposed by \cite{rosenberg2022fast}. The empirical power and type-I error are evaluated over 100 realizations of the data.}
    \label{fig:stud_t}
\end{figure}

\section{SUPPLEMENTARY DETAILS ON REAL DATA EXPERIMENTS}
\label{supp:real_exp}

\subsection{Fund Replication Experiment}

\subsubsection{Supplementary Implementation Details}
\label{supp:fund_imp}
Here we provide supplementary details on the implementation of the fund replication experiment, described in Section~\ref{sec:fund_rep}. We denote by $X_t^j \in \mathbb{R}$ the $t$th log return of the $j$th stock, and by $Z_t^j \in \mathbb{R}^{d-1}$ the vector of the $t$th log returns of all the stocks except $X_t^j$.
We deploy the e-CRT on the above data stream as in the synthetic experiments from Section~\ref{sec:exp_setup}, but with the following adaptations. 
Since real data sets tend to have outliers, we choose to work with larger batches and a more moderate betting function that takes into account the magnitude of the error, not only its direction as happen in the sign function used in the synthetic experiments. Specifically, we form the betting function as $g(a,b)=\textrm{tanh}(20 \cdot (b-a)/\textrm{max}\{a,b\})$, and implement our method with an ensemble over batches of size $\{5,10,20\}$. As a strong baseline for reference, we apply the offline CRT and offline HRT on the whole data set, and use lasso regression model with 5-fold cross-validation to tune its hyper-parameter.  In contrast to the controlled synthetic experiments from Section~\ref{sec:syn_experiments}, here $P_{X^j \mid Z^j}$ is unknown and thus we must estimate it from the data to generate $\tilde{X}^j$, both for e-CRT and for CRT and HRT. For the offline tests, we approximate it by fitting a multivariate Gaussian on all the samples. For e-CRT, we set $n_{\text{init}}=500$, fit a multivariate Gaussian on the first 500 samples, and use the rest samples for testing.

\subsubsection{Supplementary Table}
\label{supp:fund_table}
\begin{longtable}{
    l|l|
    >{\collectcell\pvalf}c<{\endcollectcell}| 
    >{\collectcell\pvalf}c<{\endcollectcell}| 
    >{\collectcell\evalf}c<{\endcollectcell}
    c
}
\caption{Summary of results obtained by CRT, HRT and e-CRT applied to each of the stocks in the data. The stocks belong to the Information Technology sector are highlighted in light green. For CRT and HRT, green (resp. red) value represents a p-value below (resp. above) $\alpha = 0.05$. For e-CRT, we color each value according to the categorization in \cite{evalues}: red for insignificant, orange for ‘worth a bare mention’, blue and green indicate that
the evidence against the null is substantial or strong, respectively.}
\label{tab:fund_rep}\\
\hline
\multirow{2}{*}{Symbol } & \multirow{2}{*}{Sector } & CRT & HRT & \multicolumn{2}{c}{e-CRT} \\
& & p-value & p-value & \ccell{$S_{t_\text{stop}}$} & \ccell{$t_\text{stop}$} \\
\hline
\rowcolor{green!6}FTNT & Information Technology & 0.035 & 1 & 52.8 & 2220 \\ 
\rowcolor{green!6}AAPL & Information Technology & 0.005 & 0.001 & 52.6 & 580 \\ 
\rowcolor{green!6}CDNS & Information Technology & 0.095 & 1 & 45.6 & 700 \\
\rowcolor{green!6}SNPS & Information Technology & 0.005 & 0.011 & 40.9 & 1940 \\ 
\rowcolor{green!6}MSFT & Information Technology & 0.005 & 0.001 & 37.7 & 560 \\ 
\rowcolor{green!6}ORCL & Information Technology & 0.005 & 0.157 & 33.2 & 560 \\ 
\rowcolor{green!6}V & Information Technology & 0.005 & 0.032 & 28.1 & 640 \\ 
\rowcolor{green!6}NVDA & Information Technology & 0.005 & 0.001 & 27.7 & 840 \\ 
\rowcolor{green!6}AMAT & Information Technology & 0.03 & 0.206 & 26.7 & 2260 \\ 
\rowcolor{green!6}QCOM & Information Technology & 0.005 & 0.001 & 25.5 & 1160 \\ 
\rowcolor{green!6}HPQ & Information Technology & 0.005 & 0.381 & 24.1 & 900 \\ 
\rowcolor{green!6}MA & Information Technology & 0.005 & 0.008 & 23.5 & 2200 \\ 
\rowcolor{green!6}IBM & Information Technology & 0.005 & 0.236 & 23.1 & 620 \\ 
\rowcolor{green!6}INTC & Information Technology & 0.005 & 0.002 & 22.8 & 600 \\ 
\rowcolor{green!6}ADBE & Information Technology & 0.005 & 0.006 & 22.7 & 1020 \\ 
\rowcolor{green!6}TER & Information Technology & 0.005 & 0.036 & 22.5 & 2180 \\ 
\rowcolor{green!6}ACN & Information Technology & 0.005 & 0.015 & 22.5 & 580 \\ 
\rowcolor{green!6}INTU & Information Technology & 0.005 & 1 & 21.7 & 2180 \\ 
\rowcolor{green!6}CTXS & Information Technology & 0.507 & 0.76 & 21.7 & 900 \\ 
\rowcolor{green!6}TXN & Information Technology & 0.005 & 0.085 & 20.9 & 1660 \\ 
\rowcolor{green!6}MU & Information Technology & 0.005 & 0.078 & 20.7 & 2200 \\ 
\rowcolor{green!6}CSCO & Information Technology & 0.005 & 0.093 & 20.7 & 580 \\ 
\rowcolor{green!6}AVGO & Information Technology & 0.005 & 0.281 & 20.6 & 1200 \\ 
\rowcolor{green!6}CTSH & Information Technology & 0.891 & 0.495 & 20.2 & 680 \\ 
\rowcolor{green!6}CRM & Information Technology & 0.005 & 0.008 & 20.1 & 900 \\ 
\rowcolor{green!6}KLAC & Information Technology & 0.085 & 0.223 & 8.6 & 2421 \\ 
\rowcolor{green!6}WU & Information Technology & 0.96 & 1 & 7.1 & 2421 \\ 
\rowcolor{green!6}FIS & Information Technology & 0.005 & 0.011 & 4.7 & 2421 \\ 
\rowcolor{green!6}AMD & Information Technology & 0.015 & 0.278 & 4.6 & 2421 \\ 
\rowcolor{green!6}TYL & Information Technology & 0.02 & 1 & 4.5 & 2421 \\ 
\rowcolor{green!6}IT & Information Technology & 0.02 & 0.154 & 3 & 2421 \\ 
\rowcolor{green!6}NXPI & Information Technology & 0.572 & 1 & 2.3 & 2421 \\
\rowcolor{green!6}PAYX & Information Technology & 0.731 & 1 & 1.6 & 2421 \\
\rowcolor{green!6}LRCX & Information Technology & 0.005 & 0.165 & 1.2 & 2421 \\ 
\rowcolor{green!6}ADI & Information Technology & 0.602 & 0.765 & 1 & 2421 \\ 
\rowcolor{green!6}TRMB & Information Technology & 0.716 & 1 & 1 & 2421 \\ 
\rowcolor{green!6}APH & Information Technology & 0.612 & 1 & 1 & 2421 \\ 
\rowcolor{green!6}BR & Information Technology & 0.522 & 1 & 1 & 2421 \\ 
\rowcolor{green!6}ADP & Information Technology & 0.428 & 1 & 0.9 & 2421 \\ 
\rowcolor{green!6}ADSK & Information Technology & 0.06 & 1 & 0.9 & 2421 \\ 
\rowcolor{green!6}ENPH & Information Technology & 0.537 & 1 & 0.9 & 2421 \\ 
\rowcolor{green!6}FISV & Information Technology & 0.204 & 0.286 & 0.8 & 2421 \\ 
\rowcolor{green!6}ZBRA & Information Technology & 0.199 & 1 & 0.8 & 2421 \\ 
\rowcolor{green!6}MPWR & Information Technology & 0.164 & 1 & 0.8 & 2421 \\ 
\rowcolor{green!6}MSI & Information Technology & 0.134 & 1 & 0.8 & 2421 \\ 
\rowcolor{green!6}GPN & Information Technology & 0.378 & 0.723 & 0.6 & 2421 \\ 
\rowcolor{green!6}PTC & Information Technology & 0.577 & 1 & 0.6 & 2421 \\ 
\rowcolor{green!6}ANSS & Information Technology & 0.085 & 1 & 0.6 & 2421 \\ 
\rowcolor{green!6}TEL & Information Technology & 0.005 & 0.067 & 0.6 & 2421 \\ 
\rowcolor{green!6}AKAM & Information Technology & 0.522 & 1 & 0.5 & 2421 \\ 
\rowcolor{green!6}NTAP & Information Technology & 0.383 & 1 & 0.5 & 2421 \\ 
\rowcolor{green!6}VRSN & Information Technology & 0.204 & 0.674 & 0.5 & 2421 \\ 
\rowcolor{green!6}GLW & Information Technology & 0.766 & 1 & 0.5 & 2421 \\ 
\rowcolor{green!6}FLT & Information Technology & 0.507 & 0.731 & 0.4 & 2421 \\ 
\rowcolor{green!6}DXC & Information Technology & 0.836 & 1 & 0.4 & 2421 \\ 
\rowcolor{green!6}STX & Information Technology & 0.677 & 0.593 & 0.4 & 2421 \\ 
\rowcolor{green!6}MCHP & Information Technology & 0.065 & 0.21 & 0.4 & 2421 \\ 
\rowcolor{green!6}IPGP & Information Technology & 0.517 & 1 & 0.3 & 2421 \\ 
\rowcolor{green!6}NOW & Information Technology & 0.403 & 1 & 0.3 & 2421 \\ 
\rowcolor{green!6}SWKS & Information Technology & 0.776 & 1 & 0.3 & 2421 \\ 
\rowcolor{green!6}JNPR & Information Technology & 0.428 & 1 & 0.3 & 2421 \\ 
\rowcolor{green!6}NLOK & Information Technology & 0.085 & 0.395 & 0.2 & 2421 \\ 
\rowcolor{green!6}WDC & Information Technology & 0.209 & 0.514 & 0.2 & 2421 \\ 
\rowcolor{green!6}JKHY & Information Technology & 0.831 & 0.368 & 0.2 & 2421 \\ 
\rowcolor{green!6}FFIV & Information Technology & 0.965 & 0.308 & 0.1 & 2421 \\ 
CRL & Health Care & 0.308 & 0.587 & 37.7 & 1780 \\ 
PKI & Health Care & 0.662 & 0.596 & 30.7 & 1660 \\ 
GOOGL & Communication Services & 0.045 & 0.392 & 27.8 & 1120 \\ 
EL & Consumer Staples & 0.005 & 1 & 27.4 & 2160 \\ 
ALB & Materials & 0.015 & 1 & 25.6 & 2300 \\ 
VZ & Communication Services & 0.005 & 0.386 & 25.4 & 740 \\ 
DVN & Energy & 0.03 & 1 & 24.7 & 2220 \\ 
CMG & Consumer Discretionary & 0.005 & 1 & 24.5 & 2140 \\ 
LYB & Materials & 0.194 & 1 & 23.3 & 880 \\ 
T & Communication Services & 0.01 & 0.797 & 22.7 & 700 \\ 
LUMN & Communication Services & 0.005 & 0.301 & 22.2 & 1760 \\ 
GOOG & Communication Services & 0.055 & 0.56 & 22.1 & 960 \\ 
ROP & Industrials & 0.005 & 1 & 22 & 2220 \\ 
NDAQ & Financials & 0.005 & 1 & 16.7 & 2421 \\ 
NFLX & Communication Services & 0.06 & 0.366 & 6 & 2421 \\ 
LEN & Consumer Discretionary & 0.96 & 0.165 & 5.8 & 2421 \\ 
HUM & Health Care & 0.169 & 1 & 4.6 & 2421 \\ 
IPG & Communication Services & 0.025 & 1 & 4.1 & 2421 \\ 
GNRC & Industrials & 0.104 & 1 & 3.9 & 2421 \\ 
XYL & Industrials & 0.498 & 0.563 & 3.9 & 2421 \\ 
AOS & Industrials & 0.124 & 0.956 & 3.5 & 2421 \\ 
ATVI & Communication Services & 0.677 & 0.624 & 3.4 & 2421 \\ 
VLO & Energy & 0.602 & 0.631 & 3.2 & 2421 \\ 
EA & Communication Services & 0.02 & 0.354 & 3.1 & 2421 \\ 
PWR & Industrials & 0.09 & 1 & 2.7 & 2421 \\ 
FCX & Materials & 0.015 & 0.189 & 2.7 & 2421 \\ 
EBAY & Consumer Discretionary & 0.01 & 0.432 & 2.6 & 2421 \\ 
WAT & Health Care & 0.781 & 1 & 2.2 & 2421 \\ 
DVA & Health Care & 0.005 & 1 & 2.2 & 2421 \\ 
TROW & Financials & 0.771 & 0.326 & 2.1 & 2421 \\ 
WHR & Consumer Discretionary & 0.761 & 1 & 1.9 & 2421 \\ 
DISH & Communication Services & 0.149 & 1 & 1.9 & 2421 \\ 
CE & Materials & 0.045 & 1 & 1.9 & 2421 \\ 
PEAK & Real Estate & 0.522 & 1 & 1.8 & 2421 \\ 
WMB & Energy & 0.04 & 1 & 1.8 & 2421 \\ 
EXPD & Industrials & 0.1 & 1 & 1.8 & 2421 \\ 
REGN & Health Care & 0.806 & 1 & 1.7 & 2421 \\ 
MTCH & Communication Services & 0.02 & 0.221 & 1.7 & 2421 \\ 
EFX & Industrials & 0.065 & 1 & 1.6 & 2421 \\ 
ABBV & Health Care & 0.328 & 1 & 1.5 & 2421 \\ 
KEY & Financials & 0.468 & 1 & 1.5 & 2421 \\ 
INCY & Health Care & 0.119 & 1 & 1.5 & 2421 \\ 
FAST & Industrials & 0.632 & 1 & 1.5 & 2421 \\ 
DHI & Consumer Discretionary & 0.886 & 1 & 1.5 & 2421 \\ 
SWK & Industrials & 0.164 & 1 & 1.4 & 2421 \\ 
UAL & Industrials & 0.587 & 0.473 & 1.4 & 2421 \\ 
ZION & Financials & 0.811 & 1 & 1.4 & 2421 \\ 
BKR & Energy & 0.303 & 1 & 1.4 & 2421 \\ 
TJX & Consumer Discretionary & 0.055 & 1 & 1.3 & 2421 \\ 
RHI & Industrials & 0.438 & 1 & 1.3 & 2421 \\ 
DIS & Communication Services & 0.134 & 1 & 1.3 & 2421 \\ 
CMCSA & Communication Services & 0.458 & 1 & 1.2 & 2421 \\ 
NEM & Materials & 0.194 & 1 & 1.2 & 2421 \\ 
DRE & Real Estate & 0.731 & 1 & 1.2 & 2421 \\ 
BKNG & Consumer Discretionary & 0.602 & 1 & 1.1 & 2421 \\ 
LVS & Consumer Discretionary & 0.741 & 1 & 1.1 & 2421 \\ 
PHM & Consumer Discretionary & 0.403 & 1 & 1.1 & 2421 \\ 
PFG & Financials & 0.443 & 1 & 1.1 & 2421 \\ 
CME & Financials & 0.975 & 1 & 1.1 & 2421 \\ 
COP & Energy & 0.786 & 1 & 1.1 & 2421 \\ 
EIX & Utilities & 0.313 & 1 & 1.1 & 2421 \\ 
BAX & Health Care & 0.323 & 1 & 1.1 & 2421 \\ 
APTV & Consumer Discretionary & 0.01 & 1 & 1.1 & 2421 \\ 
TECH & Health Care & 0.05 & 1 & 1.1 & 2421 \\ 
MNST & Consumer Staples & 0.313 & 1 & 1.1 & 2421 \\ 
KMX & Consumer Discretionary & 0.244 & 0.418 & 1.1 & 2421 \\ 
PPG & Materials & 0.617 & 1 & 1.1 & 2421 \\ 
LOW & Consumer Discretionary & 0.905 & 1 & 1.1 & 2421 \\ 
NWL & Consumer Discretionary & 0.468 & 1 & 1.1 & 2421 \\ 
AMZN & Consumer Discretionary & 0.015 & 0.18 & 1.1 & 2421 \\ 
VTRS & Health Care & 0.682 & 1 & 1.1 & 2421 \\ 
AMP & Financials & 0.572 & 1 & 1.1 & 2421 \\ 
JNJ & Health Care & 0.945 & 1 & 1.1 & 2421 \\ 
LUV & Industrials & 0.736 & 1 & 1 & 2421 \\ 
LNT & Utilities & 0.458 & 1 & 1 & 2421 \\ 
LMT & Industrials & 0.667 & 1 & 1 & 2421 \\ 
COST & Consumer Staples & 0.786 & 1 & 1 & 2421 \\ 
LLY & Health Care & 0.811 & 1 & 1 & 2421 \\ 
EXR & Real Estate & 0.229 & 1 & 1 & 2421 \\ 
COO & Health Care & 0.478 & 1 & 1 & 2421 \\ 
CMS & Utilities & 0.378 & 1 & 1 & 2421 \\ 
MDT & Health Care & 0.557 & 1 & 1 & 2421 \\ 
EMR & Industrials & 0.627 & 1 & 1 & 2421 \\ 
MMC & Financials & 0.468 & 1 & 1 & 2421 \\ 
MMM & Industrials & 0.711 & 1 & 1 & 2421 \\ 
MO & Consumer Staples & 0.677 & 1 & 1 & 2421 \\ 
CLX & Consumer Staples & 0.512 & 1 & 1 & 2421 \\ 
MRK & Health Care & 0.507 & 1 & 1 & 2421 \\ 
CL & Consumer Staples & 0.537 & 1 & 1 & 2421 \\ 
MSCI & Financials & 0.294 & 1 & 1 & 2421 \\ 
MTB & Financials & 0.02 & 1 & 1 & 2421 \\ 
MTD & Health Care & 0.279 & 1 & 1 & 2421 \\ 
CI & Health Care & 0.408 & 1 & 1 & 2421 \\ 
ECL & Materials & 0.453 & 1 & 1 & 2421 \\ 
CMA & Financials & 0.368 & 1 & 1 & 2421 \\ 
KMB & Consumer Staples & 0.771 & 1 & 1 & 2421 \\ 
LIN & Materials & 0.015 & 1 & 1 & 2421 \\ 
DHR & Health Care & 0.846 & 1 & 1 & 2421 \\ 
HBI & Consumer Discretionary & 0.766 & 1 & 1 & 2421 \\ 
ETR & Utilities & 0.726 & 1 & 1 & 2421 \\ 
HBAN & Financials & 0.179 & 1 & 1 & 2421 \\ 
HAS & Consumer Discretionary & 0.692 & 1 & 1 & 2421 \\ 
DG & Consumer Discretionary & 0.184 & 1 & 1 & 2421 \\ 
GWW & Industrials & 0.557 & 1 & 1 & 2421 \\ 
DGX & Health Care & 0.612 & 1 & 1 & 2421 \\ 
EVRG & Utilities & 0.388 & 1 & 1 & 2421 \\ 
GPC & Consumer Discretionary & 0.348 & 1 & 1 & 2421 \\ 
EOG & Energy & 0.866 & 1 & 1 & 2421 \\ 
EW & Health Care & 0.746 & 1 & 1 & 2421 \\ 
GIS & Consumer Staples & 0.428 & 1 & 1 & 2421 \\ 
DLR & Real Estate & 0.572 & 1 & 1 & 2421 \\ 
GD & Industrials & 0.687 & 1 & 1 & 2421 \\ 
FRC & Financials & 0.706 & 1 & 1 & 2421 \\ 
FMC & Materials & 0.761 & 1 & 1 & 2421 \\ 
FITB & Financials & 0.289 & 1 & 1 & 2421 \\ 
EXC & Utilities & 0.493 & 1 & 1 & 2421 \\ 
HD & Consumer Discretionary & 0.756 & 1 & 1 & 2421 \\ 
HRL & Consumer Staples & 0.552 & 1 & 1 & 2421 \\ 
HST & Real Estate & 0.746 & 1 & 1 & 2421 \\ 
ESS & Real Estate & 0.612 & 1 & 1 & 2421 \\ 
LH & Health Care & 0.209 & 1 & 1 & 2421 \\ 
DUK & Utilities & 0.846 & 1 & 1 & 2421 \\ 
LEG & Consumer Discretionary & 0.06 & 1 & 1 & 2421 \\ 
L & Financials & 0.587 & 1 & 1 & 2421 \\ 
CVS & Health Care & 0.144 & 1 & 1 & 2421 \\ 
KMI & Energy & 0.164 & 1 & 1 & 2421 \\ 
F & Consumer Discretionary & 0.129 & 1 & 1 & 2421 \\ 
KIM & Real Estate & 0.582 & 1 & 1 & 2421 \\ 
K & Consumer Staples & 0.756 & 1 & 1 & 2421 \\ 
JCI & Industrials & 0.687 & 1 & 1 & 2421 \\ 
ITW & Industrials & 0.602 & 1 & 1 & 2421 \\ 
D & Utilities & 0.308 & 1 & 1 & 2421 \\ 
DD & Materials & 0.473 & 1 & 1 & 2421 \\ 
EQR & Real Estate & 0.896 & 1 & 1 & 2421 \\ 
ILMN & Health Care & 0.279 & 1 & 1 & 2421 \\ 
IDXX & Health Care & 0.597 & 1 & 1 & 2421 \\ 
ES & Utilities & 0.736 & 1 & 1 & 2421 \\ 
IEX & Industrials & 0.771 & 1 & 1 & 2421 \\ 
JPM & Financials & 0.736 & 0.745 & 1 & 2421 \\ 
SNA & Industrials & 0.781 & 1 & 1 & 2421 \\ 
AMT & Real Estate & 0.791 & 1 & 1 & 2421 \\ 
AON & Financials & 0.512 & 1 & 1 & 2421 \\ 
TGT & Consumer Discretionary & 0.637 & 1 & 1 & 2421 \\ 
TFX & Health Care & 0.438 & 1 & 1 & 2421 \\ 
TFC & Financials & 0.836 & 1 & 1 & 2421 \\ 
ARE & Real Estate & 0.07 & 1 & 1 & 2421 \\ 
SYK & Health Care & 0.567 & 1 & 1 & 2421 \\ 
STT & Financials & 0.562 & 1 & 1 & 2421 \\ 
AVB & Real Estate & 0.617 & 1 & 1 & 2421 \\ 
SO & Utilities & 0.338 & 1 & 1 & 2421 \\ 
SLB & Energy & 0.761 & 1 & 1 & 2421 \\ 
SJM & Consumer Staples & 0.826 & 1 & 1 & 2421 \\ 
SIVB & Financials & 0.453 & 1 & 1 & 2421 \\ 
SCHW & Financials & 0.9 & 1 & 1 & 2421 \\ 
SBUX & Consumer Discretionary & 0.95 & 1 & 1 & 2421 \\ 
RTX & Industrials & 0.393 & 1 & 1 & 2421 \\ 
RSG & Industrials & 0.294 & 1 & 1 & 2421 \\ 
AWK & Utilities & 0.771 & 1 & 1 & 2421 \\ 
ROL & Industrials & 0.542 & 1 & 1 & 2421 \\ 
AXP & Financials & 0.542 & 1 & 1 & 2421 \\ 
TSCO & Consumer Discretionary & 0.333 & 1 & 1 & 2421 \\ 
ALK & Industrials & 0.677 & 1 & 1 & 2421 \\ 
RF & Financials & 0.572 & 1 & 1 & 2421 \\ 
UNH & Health Care & 0.801 & 1 & 1 & 2421 \\ 
AAP & Consumer Discretionary & 0.721 & 1 & 1 & 2421 \\ 
YUM & Consumer Discretionary & 0.905 & 1 & 1 & 2421 \\ 
XRAY & Health Care & 0.557 & 1 & 1 & 2421 \\ 
XEL & Utilities & 0.249 & 1 & 1 & 2421 \\ 
ABC & Health Care & 0.627 & 1 & 1 & 2421 \\ 
ABT & Health Care & 0.522 & 1 & 1 & 2421 \\ 
ADM & Consumer Staples & 0.473 & 1 & 1 & 2421 \\ 
AEE & Utilities & 0.557 & 1 & 1 & 2421 \\ 
AEP & Utilities & 0.662 & 1 & 1 & 2421 \\ 
WFC & Financials & 0.572 & 1 & 1 & 2421 \\ 
WELL & Real Estate & 0.582 & 1 & 1 & 2421 \\ 
AFL & Financials & 0.388 & 0.358 & 1 & 2421 \\ 
AJG & Financials & 0.731 & 1 & 1 & 2421 \\ 
VRSK & Industrials & 0.9 & 1 & 1 & 2421 \\ 
VNO & Real Estate & 0.637 & 1 & 1 & 2421 \\ 
VMC & Materials & 0.269 & 1 & 1 & 2421 \\ 
VFC & Consumer Discretionary & 0.169 & 1 & 1 & 2421 \\ 
USB & Financials & 0.164 & 1 & 1 & 2421 \\ 
UNP & Industrials & 0.308 & 1 & 1 & 2421 \\ 
BAC & Financials & 0.532 & 1 & 1 & 2421 \\ 
NI & Utilities & 0.607 & 1 & 1 & 2421 \\ 
PSA & Real Estate & 0.06 & 1 & 1 & 2421 \\ 
O & Real Estate & 0.612 & 1 & 1 & 2421 \\ 
PNW & Utilities & 0.045 & 1 & 1 & 2421 \\ 
PNR & Industrials & 0.876 & 1 & 1 & 2421 \\ 
CAH & Health Care & 0.04 & 1 & 1 & 2421 \\ 
BK & Financials & 0.348 & 1 & 1 & 2421 \\ 
PPL & Utilities & 0.428 & 1 & 1 & 2421 \\ 
OMC & Communication Services & 0.786 & 1 & 1 & 2421 \\ 
PKG & Materials & 0.667 & 1 & 1 & 2421 \\ 
BRO & Financials & 0.627 & 1 & 1 & 2421 \\ 
PH & Industrials & 0.532 & 1 & 1 & 2421 \\ 
PM & Consumer Staples & 0.552 & 1 & 1 & 2421 \\ 
PLD & Real Estate & 0.866 & 1 & 1 & 2421 \\ 
BLK & Financials & 0.458 & 1 & 1 & 2421 \\ 
PEG & Utilities & 0.353 & 1 & 1 & 2421 \\ 
PVH & Consumer Discretionary & 0.239 & 1 & 1 & 2421 \\ 
PEP & Consumer Staples & 0.109 & 1 & 1 & 2421 \\ 
PG & Consumer Staples & 0.94 & 1 & 1 & 2421 \\ 
PFE & Health Care & 0.687 & 1 & 1 & 2421 \\ 
NTRS & Financials & 0.403 & 1 & 1 & 2421 \\ 
NSC & Industrials & 0.557 & 1 & 1 & 2421 \\ 
PNC & Financials & 0.224 & 1 & 1 & 2421 \\ 
HCA & Health Care & 0.741 & 1 & 0.9 & 2421 \\ 
TDG & Industrials & 0.821 & 1 & 0.9 & 2421 \\ 
MCK & Health Care & 0.323 & 1 & 0.9 & 2421 \\ 
BWA & Consumer Discretionary & 0.846 & 1 & 0.9 & 2421 \\ 
COF & Financials & 0.04 & 0.959 & 0.9 & 2421 \\ 
AIG & Financials & 0.537 & 0.291 & 0.9 & 2421 \\ 
CNC & Health Care & 0.463 & 1 & 0.9 & 2421 \\ 
DE & Industrials & 0.134 & 1 & 0.9 & 2421 \\ 
TTWO & Communication Services & 0.443 & 1 & 0.9 & 2421 \\ 
UDR & Real Estate & 0.02 & 1 & 0.9 & 2421 \\ 
DLTR & Consumer Discretionary & 0.756 & 1 & 0.9 & 2421 \\ 
UPS & Industrials & 0.393 & 1 & 0.9 & 2421 \\ 
HIG & Financials & 0.731 & 1 & 0.9 & 2421 \\ 
STZ & Consumer Staples & 0.841 & 0.719 & 0.9 & 2421 \\ 
HOLX & Health Care & 0.9 & 1 & 0.9 & 2421 \\ 
BIIB & Health Care & 0.468 & 1 & 0.9 & 2421 \\ 
WYNN & Consumer Discretionary & 0.572 & 1 & 0.9 & 2421 \\ 
BDX & Health Care & 0.935 & 1 & 0.9 & 2421 \\ 
LKQ & Consumer Discretionary & 0.756 & 1 & 0.9 & 2421 \\ 
PSX & Energy & 0.234 & 1 & 0.9 & 2421 \\ 
AVY & Materials & 0.388 & 1 & 0.9 & 2421 \\ 
JBHT & Industrials & 0.652 & 1 & 0.9 & 2421 \\ 
CVX & Energy & 0.547 & 1 & 0.9 & 2421 \\ 
SHW & Materials & 0.617 & 1 & 0.9 & 2421 \\ 
FDX & Industrials & 0.697 & 1 & 0.9 & 2421 \\ 
HSY & Consumer Staples & 0.547 & 0.859 & 0.9 & 2421 \\ 
HSIC & Health Care & 0.04 & 1 & 0.9 & 2421 \\ 
NEE & Utilities & 0.025 & 1 & 0.9 & 2421 \\ 
IFF & Materials & 0.592 & 1 & 0.9 & 2421 \\ 
HON & Industrials & 0.463 & 1 & 0.9 & 2421 \\ 
IP & Materials & 0.502 & 1 & 0.9 & 2421 \\ 
FBHS & Industrials & 0.383 & 1 & 0.8 & 2421 \\ 
CAT & Industrials & 0.726 & 1 & 0.8 & 2421 \\ 
MKTX & Financials & 0.433 & 1 & 0.8 & 2421 \\ 
FANG & Energy & 0.642 & 1 & 0.8 & 2421 \\ 
VTR & Real Estate & 0.94 & 1 & 0.8 & 2421 \\ 
BXP & Real Estate & 0.463 & 1 & 0.8 & 2421 \\ 
MPC & Energy & 0.498 & 1 & 0.8 & 2421 \\ 
NUE & Materials & 0.189 & 1 & 0.8 & 2421 \\ 
NVR & Consumer Discretionary & 0.338 & 1 & 0.8 & 2421 \\ 
WEC & Utilities & 0.93 & 1 & 0.8 & 2421 \\ 
MS & Financials & 0.816 & 1 & 0.8 & 2421 \\ 
REG & Real Estate & 0.582 & 1 & 0.8 & 2421 \\ 
FE & Utilities & 0.065 & 0.287 & 0.8 & 2421 \\ 
RE & Financials & 0.856 & 1 & 0.8 & 2421 \\ 
CPRT & Industrials & 0.194 & 1 & 0.8 & 2421 \\ 
SEE & Materials & 0.025 & 1 & 0.8 & 2421 \\ 
J & Industrials & 0.736 & 1 & 0.8 & 2421 \\ 
SPG & Real Estate & 0.493 & 1 & 0.8 & 2421 \\ 
ICE & Financials & 0.935 & 0.348 & 0.8 & 2421 \\ 
STE & Health Care & 0.025 & 1 & 0.8 & 2421 \\ 
MHK & Consumer Discretionary & 0.09 & 1 & 0.8 & 2421 \\ 
CTAS & Industrials & 0.01 & 0.054 & 0.8 & 2421 \\ 
HES & Energy & 0.672 & 1 & 0.8 & 2421 \\ 
LDOS & Industrials & 0.607 & 1 & 0.8 & 2421 \\ 
TMUS & Communication Services & 0.612 & 1 & 0.8 & 2421 \\ 
TRV & Financials & 0.03 & 1 & 0.8 & 2421 \\ 
TSLA & Consumer Discretionary & 0.005 & 1 & 0.8 & 2421 \\ 
AMGN & Health Care & 0.537 & 1 & 0.8 & 2421 \\ 
DOV & Industrials & 0.522 & 1 & 0.8 & 2421 \\ 
ROK & Industrials & 0.647 & 1 & 0.8 & 2421 \\ 
BA & Industrials & 0.632 & 1 & 0.7 & 2421 \\ 
CBRE & Real Estate & 0.095 & 1 & 0.7 & 2421 \\ 
CPB & Consumer Staples & 0.234 & 1 & 0.7 & 2421 \\ 
DTE & Utilities & 0.915 & 1 & 0.7 & 2421 \\ 
APD & Materials & 0.682 & 1 & 0.7 & 2421 \\ 
ALL & Financials & 0.328 & 1 & 0.7 & 2421 \\ 
CCL & Consumer Discretionary & 0.493 & 1 & 0.7 & 2421 \\ 
ZTS & Health Care & 0.448 & 1 & 0.7 & 2421 \\ 
A & Health Care & 0.652 & 1 & 0.7 & 2421 \\ 
HII & Industrials & 0.881 & 0.397 & 0.7 & 2421 \\ 
ED & Utilities & 0.363 & 1 & 0.7 & 2421 \\ 
OKE & Energy & 0.592 & 1 & 0.7 & 2421 \\ 
ORLY & Consumer Discretionary & 0.562 & 1 & 0.7 & 2421 \\ 
PGR & Financials & 0.562 & 1 & 0.7 & 2421 \\ 
MLM & Materials & 0.249 & 1 & 0.7 & 2421 \\ 
MET & Financials & 0.214 & 1 & 0.7 & 2421 \\ 
RCL & Consumer Discretionary & 0.617 & 1 & 0.7 & 2421 \\ 
SBAC & Real Estate & 0.378 & 1 & 0.7 & 2421 \\ 
SRE & Utilities & 0.284 & 1 & 0.7 & 2421 \\ 
LNC & Financials & 0.667 & 1 & 0.7 & 2421 \\ 
NRG & Utilities & 0.199 & 0.448 & 0.7 & 2421 \\ 
TMO & Health Care & 0.045 & 0.28 & 0.7 & 2421 \\ 
LHX & Industrials & 0.756 & 1 & 0.7 & 2421 \\ 
NOC & Industrials & 0.348 & 1 & 0.7 & 2421 \\ 
GPS & Consumer Discretionary & 0.358 & 1 & 0.7 & 2421 \\ 
WBA & Consumer Staples & 0.234 & 1 & 0.7 & 2421 \\ 
WM & Industrials & 0.94 & 1 & 0.7 & 2421 \\ 
WMT & Consumer Staples & 0.915 & 0.349 & 0.7 & 2421 \\ 
GE & Industrials & 0.244 & 1 & 0.7 & 2421 \\ 
WST & Health Care & 0.493 & 1 & 0.7 & 2421 \\ 
CHRW & Industrials & 0.169 & 1 & 0.6 & 2421 \\ 
BMY & Health Care & 0.886 & 1 & 0.6 & 2421 \\ 
HAL & Energy & 0.254 & 1 & 0.6 & 2421 \\ 
MRO & Energy & 0.781 & 1 & 0.6 & 2421 \\ 
EQIX & Real Estate & 0.597 & 0.687 & 0.6 & 2421 \\ 
POOL & Consumer Discretionary & 0.761 & 1 & 0.6 & 2421 \\ 
C & Financials & 0.781 & 0.058 & 0.6 & 2421 \\ 
CCI & Real Estate & 0.9 & 1 & 0.6 & 2421 \\ 
MDLZ & Consumer Staples & 0.353 & 1 & 0.6 & 2421 \\ 
PXD & Energy & 0.542 & 1 & 0.6 & 2421 \\ 
NCLH & Consumer Discretionary & 0.662 & 1 & 0.6 & 2421 \\ 
TXT & Industrials & 0.697 & 1 & 0.6 & 2421 \\ 
ROST & Consumer Discretionary & 0.537 & 1 & 0.6 & 2421 \\ 
IVZ & Financials & 0.821 & 1 & 0.6 & 2421 \\ 
TT & Industrials & 0.139 & 1 & 0.6 & 2421 \\ 
IRM & Real Estate & 0.413 & 1 & 0.6 & 2421 \\ 
ISRG & Health Care & 0.139 & 0.394 & 0.6 & 2421 \\ 
AES & Utilities & 0.706 & 1 & 0.6 & 2421 \\ 
WAB & Industrials & 0.249 & 1 & 0.6 & 2421 \\ 
MAA & Real Estate & 0.592 & 1 & 0.6 & 2421 \\ 
VRTX & Health Care & 0.801 & 0.63 & 0.5 & 2421 \\ 
ATO & Utilities & 0.891 & 1 & 0.5 & 2421 \\ 
WRB & Financials & 0.776 & 1 & 0.5 & 2421 \\ 
SYY & Consumer Staples & 0.104 & 0.483 & 0.5 & 2421 \\ 
SPGI & Financials & 0.279 & 1 & 0.5 & 2421 \\ 
TAP & Consumer Staples & 0.035 & 1 & 0.5 & 2421 \\ 
PRU & Financials & 0.353 & 1 & 0.5 & 2421 \\ 
TPR & Consumer Discretionary & 0.826 & 0.538 & 0.5 & 2421 \\ 
BEN & Financials & 0.02 & 1 & 0.5 & 2421 \\ 
AME & Industrials & 0.408 & 1 & 0.5 & 2421 \\ 
WY & Real Estate & 0.174 & 1 & 0.5 & 2421 \\ 
PCAR & Industrials & 0.597 & 1 & 0.5 & 2421 \\ 
CAG & Consumer Staples & 0.383 & 1 & 0.5 & 2421 \\ 
CNP & Utilities & 0.065 & 1 & 0.5 & 2421 \\ 
CINF & Financials & 0.403 & 1 & 0.5 & 2421 \\ 
GS & Financials & 0.806 & 1 & 0.5 & 2421 \\ 
DAL & Industrials & 0.542 & 1 & 0.5 & 2421 \\ 
GM & Consumer Discretionary & 0.284 & 1 & 0.5 & 2421 \\ 
CSX & Industrials & 0.771 & 1 & 0.5 & 2421 \\ 
GL & Financials & 0.284 & 1 & 0.5 & 2421 \\ 
MAR & Consumer Discretionary & 0.174 & 1 & 0.5 & 2421 \\ 
EXPE & Consumer Discretionary & 0.711 & 0.369 & 0.5 & 2421 \\ 
DRI & Consumer Discretionary & 0.438 & 1 & 0.5 & 2421 \\ 
GRMN & Consumer Discretionary & 0.483 & 1 & 0.5 & 2421 \\ 
CHTR & Communication Services & 0.418 & 1 & 0.5 & 2421 \\ 
ABMD & Health Care & 0.637 & 1 & 0.4 & 2421 \\ 
GILD & Health Care & 0.493 & 1 & 0.4 & 2421 \\ 
UHS & Health Care & 0.01 & 0.272 & 0.4 & 2421 \\ 
ODFL & Industrials & 0.07 & 1 & 0.4 & 2421 \\ 
DFS & Financials & 0.005 & 1 & 0.4 & 2421 \\ 
NLSN & Industrials & 0.02 & 1 & 0.4 & 2421 \\ 
TDY & Industrials & 0.682 & 1 & 0.4 & 2421 \\ 
CF & Materials & 0.119 & 1 & 0.4 & 2421 \\ 
CHD & Consumer Staples & 0.597 & 1 & 0.4 & 2421 \\ 
MCD & Consumer Discretionary & 0.836 & 0.671 & 0.4 & 2421 \\ 
RMD & Health Care & 0.657 & 0.687 & 0.4 & 2421 \\ 
MAS & Industrials & 0.905 & 1 & 0.4 & 2421 \\ 
RL & Consumer Discretionary & 0.244 & 1 & 0.4 & 2421 \\ 
EMN & Materials & 0.045 & 1 & 0.4 & 2421 \\ 
MKC & Consumer Staples & 0.353 & 0.687 & 0.4 & 2421 \\ 
BIO & Health Care & 0.905 & 1 & 0.4 & 2421 \\ 
MGM & Consumer Discretionary & 0.637 & 1 & 0.4 & 2421 \\ 
BBY & Consumer Discretionary & 0.383 & 1 & 0.4 & 2421 \\ 
AMCR & Materials & 0.363 & 1 & 0.4 & 2421 \\ 
CB & Financials & 0.771 & 1 & 0.4 & 2421 \\ 
AIZ & Financials & 0.736 & 1 & 0.3 & 2421 \\ 
AAL & Industrials & 0.771 & 1 & 0.3 & 2421 \\ 
FRT & Real Estate & 0.01 & 0.198 & 0.3 & 2421 \\ 
URI & Industrials & 0.716 & 0.805 & 0.3 & 2421 \\ 
ALGN & Health Care & 0.592 & 1 & 0.3 & 2421 \\ 
DPZ & Consumer Discretionary & 0.856 & 1 & 0.3 & 2421 \\ 
ULTA & Consumer Discretionary & 0.259 & 0.353 & 0.3 & 2421 \\ 
ZBH & Health Care & 0.388 & 1 & 0.3 & 2421 \\ 
ETN & Industrials & 0.537 & 1 & 0.3 & 2421 \\ 
DXCM & Health Care & 0.214 & 0.327 & 0.3 & 2421 \\ 
UAA & Consumer Discretionary & 0.791 & 1 & 0.3 & 2421 \\ 
KO & Consumer Staples & 0.02 & 1 & 0.3 & 2421 \\ 
NKE & Consumer Discretionary & 0.204 & 0.586 & 0.3 & 2421 \\ 
BSX & Health Care & 0.517 & 0.834 & 0.3 & 2421 \\ 
CBOE & Financials & 0.015 & 1 & 0.3 & 2421 \\ 
LYV & Communication Services & 0.234 & 0.495 & 0.3 & 2421 \\ 
AZO & Consumer Discretionary & 0.02 & 0.245 & 0.3 & 2421 \\ 
PENN & Consumer Discretionary & 0.02 & 1 & 0.3 & 2421 \\ 
CMI & Industrials & 0.015 & 0.217 & 0.3 & 2421 \\ 
MOS & Materials & 0.537 & 1 & 0.2 & 2421 \\ 
XOM & Energy & 0.189 & 1 & 0.2 & 2421 \\ 
OXY & Energy & 0.632 & 1 & 0.2 & 2421 \\ 
BBWI & Consumer Discretionary & 0.617 & 0.605 & 0.2 & 2421 \\ 
TSN & Consumer Staples & 0.423 & 1 & 0.2 & 2421 \\ 
RJF & Financials & 0.592 & 1 & 0.2 & 2421 \\ 
CTRA & Energy & 0.383 & 1 & 0.2 & 2421 \\ 
APA & Energy & 0.672 & 1 & 0.2 & 2421 \\ 
KR & Consumer Staples & 0.562 & 0.719 & 0.2 & 2421 \\ 
MCO & Financials & 0.905 & 0.343 & 0.1 & 2421 \\ 
 \hline 
\end{longtable}

\subsection{Supplementary Details on the HIV Drug Resistance Experiment}

\subsubsection{Data and Implementation Details}
\label{supp:hiv_imp}



In Section~\ref{sec:hiv} we present an experiment of detection mutations in HIV that are associated with drug resistance. The data set\footnote{Data is available online at \url{https://hivdb.stanford.edu/pages/published_analysis/genophenoPNAS2006}} we consider there has not been collected sequentially, and thus it is not ideal to present the strength of our sequential testing procedure. Yet, we choose this task because of its importance, and since it has been studied in depth in the knockoff literature \cite{barber2015controlling,lu2018deeppink,romano2020deep,MRD}. In particular, this data set is convenient to analyze as the effect of each mutation on drug resistance---reported by previous scientific works---is summarized in \url{https://hivdb.stanford.edu/dr-summary/comments/PI/}.

We denote by $(X^{j}_t, Z^{j}_t, Y_t)$ the $t$th sample, where $X^j_t \in \{0,1\}$ indicates the presence or absence of the $j$th mutation, and $Z^j_t \in \mathbb{R}^{d-1}$ is a vector that contains all the remaining measured mutations in locations $1, 2, \dots, j-1, j+1, \dots, d$. The response $Y_t$ represents the log-fold increase in drug resistance. We deploy the e-CRT the same way as described in Section~\ref{sec:fund_rep}, with an additional adaptation; we fit a 5-fold cross-validated lasso model $\hat{f}_t$ on $\{(X_s,Y_s,Z_s)\}_{s=1}^{t-1}$ at each time step $t$, in contrast to the 1-fold cross-validation approach we used in the previous experiments; we use the latter to reduce computational cost, illustrating how to combine e-CRT with online learning algorithms. Naturally, the 5-fold cross validation approach leads to a better choice of lasso's hyper-parameter, and thus obtaining more accurate predictive models. 
Next, we approximate $P_{X^j \mid Z^j}$ as follows. Since $X^j$ is binary, we sample $\tilde{X}^j_t$ from a Bernoulli distribution with probability of success $\hat{\pi}^j(Z_t^j)$, where $\hat{\pi}^j(Z_t^j)$ estimates $P_{X^j \mid Z^j}(X^j = 1\mid Z^j)$. We formulate this estimator by fitting a logistic regression model on the unlabeled data $\{(X_t^j, Z_t^j)\}_{t=1}^n$, with an $l_2$ regularization whose penalty parameter is tuned via 10-fold cross-validation.

\subsubsection{Supplementary Results}
\label{supp:hiv_res}

\begin{longtable}{
    c|c|
    >{\collectcell\pvalf}c<{\endcollectcell}| 
    >{\collectcell\pvalf}c<{\endcollectcell}| 
    >{\collectcell\evalf}c<{\endcollectcell}
    c
}
\caption{Summary of the output of CRT, HRT and e-CRT applied to each of the HIV mutations in the data. The type of each mutation (Major, Minor, Accessory, Other, Unknown) represents the effect of the feature on drug resistance as reported by previous studies.  The other details are as in Table~\ref{tab:fund_rep}.}\label{tab:real_full_results}\\
\hline
\multirow{2}{*}{ Feature Name } & \multirow{2}{*}{ Mutation Type } & CRT & HRT & \multicolumn{2}{c}{e-CRT} \\
& & p-value & p-value & \ccell{$S_{t_\text{stop}}$} & \ccell{$t_\text{stop}$} \\
\hline
10F & Accessory & 0.005 & 0.001 & 25.5 & 680 \\ 
10I & Other & 0.005 & 0.001 & 31.8 & 520 \\ 
10V & Other & 0.01 & 0.001 & 25.5 & 1280 \\ 
11I & Other & 0.005 & 0.009 & 5.5 & 1555 \\ 
11L & Other & 0.98 & 1 & 1 & 1555 \\ 
12A & Unknown & 0.94 & 1 & 1 & 1555 \\ 
12I & Unknown & 0.716 & 1 & 1 & 1555 \\ 
12K & Unknown & 0.806 & 1 & 0.9 & 1555 \\ 
12N & Unknown & 0.711 & 1 & 1 & 1555 \\ 
12P & Unknown & 0.458 & 1 & 1 & 1555 \\ 
12S & Unknown & 0.204 & 0.556 & 1 & 1555 \\ 
13V & Unknown & 0.005 & 0.001 & 23.8 & 1540 \\ 
14R & Unknown & 0.03 & 0.042 & 1.2 & 1555 \\ 
15V & Unknown & 0.925 & 0.929 & 0.2 & 1555 \\ 
16A & Unknown & 0.005 & 0.001 & 20.7 & 1020 \\ 
16E & Unknown & 0.647 & 1 & 0.8 & 1555 \\ 
18H & Unknown & 0.925 & 1 & 0.5 & 1555 \\ 
19I & Unknown & 0.438 & 1 & 0.4 & 1555 \\ 
19P & Unknown & 0.637 & 1 & 1 & 1555 \\ 
19Q & Unknown & 0.826 & 1 & 1 & 1555 \\ 
19T & Unknown & 0.816 & 1 & 1 & 1555 \\ 
19V & Unknown & 0.776 & 1 & 0.7 & 1555 \\ 
20I & Other & 0.01 & 0.033 & 1 & 1555 \\ 
20M & Other & 1 & 1 & 1 & 1555 \\ 
20R & Other & 0.005 & 0.004 & 26 & 1240 \\ 
20T & Accessory & 0.005 & 0.002 & 40.9 & 1160 \\ 
20V & Other & 0.254 & 1 & 0.8 & 1555 \\ 
22V & Unknown & 0.025 & 0.116 & 5.7 & 1555 \\ 
23I & Accessory & 1 & 0.357 & 0.5 & 1555 \\ 
24F & Accessory & 0.005 & 0.001 & 20.9 & 1000 \\ 
24I & Accessory & 0.005 & 0.002 & 24.9 & 940 \\ 
30N & Major & 0.03 & 0.148 & 3.2 & 1555 \\ 
32I & Major & 0.01 & 0.007 & 29 & 880 \\ 
33F & Accessory & 0.005 & 0.001 & 27.3 & 600 \\ 
33I & Other & 0.045 & 0.096 & 1.9 & 1555 \\ 
33V & Other & 0.622 & 1 & 1 & 1555 \\ 
34D & Other & 0.09 & 1 & 1 & 1555 \\ 
34Q & Other & 0.03 & 0.002 & 27.3 & 1540 \\ 
35D & Other & 0.388 & 0.282 & 0.4 & 1555 \\ 
35G & Other & 0.826 & 0.887 & 1.2 & 1555 \\ 
35N & Other & 0.945 & 1 & 0.9 & 1555 \\ 
35Q & Other & 1 & 1 & 1 & 1555 \\ 
36I & Other & 0.925 & 0.777 & 0.4 & 1555 \\ 
36L & Other & 0.08 & 0.125 & 1.2 & 1555 \\ 
36V & Other & 0.667 & 0.423 & 2.7 & 1555 \\ 
37C & Other & 0.662 & 1 & 1 & 1555 \\ 
37D & Other & 0.015 & 0.1 & 3.2 & 1555 \\ 
37E & Other & 0.398 & 1 & 0.8 & 1555 \\ 
37H & Other & 0.672 & 1 & 1 & 1555 \\ 
37S & Other & 0.672 & 0.481 & 0.4 & 1555 \\ 
37T & Other & 0.776 & 1 & 0.5 & 1555 \\ 
37X & Other & 0.975 & 0.248 & 1 & 1555 \\ 
39Q & Other & 0.985 & 1 & 0.8 & 1555 \\ 
39S & Other & 0.965 & 0.392 & 1 & 1555 \\ 
41K & Other & 0.368 & 1 & 0.4 & 1555 \\ 
43T & Accessory & 0.005 & 0.001 & 21 & 880 \\ 
45R & Unknown & 0.209 & 0.223 & 0.8 & 1555 \\ 
46I & Major & 0.005 & 0.001 & 23.3 & 520 \\ 
46L & Major & 0.005 & 0.001 & 21 & 980 \\ 
46V & Accessory & 0.075 & 1 & 0.8 & 1555 \\ 
47A & Major & 0.005 & 0.028 & 8 & 1555 \\ 
47V & Major & 0.005 & 0.001 & 22.4 & 240 \\ 
48M & Major & 0.005 & 0.067 & 6.5 & 1555 \\ 
48V & Major & 0.005 & 0.003 & 23.6 & 660 \\ 
50L & Major & 0.005 & 0.001 & 20 & 700 \\ 
50V & Major & 0.005 & 0.001 & 21.3 & 440 \\ 
53L & Accessory & 0.025 & 1 & 37 & 1320 \\ 
54A & Major & 0.005 & 0.001 & 29.5 & 760 \\ 
54L & Major & 0.02 & 0.071 & 20.6 & 1160 \\ 
54M & Major & 0.005 & 0.003 & 21.5 & 960 \\ 
54S & Major & 0.095 & 0.363 & 1.8 & 1555 \\ 
54T & Major & 1 & 0.998 & 1 & 1555 \\ 
54V & Major & 0.005 & 0.001 & 33.1 & 180 \\ 
55R & Unknown & 1 & 1 & 0.6 & 1555 \\ 
57G & Unknown & 0.478 & 1 & 0.3 & 1555 \\ 
57K & Unknown & 0.1 & 0.07 & 0.9 & 1555 \\ 
58E & Accessory & 0.02 & 0.002 & 1.9 & 1555 \\ 
60E & Unknown & 0.577 & 1 & 0.5 & 1555 \\ 
61E & Unknown & 0.025 & 0.352 & 0.9 & 1555 \\ 
61H & Unknown & 0.756 & 0.881 & 0.5 & 1555 \\ 
61N & Unknown & 0.637 & 0.642 & 1.1 & 1555 \\ 
62V & Unknown & 0.129 & 0.057 & 1.1 & 1555 \\ 
63A & Unknown & 0.174 & 1 & 0.8 & 1555 \\ 
63C & Unknown & 0.925 & 1 & 1 & 1555 \\ 
63H & Unknown & 0.915 & 1 & 1 & 1555 \\ 
63P & Unknown & 0.005 & 0.072 & 28.4 & 160 \\ 
63Q & Unknown & 0.642 & 1 & 0.8 & 1555 \\ 
63S & Unknown & 0.701 & 1 & 1.3 & 1555 \\ 
63T & Unknown & 0.328 & 1 & 0.4 & 1555 \\ 
63V & Unknown & 0.826 & 1 & 1 & 1555 \\ 
63X & Unknown & 0.766 & 1 & 1 & 1555 \\ 
64L & Unknown & 0.005 & 0.011 & 28.6 & 1300 \\ 
64M & Unknown & 0.214 & 0.706 & 0.4 & 1555 \\ 
64V & Unknown & 0.95 & 0.442 & 0.5 & 1555 \\ 
65D & Unknown & 0.726 & 0.642 & 0.6 & 1555 \\ 
66F & Unknown & 0.01 & 0.144 & 3.1 & 1555 \\ 
66V & Unknown & 0.652 & 1 & 0.9 & 1555 \\ 
67E & Unknown & 0.91 & 1 & 0.9 & 1555 \\ 
67F & Unknown & 0.159 & 0.557 & 1.1 & 1555 \\ 
69Q & Unknown & 0.811 & 1 & 0.7 & 1555 \\ 
69R & Unknown & 0.279 & 1 & 0.8 & 1555 \\ 
69Y & Unknown & 1 & 1 & 0.8 & 1555 \\ 
70E & Unknown & 0.736 & 1 & 1 & 1555 \\ 
70R & Unknown & 0.085 & 0.061 & 1.4 & 1555 \\ 
70T & Unknown & 0.846 & 1 & 1 & 1555 \\ 
71I & Other & 0.488 & 0.546 & 1.6 & 1555 \\ 
71L & Other & 0.114 & 1 & 1 & 1555 \\ 
71T & Other & 0.03 & 0.097 & 0.8 & 1555 \\ 
71V & Other & 0.03 & 0.356 & 7.3 & 1555 \\ 
72E & Unknown & 0.716 & 1 & 0.6 & 1555 \\ 
72M & Unknown & 0.055 & 1 & 5.2 & 1555 \\ 
72R & Unknown & 0.756 & 1 & 0.8 & 1555 \\ 
72T & Unknown & 0.07 & 0.034 & 1 & 1555 \\ 
72V & Unknown & 0.612 & 0.699 & 0.6 & 1555 \\ 
73A & Accessory & 0.721 & 0.721 & 0.6 & 1555 \\ 
73C & Accessory & 0.005 & 0.004 & 20.3 & 1080 \\ 
73S & Accessory & 0.005 & 0.001 & 22.4 & 1260 \\ 
73T & Accessory & 0.01 & 1 & 3.4 & 1555 \\ 
74A & Unknown & 0.756 & 1 & 1 & 1555 \\ 
74S & Other & 0.284 & 0.404 & 2.5 & 1555 \\ 
76V & Major & 0.005 & 0.001 & 33.1 & 400 \\ 
77I & Unknown & 0.005 & 0.003 & 6.3 & 1555 \\ 
79A & Unknown & 0.806 & 0.341 & 1 & 1555 \\ 
79S & Unknown & 0.726 & 1 & 0.9 & 1555 \\ 
82A & Major & 0.005 & 0.001 & 21.1 & 480 \\ 
82C & Major & 0.005 & 0.001 & 9.2 & 1555 \\ 
82F & Major & 0.005 & 0.012 & 21.4 & 800 \\ 
82I & Other & 1 & 0.941 & 0.7 & 1555 \\ 
82L & Major & 0.836 & 1 & 0.6 & 1555 \\ 
82M & Major & 0.94 & 1 & 0.5 & 1555 \\ 
82S & Major & 0.229 & 0.544 & 1.3 & 1555 \\ 
82T & Major & 0.005 & 0.003 & 25.4 & 880 \\ 
83D & Accessory & 0.075 & 1 & 0.7 & 1555 \\ 
84A & Major & 0.169 & 0.328 & 1.9 & 1555 \\ 
84C & Major & 0.746 & 0.683 & 0.9 & 1555 \\ 
84V & Major & 0.005 & 0.001 & 35.4 & 220 \\ 
85V & Other & 0.866 & 0.874 & 0.4 & 1555 \\ 
88D & Accessory & 0.005 & 0.059 & 20.1 & 1300 \\ 
88G & Major & 0.657 & 1 & 1 & 1555 \\ 
88S & Major & 1 & 0.172 & 0.7 & 1555 \\ 
88T & Major & 0.975 & 1 & 1 & 1555 \\ 
89I & Unknown & 0.005 & 0.044 & 1.4 & 1555 \\ 
89M & Unknown & 0.328 & 0.669 & 0.3 & 1555 \\ 
89V & Accessory & 0.01 & 0.026 & 20.6 & 600 \\ 
90M & Major & 0.005 & 0.001 & 35.7 & 780 \\ 
91S & Unknown & 0.189 & 1 & 1.1 & 1555 \\ 
92K & Unknown & 0.289 & 0.555 & 0.6 & 1555 \\ 
92R & Unknown & 0.736 & 0.62 & 0.7 & 1555 \\ 
93L & Unknown & 0.005 & 0.086 & 3.2 & 1555 \\ 
95F & Unknown & 0.408 & 0.514 & 0.9 & 1555 \\ 
\hline
\end{longtable}

\begin{figure}[t]
     \centering
     \begin{subfigure}[b]{0.329\textwidth}
         \centering
         \includegraphics[width=\textwidth]{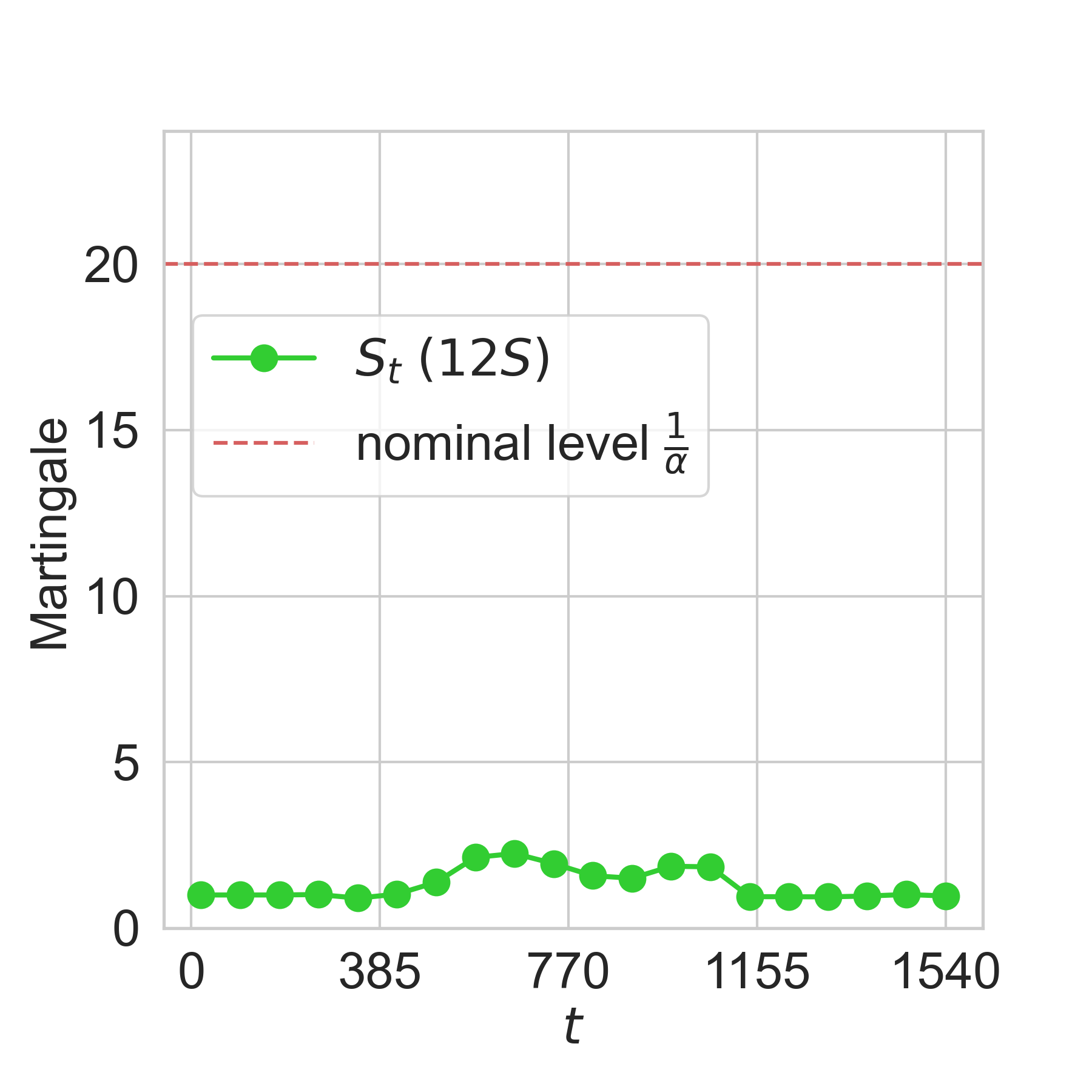}
         \caption{}
         \label{fig:real_data_ecrt_12s}
     \end{subfigure}
     \hfill
     \begin{subfigure}[b]{0.329\textwidth}
         \centering
         \includegraphics[width=\textwidth]{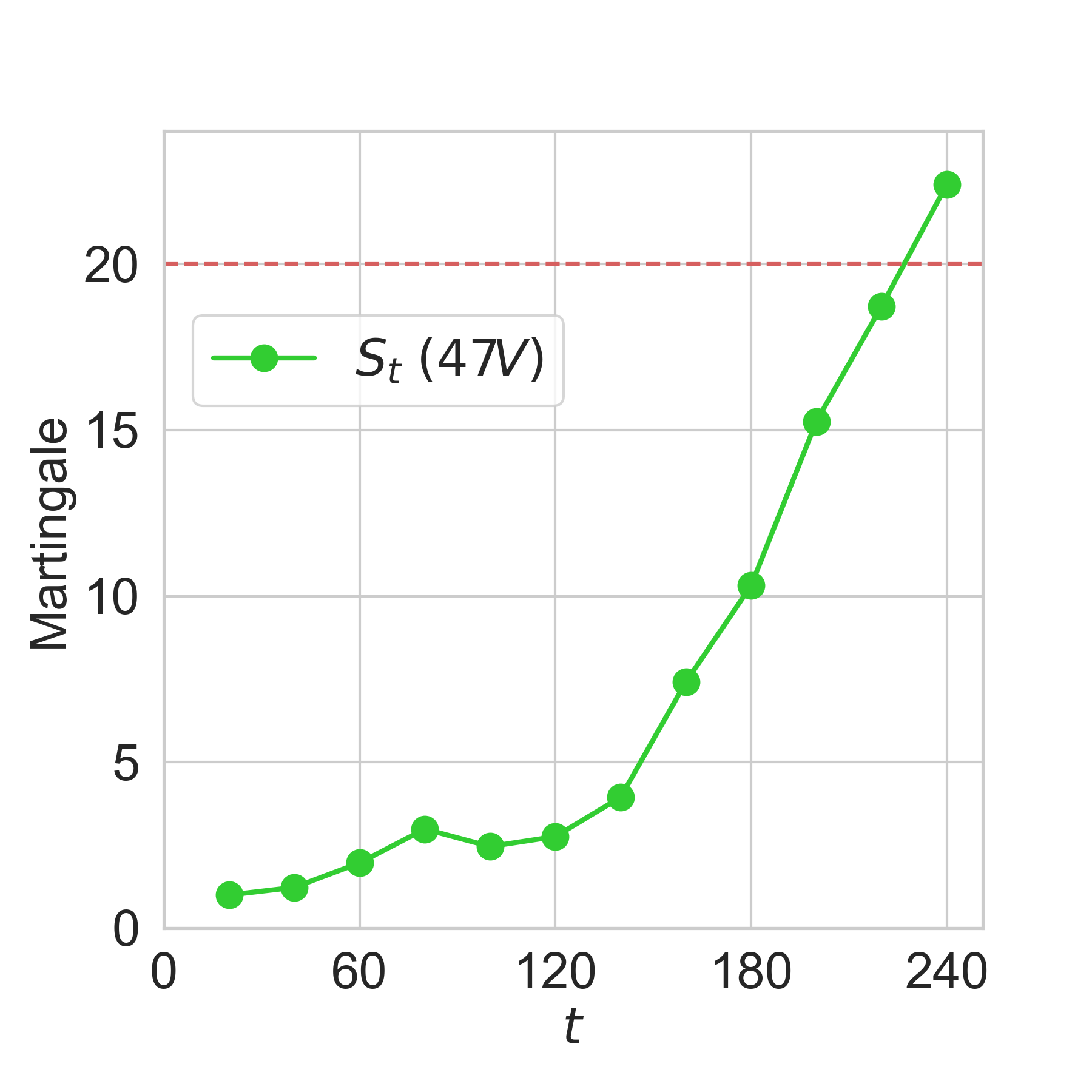}
         \caption{}
         \label{fig:real_data_ecrt_47v}
     \end{subfigure}
     \hfill
     \begin{subfigure}[b]{0.329\textwidth}
         \centering
         \includegraphics[width=\textwidth]{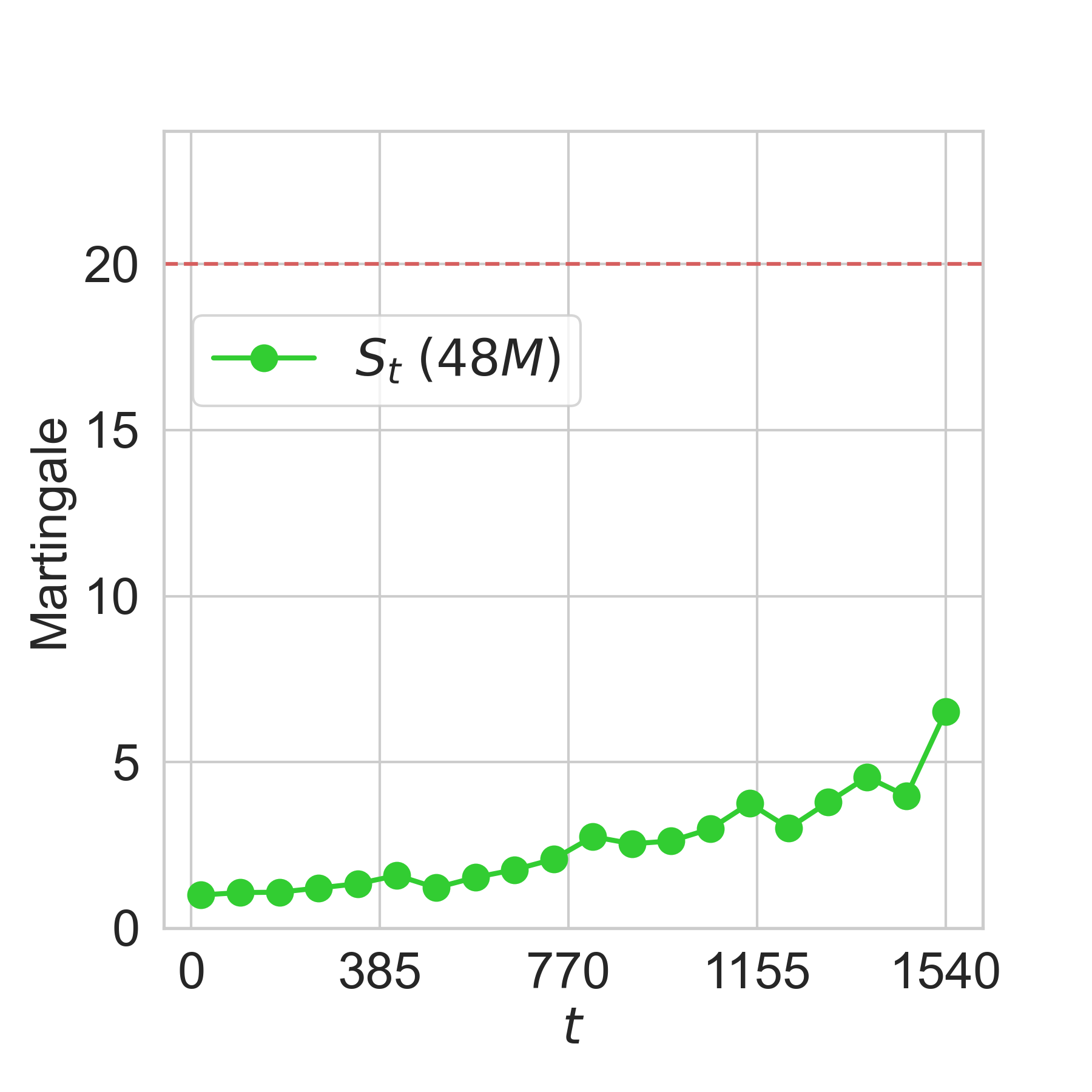}
         \caption{}
         \label{fig:real_data_ecrt_48m}
     \end{subfigure}
     
        \caption{\textbf{Real HIV data experiment.} The test martingale $S_t$ as a function of $t$, evaluated on three representative mutations of HIV. (a) Mutation `12S', which has not been reported in previous studies to have an effect on drug resistance. (b) Mutation `47V', which has been reported to have a major effect. (c) Mutation `48M', which has been reported to have a major effect.}
        \label{fig:real_data_martingales}
\end{figure}

\end{appendix}

\end{document}